\documentclass[]{article}
\usepackage{graphicx}
\usepackage{epstopdf}
\usepackage{amsfonts}
\usepackage{amsmath}
\usepackage{amssymb}
\usepackage{fancyhdr}
\usepackage{titlesec}
\usepackage{indentfirst}
\usepackage{booktabs}
\usepackage{verbatim}
\usepackage{color}
\usepackage{amsthm}
\usepackage{subfigure}

\usepackage[page,header]{appendix}
\usepackage{titletoc}

\numberwithin{equation}{section}
\newtheorem{theorem}{Theorem}[section]
\newtheorem{lemma}[theorem]{Lemma}
\newtheorem{corollary}[theorem]{Corollary}
\newtheorem{proposition}[theorem]{Proposition}

\newtheorem{remark}[theorem]{Remark}

\newcommand{\rd}{\mathrm{d}}
\newcommand{\q}{\hat{q}}
\newcommand{\aq}{|q|}
\newcommand{\h}{\hat{h}}

\newcommand{\jj}{\hat{j}}
\newcommand{\F}{\tilde{F}}
\newcommand{\im}{\mathrm{i}}
\newcommand{\e}{\mathrm{e}}
\newcommand{\mQ}{\mathcal{Q}}
\newcommand{\mP}{\mathcal{P}}

\begin{document}

\title{A fast Fourier spectral method for the homogeneous Boltzmann equation with non-cutoff collision kernels}

\author{
	Jingwei Hu\footnote{Department of Mathematics, Purdue University, West Lafayette, IN 47907, USA (jingweihu@purdue.edu). JH's research was partially supported by NSF grant DMS-1620250 and NSF CAREER grant DMS-1654152.} \  \ and \
	Kunlun Qi\footnote{Department of Mathematics, City University of Hong Kong, Hong Kong, China (kunlun.qi@my.cityu.edu.hk).}
	} 
\maketitle

\begin{abstract}
We introduce a fast Fourier spectral method for the spatially homogeneous Boltzmann equation with non-cutoff collision kernels. Such kernels contain non-integrable singularity in the deviation angle which arise in a wide range of interaction potentials (e.g., the inverse power law potentials). Albeit more physical, the non-cutoff kernels bring a lot of difficulties in both analysis and numerics, hence are often cut off in most studies (the well-known Grad's angular cutoff assumption). We demonstrate that the general framework of the fast Fourier spectral method developed in \cite{GHHH17, HM19} can be extended to handle the non-cutoff kernels, achieving the accuracy/efficiency comparable to the cutoff case. We also show through several numerical examples that the solution to the non-cutoff Boltzmann equation enjoys the smoothing effect, a striking property absent in the cutoff case.
\end{abstract}

{\small 


}



\section{Introduction}
\label{sec:intro}

The Boltzmann equation, proposed by Maxwell and Boltzmann, is one of the fundamental equations in kinetic theory and models the fluid flow behavior at a wide range of physical conditions \cite{CC, Cercignani, Villani02}. Generally speaking, when the mean free path of the system is comparable to the characteristic length of the problem, the Navier-Stokes based macroscopic description would break down and one has to resort to the mesoscopic kinetic description. This situation often occurs when the mean free path is large (e.g., in design of spacecrafts in outer atmosphere where the air is rarefied), or when the characteristic length is small (e.g., in modeling of microsystems where the devices are small), bespeaking the wide applicability of the Boltzmann equation in various science and engineering disciplines.

The Boltzmann equation reads
\begin{equation}
\partial_{t} f+v\cdot \nabla_x f=\mQ(f,f), \quad t>0, \ x\in \Omega\subset\mathbb{R}^3, \ v\in \mathbb{R}^3,
\end{equation}
where $f = f(t,x,v)$ is the probability density function of time $t$, position $x$, and velocity $v$, and $\mQ(f,f)$ is the so-called Boltzmann collision operator describing the binary collisions among particles:
\begin{equation} \label{CO1}
\mQ(f,f)(v)=\int_{\mathbb{R}^3}\int_{S^{2}}\mathcal{B}(v-v_*,\sigma)\left[f(v')f(v_*')-f(v)f(v_*)\right]\,\rd{\sigma}\,\rd{v_*}.
\end{equation}
In the formula above, $t$ and $x$ are suppressed since $\mQ(f,f)$ acts on $f$ only through the velocity. $(v',v_*')$ and $(v,v_*)$ represent the velocity pairs before and after a collision, which satisfy the conservation of momentum and energy:
\begin{equation} 
v' + v_{*}' = v + v_{*}, \quad  |v'|^{2} + |v_{*}'|^{2} = |v|^{2} + |v_{*}|^{2},
\end{equation}
so that $(v',v_*')$ can be expressed in terms of $(v,v_*)$ as
\begin{equation}  
v'=\frac{v+v_*}{2}+\frac{|v-v_*|}{2}\sigma, \quad  v_*'=\frac{v+v_*}{2}-\frac{|v-v_*|}{2}\sigma,
\end{equation}
where $\sigma$ is a vector varying over the unit sphere $S^{2}$. 

It can be shown that $\mQ(f,f)$ satisfies the conservation of mass, momentum, and energy:
\begin{equation}
\int_{\mathbb{R}^3} \mQ(f,f) \,\rd v=\int_{\mathbb{R}^3} \mQ(f,f) v\,\rd v=\int_{\mathbb{R}^3} \mQ(f,f) |v|^2\,\rd v=0,
\end{equation}
and the celebrated Boltzmann's H-theorem:
\begin{equation}
\int_{\mathbb{R}^3} \mQ(f,f) \ln f\, \rd v \leq 0,
\end{equation}
with equality holds if and only if $ f $ reaches the equilibrium:
\begin{equation}
M(v) = \frac{\rho}{\left(2\pi T\right)^{\frac{3}{2}}} \e^{-\frac{|v-u|^{2}}{2T}},
\end{equation}
where the density $ \rho $, bulk velocity $ u $, and temperature $ T $ are given by
\begin{equation}
\rho = \int_{\mathbb{R}^3} f \,\rd v, \quad u=\frac{1}{\rho}\int_{\mathbb{R}^3} fv\, \rd v, \quad T=\frac{1}{3\rho} \int_{\mathbb{R}^3} f |v-u|^{2}\, \rd v.
\end{equation}

In (\ref{CO1}), the collision kernel $\mathcal{B}$ is a non-negative function depending only on $|v-v_*|$ and cosine of the deviation angle $\theta$ (angle between $v-v_*$ and $v'-v_*'$). Thus $\mathcal{B}$ is often written as
\begin{equation} 
\mathcal{B}(v-v_*, \sigma)=B(|v-v_*|,\cos \theta), \quad \cos\theta=\frac{\sigma\cdot (v-v_*)}{|v-v_*|}.
\end{equation}
The specific form of $B$ can be determined from the intermolecular potential using classical scattering theory \cite{Cercignani}, yet its explicit form is not known except for some simple potentials. For example, in the case of inverse power law potentials $U(r) = r^{-(s-1)}, 2< s < \infty$, where $r$ is the distance between two interacting particles, it can be shown that the angular part and velocity part of $B$ are separated:
\begin{equation}\label{Bb}
B(|v-v_*|,\cos\theta) = b(\cos\theta) \Phi(|v-v_*|),
\end{equation} 
where $\Phi(|v-v_*|)=|v-v_*|^{\gamma}$, $\gamma = \frac{s-5}{s-1}$, $-3< \gamma < 1$, and $b(\cos\theta)$ is some function defined implicitly. Using simple asymptotic expansion, one can show that $b(\cos\theta)$ when $\theta\rightarrow 0$ behaves as
\begin{equation} \label{Bb1}
\sin \theta b(\cos\theta) \Big|_{\theta\rightarrow 0} \sim K \theta^{-1-\nu}, \quad \nu = \frac{2}{s-1}, \quad 0 < \nu<2,
\end{equation}
i.e., it has a \textit{non-integrable singularity} when the deviation angle is small. The kernel (\ref{Bb}) encompasses a wide range of potentials, and we just mention two marginal cases: $s=\infty$, $\gamma=1$, $\nu=0$ corresponds to the hard spheres, and $s=2$, $\gamma=-3$, $\nu=2$ corresponds to the Coulomb interaction (in fact, the Boltzmann collision operator loses the validity in this case and one has to use its grazing limit, the so-called Landau operator \cite{Villani02}, a diffusive type operator). Finally, we point out that the non-integrable singularity in the collision kernel is a generic phenomenon when the interaction is long range and does not only appear in the inverse power potentials. For instance, the Debye-Yukawa potential $U(r) = r^{-1}\e^{-r}$ also leads to an angular singularity as $\theta\rightarrow 0$ (c.f. \cite{morimoto2009regularity}):
\begin{equation}\label{debye}
\sin \theta B(|v-v_*|,\cos \theta) \Big|_{\theta\rightarrow 0}  \sim K |v-v_{*}|\theta^{-1}|\log \theta^{-1}|.
\end{equation} 

Albeit more physical, the non-integrable singularity in the collision kernel brings a lot of difficulties in both theoretical and numerical treatment of the Boltzmann equation. Due to this, Grad \cite{grad1958principles} introduced the famous \textit{angular cutoff assumption}, replacing the collision kernel by a locally integrable one, and it is henceforth used in the majority of works on the Boltzmann equation. The Grad's cutoff assumption greatly simplifies the analysis, but also changes the qualitative behavior of the solutions. Since the work of Desvillettes \cite{desvillettes1995regularizing}, it has been realized that the solution to the non-cutoff Boltzmann equation enjoys the {\it smoothing effect}, which is not true in the cutoff case where the solution can be at best as regular as the initial data. Without going into technical detail, we quote the following statement from \cite{alexandre2000entropy} to help readers better understand the structure of the problem: {\it ``The non-cutoff Boltzmann operator $\mathcal{Q}(f,\cdot)$ behaves like the fractional Laplacian $-(-\Delta)^{\frac{\nu}{2}}$. In the limit case $\nu=2$, it has to be replaced by the Landau operator, which is precisely diffusive in nature."} There are by now a large number of theoretical results related to the non-cutoff Boltzmann equation. We refer to the recent review by Alexandre \cite{alexandre2009review} for further references.

{\bf Our contribution.} Numerical approximation of the Boltzmann equation is also largely influenced by the Grad's cutoff assumption. This includes both the direct simulation Monte Carlo method \cite{Bird} and deterministic methods such as the Fourier spectral method \cite{PR00, MP06, GT09, GHHH17}. Therefore, it is our goal of this work to introduce a reliable numerical method to solve the more physical non-cutoff Boltzmann equation. We will show that the general framework of the fast Fourier spectral method developed in \cite{GHHH17, HM19} can be extended to handle the non-cutoff kernels, achieving the accuracy/efficiency comparable to the cutoff case. In particular, we will carefully compare the solutions computed with and without cut-off assumptions, and verify the regularizing effect as predicted by the theory. 

{\bf Related work.} There are some existing numerical work related to the non-cutoff Boltzmann equation. \cite{PTV03, GH14} are the closest to ours, where the authors considered the grazing collision limit of the Fourier spectral method for the Boltzmann operator and showed that it reduces to the Fourier spectral method for the limiting Landau operator. We mention that the Taylor expansion has been used in \cite{PTV03, GH14} to study the integrability of the kernel and our approach in Section~\ref{subsec:intergability} shares a similar spirit. Yet, both works focused on the transition from the Boltzmann to the Landau equation (i.e., $\nu\rightarrow 2$ in (\ref{Bb1})) and no fast algorithm was introduced. The recent work \cite{GJ18, YZ20} indeed considered the non-cutoff Boltzmann equation (i.e., $0<\nu< 2$ in (\ref{Bb1})): the former solved a radially symmetric version with Maxwell molecules using symbolic calculation, and the latter proposed a modified equation by adding a scaled Landau operator to account for the singularity.

The rest of this paper is organized as follows: Section~\ref{sec:method} describes the basic formulation of the Fourier spectral method for the non-cutoff Boltzmann equation, where the focus is to prove the integrability of the weight in the method. In Section~\ref{sec:consistency} we establish the consistency and spectral accuracy of the method in approximating the collision operator. In Section~\ref{sec:fast} we introduce a fast algorithm to accelerate the method and discuss some implementation detail. Numerical examples are presented in Section~\ref{sec:num} to demonstrate the efficiency and accuracy of the proposed method. The paper is concluded in Section~\ref{sec:conc}.


\section{A Fourier-Galerkin spectral method for the non-cutoff Boltzmann equation}
\label{sec:method}

In this section, we describe the Fourier-Galerkin spectral method for solving the non-cutoff Boltzmann equation. Since the main difficulty comes from the collision operator, for the rest of this paper we will consider the following Cauchy problem of the spatially homogeneous Boltzmann equation:
\begin{equation}\label{eqn}
\left\{
\begin{split}
&\partial_{t} f(t,v) = \mQ(f,f), \quad t>0, \ v\in \mathbb{R}^d, \ d=2 \text{ or } 3,  \\
& f(0,v) = f^{0}(v),
\end{split}
\right.
\end{equation}
where the collision operator is rewritten here for clarity (we include the 2D model as well since it leads to some numerical simplicity):
\begin{equation} \label{CO}
\mQ(f,f)(v) = \int_{\mathbb{R}^d}\int_{S^{d-1}}B(|q|,\sigma\cdot\q) \left[f(v_*')f(v')-f(v-q)f(v)\right]\,\rd{\sigma}\,\rd{q},
\end{equation}
with
\begin{equation}
v'=v-\frac{1}{2}(q-|q|\sigma), \quad v_*'=v-\frac{1}{2}(q+|q|\sigma).
\end{equation}
Note that compared to the original operator (\ref{CO1}), we have done a change of variables: $v_* \rightarrow q=v-v_*$ in the above formula, and $|q|$, $\q = q/|q|$ denote the magnitude and direction of $q$, respectively. 

In order to apply the Fourier spectral method, we first need a proper truncation of the domain and integral. To this end, we assume that $f$ has a compact support in $v$: $\text{Supp}(f(v))\subset \mathcal{B}_S$, where $\mathcal{B}_S$ is a ball centered at the origin with radius $S$ (in practice, $S$ can be chosen roughly as $\max |u^0\pm c\sqrt{T^0}|$, where $u^0$ and $T^0$ are the bulk velocity and temperature corresponding to the initial data $f^0$, and $c$ is some constant $\sim 3$). It then suffices to truncate the infinite integral in $q$ to a larger ball $\mathcal{B}_R$ with radius $R\geq2S$. It is also easy to see $\text{Supp}(\mQ(f,f)(v))\subset \mathcal{B}_{\sqrt{2}S}$. Hence we can restrict $v$ to the computational domain $\mathcal{D}_L=[-L,L]^d$ with $L\geq \sqrt{2}S$ and extend the solution periodically to the whole space (in practice, $L$ can be chosen as $L\geq (3+\sqrt{2})S/2$ to avoid aliasing effect \cite{PR00}).

With the above assumptions, we consider the following truncated problem as an approximation to the original problem (\ref{eqn}):
\begin{equation}\label{eqnR}
\left\{
\begin{split}
&\partial_{t} f(t,v) = \mQ^R(f,f), \quad t>0, \ v\in \mathcal{D}_L, \ d=2 \text{ or } 3, \\
& f(0,v) = f^{0}(v),
\end{split}
\right.
\end{equation}
with
\begin{equation} 
\mQ^R(f,f)(v) = \int_{\mathcal{B}_R}\int_{S^{d-1}}B(|q|,\sigma\cdot\q) \left[f(v_*')f(v')-f(v-q)f(v)\right]\,\rd{\sigma}\,\rd{q},
\end{equation}
and its (truncated) weak form
\begin{equation}\label{weak}
\int_{D_L} \mQ^R(f,f)(v) \phi(v)\, \rd v = \int_{D_L} \int_{\mathcal{B}_R}  \int_{S^{d-1}}B(|q|,\sigma\cdot\q) f(v-q)f(v)[\phi(v')-\phi(v)]\, \rd\sigma\, \rd q\, \rd v,
\end{equation}
where $\phi(v)$ is some test function.

We now construct the Fourier-Galerkin spectral method for (\ref{eqnR}). Consider the space of trigonometric polynomials of degree up to $N/2$:
\begin{equation}
\mathbb{P}_N=\text{span} \left\{ \e^{\im \frac{\pi}{L} k\cdot v}\Big| -\frac{N}{2}\leq k \leq \frac{N}{2} \right\},\footnote{$k=(k_1,\dots,k_d). -N/2\leq k \leq N/2$ means $-N/2\leq k_j \leq N/2$, $j=1,\dots,d.$}
\end{equation} 
equipped with inner product
\begin{equation}
\langle f, g \rangle=\frac{1}{(2L)^d}\int_{\mathcal{D}_L} f\bar{g}\,\rd{v}.
\end{equation}
The method seeks a solution $f_N\in \mathbb{P}_N$ such that
\begin{equation}  
f_N(t,v)=\sum_{k={-\frac{N}{2}}}^{\frac{N}{2}}f_k(t) \e^{\im \frac{\pi}{L}k\cdot v},\footnote{$\sum_{k={-N/2}}^{N/2}:=\sum_{k_1={-N/2}}^{N/2}\cdots\sum_{k_d=-N/2}^{N/2}$.}
\end{equation}
and requires
\begin{equation}
\langle \partial_t f_N-\mathcal{Q}^R(f_N,f_N), \e^{\im \frac{\pi}{L}k\cdot v}\rangle=0, \quad  \text{for }-\frac{N}{2}\leq k \leq \frac{N}{2}.
\end{equation}

The resulting Galerkin system reads
\begin{equation} 
\left\{
\begin{split}
&\frac{\rd}{\rd t} f_k = \mQ_k^R, \quad -\frac{N}{2}\leq k \leq \frac{N}{2}, \\
&f_k(0) = f^0_k,
\end{split}
\right.
\end{equation}
with
\begin{equation}
\mQ_{k}^R: =\langle \mQ^R(f_{N},f_{N}),\e^{\im \frac{\pi}{L}k\cdot v}  \rangle
, \quad f^0_k:=\langle f^0,\e^{\im \frac{\pi}{L}k\cdot v}\rangle.
\end{equation}
Using the weak form (\ref{weak}), one can derive that
\begin{equation}
\begin{split} \label{sum}
\mQ_{k}^R&=\frac{1}{(2L)^{d}}\int_{D_{L}}\int_{\mathcal{B}_{R}}\int_{S^{d-1}}B(|q|,\sigma\cdot\q)f_{N}(v)f_{N}(v-q)(\e^{-\im \frac{\pi}{L}k\cdot v'}-\e^{-\im \frac{\pi}{L}k\cdot v})\, \rd\sigma\, \rd q \,\rd v\\
&=\sum\limits_{\substack{l,m=-\frac{N}{2}\\l+m=k}}^{\frac{N}{2}}G(l,m)f_{l}f_{m},
\end{split}
\end{equation}
where the weight $G(l,m)$ is given by
\begin{equation}\label{weight}
G(l,m) = \int_{\mathcal{B}_{R}}\e^{-\im \frac{\pi}{L}m\cdot q}\left[\int_{S^{d-1}}B(|q|,\sigma\cdot\q)(\e^{\im \frac{\pi}{2L}(l+m)\cdot (q-|q|\sigma)}-1)\,\rd\sigma\right]\,\rd q.
\end{equation}

We mention that up to this point, the derivation of the spectral method is completely formal and the singularity of the collision kernel does not play a role. In fact, if $B(|q|,\sigma\cdot\q)$ is integrable, the derivation is done and one can proceed straightforwardly to the implementation: precompute the weight $G(l,m)$ according to the formula (\ref{weight}) up to certain accuracy as it does not depend on $f$, and evaluate the sum (\ref{sum}) directly to get $\mQ_{k}^R$ at every time step. However, in the non-cutoff case $B(|q|,\sigma\cdot\q)$ has a non-integrable singularity, hence nothing guarantees the weight defined in (\ref{weight}) is well defined. Then whether the Fourier spectral method is a suitable approximation deserves further investigation. Fortunately, we will show below that the answer is positive.

\subsection{Integrability of the weight}
\label{subsec:intergability}

In this section, we show that the weight (\ref{weight}) is well defined using a simple Taylor expansion. Similar approach has been used in \cite{PTV03, GH14} to study the grazing limit of the Boltzmann equation, that is, when all collisions are concentrated around $\theta \sim 0$.

For simplicity, we assume the kernel has the form
\begin{equation} \label{sep}
B(|q|,\sigma\cdot\q)=\Phi(|q|)b(\sigma\cdot \q),
\end{equation}
and rewrite (\ref{weight}) as follows
\begin{equation} \label{GG}
\begin{split}
G(l,m) &=\int_0^R \int_{S^{d-1}} \Phi(|q|)  |q|^{d-1}\e^{-\im \frac{\pi}{L} |q| m\cdot \hat{q}} \left[\int_{S^{d-1}}b(\sigma\cdot\q)(\e^{\im \frac{\pi}{2L}|q| (l+m)\cdot (\hat{q}-\sigma)}-1)\,\rd\sigma\right]\,\rd{\hat{q}}\,\rd{|q|}\\
&=\int_0^R \int_{S^{d-1}}  \Phi(|q|)  |q|^{d-1}\e^{-\im \frac{\pi}{L}|q| m\cdot \hat{q}} F(l+m,|q|,\hat{q})\,\rd{\hat{q}}\,\rd{|q|},
\end{split}
\end{equation}
where
\begin{equation} \label{FF}
F(k,|q|,\hat{q}):=\int_{S^{d-1}}b(\sigma\cdot\q)(\e^{\im \frac{\pi}{2L}|q| k\cdot (\hat{q}-\sigma)}-1)\,\rd\sigma.
\end{equation}

We have the following result.
\begin{proposition}
Let $B(|v-v_*|,\cos \theta) $ be the collision kernel of the Boltzmann equation whose angular part $b(\cos\theta)$ satisfies the singularity condition:
\begin{equation} \label{bb}
\sin^{d-2}\theta b(\cos\theta) \Big|_{\theta\rightarrow 0} \sim K  \theta^{-1-\nu}, \quad 0 \leq \nu<2, \quad d=2 \text{ or } 3,
\end{equation}
then the weight $F(k,|q|,\hat{q})$ in (\ref{FF}) is well defined. 
\end{proposition}

\begin{proof}
We discuss the 2D and 3D cases separately.

\vspace{0.1in}
{\bf (i) 2D case}: (\ref{FF}) becomes
\begin{equation}
\begin{split} \label{FF2D1}
F(k,\aq,\q)&=\int_{S^{1}} b(\sigma\cdot \hat{q})\left(\e^{\im\frac{\pi}{2L}\aq k\cdot(\q-\sigma)}-1\right)\,\rd\sigma\\
&=\int_{0}^{2\pi}b(\cos\theta)\left(\e^{\im\frac{\pi}{2L}\aq k\cdot[\q(1-\cos\theta)-\q_{\perp}\sin\theta]}-1\right) \,\rd\theta\\
&=\int_{0}^{2\pi}b(\cos\theta)\left(\e^{\F(\theta)}-1\right) \,\rd\theta,
\end{split}
\end{equation}
where we parametrized $\sigma$ in a coordinate system determined by $(\q$, $\q_{\perp})$:
\begin{equation}
\sigma = \q\cos\theta+ \q_{\perp}\sin\theta,
\end{equation}
and
\begin{equation}
\F(\theta):=\im\frac{\pi}{2L}\aq k\cdot[\q (1-\cos\theta)-\q_{\perp}\sin\theta].
\end{equation}

Apparently, for fixed $k$, $|q|$ and $\hat{q}$, the integrand in (\ref{FF2D1}) has a singularity at $\theta=0$ and $2\pi$. Therefore, it suffices to consider the following integral:
\begin{equation}
\begin{split}
F_{\epsilon}&=\int_{0}^{\epsilon}b(\cos\theta)\left(\e^{\F(\theta)}-1\right) \,\rd\theta+\int_{2\pi-\epsilon}^{2\pi}b(\cos\theta)\left(\e^{\F(\theta)}-1\right) \,\rd\theta\\
&=\int_{0}^{\epsilon}b(\cos\theta)\left(\tilde{F}(\theta)+O(\F^{2}(\theta))\right) \,\rd\theta+\int_{2\pi-\epsilon}^{2\pi}b(\cos\theta)\left(\tilde{F}(\theta)+O(\F^{2}(\theta))\right) \,\rd\theta,
\end{split}
\end{equation}
where a Taylor expansion is applied to the exponential function.

For the first order terms in $F_{\epsilon}$, we have
\begin{equation}
\begin{split} \label{Feps}
&\int_{0}^{\epsilon}b(\cos\theta)\tilde{F}(\theta) \,\rd\theta +\int_{2\pi-\epsilon}^{2\pi}b(\cos\theta)\tilde{F}(\theta) \,\rd\theta\\
=&\int_{0}^{\epsilon}\im\frac{\pi}{2L}\aq (k\cdot \q) b(\cos\theta) (1-\cos\theta)\, \rd\theta+\int_{2\pi-\epsilon}^{2\pi}\im\frac{\pi}{2L}\aq (k\cdot \q) b(\cos\theta)(1-\cos\theta)\, \rd\theta,
\end{split}
\end{equation}
where the terms involving $\q_{\perp}$ cancel due to parity. Using $1-\cos\theta=2\sin^{2}(\theta/2)\big|_{\theta\rightarrow 0}\sim \theta^2$ and (\ref{bb}) (with $d=2$), we have
\begin{equation}
b(\cos\theta)(1-\cos\theta)\Big |_{\theta\rightarrow 0}\sim  K \theta^{1-\nu}, \quad 0 \leq \nu<2.
\end{equation}
Similarly, 
\begin{equation}
b(\cos\theta)(1-\cos\theta)\Big |_{\theta\rightarrow 2\pi}\sim  K (2\pi-\theta)^{1-\nu}, \quad 0 \leq \nu<2.
\end{equation}
Hence the integrals in (\ref{Feps}) are integrable.

For the second order terms in $F_{\epsilon}$, it is easy to see $O(\F^{2}(\theta))\big|_{\theta\rightarrow 0}\sim \theta^2$. Using again (\ref{bb}), we have
\begin{equation}
b(\cos\theta)O(\F^{2}(\theta))\Big |_{\theta\rightarrow 0} \sim K \theta^{1-\nu}, \quad 0 \leq \nu<2.
\end{equation}
Similarly,
\begin{equation}
b(\cos\theta)O(\F^{2}(\theta))\Big |_{\theta\rightarrow 2\pi} \sim K (2\pi-\theta)^{1-\nu}, \quad 0 \leq \nu<2.
\end{equation}
Hence these terms are also integrable.

To summarize, we have shown that the integral $F_{\epsilon}$ converges.

\vspace{0.1in}
{\bf (ii) 3D case}: (\ref{FF}) becomes
\begin{equation} \label{FF3D1}
\begin{split}
F(k,\aq,\q)&=\int_{S^{2}} b(\sigma\cdot \hat{q})\left(\e^{\im\frac{\pi}{2L}\aq k\cdot(\q-\sigma)}-1\right)\,\rd\sigma\\
&=\int_{0}^{2\pi}\int_{0}^{\pi}b(\cos\theta)\left(\e^{\im\frac{\pi}{2L}\aq k\cdot[\q(1-\cos\theta)-\h\sin\theta\cos\phi-\jj\sin\theta\sin\phi]}-1\right) \sin\theta\, \rd\theta \,\rd\phi\\
&=\int_{0}^{2\pi}\int_{0}^{\pi}b(\cos\theta)\left(\e^{\F(\theta)}-1\right) \sin\theta\, \rd\theta \,\rd\phi,
\end{split}
\end{equation}
where we parametrized $\sigma$ in a coordinate system determined by ($\h$, $\jj$, $\q$):
\begin{equation}
\sigma=\h\sin\theta\cos\phi+\jj\sin\theta\sin\phi+\q\cos\theta,
\end{equation}
and
\begin{equation} 
\F(\theta)=\im\frac{\pi}{2L}\aq k\cdot[\q(1-\cos\theta)-\h\sin\theta\cos\phi-\jj\sin\theta\sin\phi].
\end{equation}

Now for fixed $k$, $|q|$ and $\hat{q}$, the integrand in (\ref{FF3D1}) has a singularity at $\theta =0$. Therefore, it suffices to consider the following integral:
\begin{equation}
F_{\epsilon}=\int_{0}^{2\pi}\int_{0}^{\epsilon}b(\cos\theta)\left(\e^{\F(\theta)}-1\right) \sin\theta\, \rd\theta \,\rd\phi =\int_{0}^{2\pi}\int_{0}^{\epsilon}b(\cos\theta)\left(\tilde{F}(\theta)+O(\F^{2}(\theta))\right)\,\sin\theta\, \rd\theta \,\rd\phi,
\end{equation}
where a Taylor expansion is applied to the exponential function.

For the first order term in $F_{\epsilon}$, we have
\begin{equation}
\begin{split} \label{Feps1}
&\int_{0}^{2\pi}\int_{0}^{\epsilon}b(\cos\theta)\F(\theta)\sin\theta\,  \rd\theta\, \rd\phi\\ 
=&\int_{0}^{2\pi}\int_{0}^{\epsilon}b(\cos\theta) \im\frac{\pi}{2L}\aq k\cdot[\q(1-\cos\theta)-\h\sin\theta \cos\phi-\jj\sin\theta \sin\phi]\sin\theta \,\rd\theta\, \rd\phi\\
=&2\pi \int_{0}^{\epsilon} \im\frac{\pi}{2L}\aq (k\cdot\q) b(\cos\theta)(1-\cos\theta) \sin\theta\, \rd\theta,
\end{split}
\end{equation}
where the terms involving $\hat{h}$ and $\hat{j}$ integrate to zero due to periodicity in $\phi$. Using $1-\cos\theta=2\sin^{2}(\theta/2)\big|_{\theta\rightarrow 0}\sim \theta^2$ and (\ref{bb}) (with $d=3$), we have
\begin{equation}
b(\cos\theta) (1-\cos\theta) \sin\theta\Big |_{\theta\rightarrow 0} \sim K \theta^{1-\nu}, \quad 0 \leq \nu<2,
\end{equation}
hence the integral in (\ref{Feps1}) is integrable.

For the second order term in $F_{\epsilon}$, it is easy to see $O(\F^{2}(\theta))\big|_{\theta\rightarrow 0}\sim \theta^2$. Using again (\ref{bb}), we have
\begin{equation}
b(\cos\theta) O(\F^{2}(\theta)) \sin\theta \Big |_{\theta\rightarrow 0}\sim K \theta^{1-\nu},\quad 0 \leq \nu<2.
\end{equation}
Hence this term is also integrable.

To summarize, we have shown that the integral $F_{\epsilon}$ converges.

\end{proof}

\begin{remark}
Note that we used (\ref{sep}) to simplify the presentation but nothing is essential about this assumption.
\end{remark}


\section{Consistency and spectral accuracy}
\label{sec:consistency}

In this section, we prove the consistency result of the spectral approximation of the non-cutoff collision operator, that is, when $f$ has certain regularity, the Fourier approximation of the collision operator enjoys spectral accuracy.

In order to do so, we need the following important regularity result of the non-cutoff collision operator.
\begin{theorem}(\cite{alexandre2009review}, Theorem 7.4) \label{origthm}
	Assume that the collision kernel $ B\left( |v-v_{*}|,\cos\theta\right)=|v-v_*|^{\gamma} b(\cos \theta)$ with $\gamma \in \mathbb{R}$ and the angular part satisfying $\sin^{d-2}\theta b(\cos\theta) \Big|_{\theta\rightarrow 0} \sim K \theta^{-1-\nu}$, $0 < \nu<2$, $d=2$ or $3$. Then, for any $ m\in\mathbb{R} $, we have
	\begin{equation}
	\|\mQ(g,f)\|_{H^m\left(\mathbb{R}^{d}\right)} \leq C\|g\|_{L^1_{(\gamma+\nu)^{+}}\left(\mathbb{R}^{d}\right)}\|f\|_{H^{m+\nu}_{(\gamma+\nu)^{+}}\left(\mathbb{R}^{d}\right)},
	\end{equation}
where $ (\gamma+\nu)^{+}= \max\left\lbrace (\gamma+\nu),0 \right\rbrace $, $\mQ(g,f)$ is the bilinear collision operator given by
 \begin{equation}
	\mQ(g,f)(v) = \int_{\mathbb{R}^d}\int_{S^{d-1}}B\left( |v-v_{*}|,\cos\theta\right)\left[g(v_*')f(v')-g(v_*)f(v)\right]\,\rd{\sigma}\,\rd{v_*},
\end{equation}
and the weighted norms are defined as
\begin{equation}
\|f\|_{L^{p}_s} = \left( \int_{\mathbb{R}^{d}}|f(v)|^p(1+|v|^2)^{sp/2}\, \rd v \right) ^{1/p}, \quad
\|f\|_{H^m_s} = \left(\sum_{|i|\leq m}\|\partial^i f\|_{L^2_s}^2 \right)^{1/2}.
\end{equation}
\end{theorem}

The above theorem can be easily generalized to our setup in the bounded domain $\mathcal{D}_L=[-L,L]^d$ and the truncated collision operator $\mQ^R(g,f)(v)$.
\begin{lemma}\label{BB}
Assume $f$, $g$ are compactly supported in $\mathcal{B}_S$, and $R\geq 2S$, $L\geq \sqrt{2}S$. Then under the same condition as Theorem \ref{origthm}, we have
	\begin{equation}
	\|\mQ^R(g,f)\|_{H^m(\mathcal{D}_L)} \leq C\|g\|_{L^2(\mathcal{D}_L)}\|f\|_{H^{m+\nu}(\mathcal{D}_L)},
	\end{equation}
	where $ C > 0 $ is a constant depending on $ d,\gamma,\nu,L $.
\end{lemma}
\begin{proof}
Note that if $f$, $g$ are compactly supported in $\mathcal{B}_S$, then $\mQ^R(g,f)(v)\equiv \mQ(g,f)(v)$, and $\mQ(g,f)(v)$ is compactly supported in $\mathcal{B}_{\sqrt{2}S}\subset \mathcal{D}_L$. Then using Theorem \ref{origthm}, we have
	\begin{equation}
	\begin{split}
	\|\mQ^R(g,f)\|_{H^m(\mathcal{D}_L)}& \leq C  (1+L^2)^{(\gamma+\nu)^{+}} \|g\|_{L^1(\mathcal{D}_L)}\|f\|_{H^{m+\nu}(\mathcal{D}_L)}\\
	&\leq C  (1+L^2)^{(\gamma+\nu)^{+}} (2L)^{\frac{d}{2}}\|g\|_{L^2(\mathcal{D}_L)}\|f\|_{H^{m+\nu}(\mathcal{D}_L)},
	\end{split}
	\end{equation}
	where we used the Cauchy-Schwartz inequality in the second inequality.	
\end{proof}

For a periodic function $f(v)\in L^2(\mathcal{D}_L)$, we define its Fourier projection as
\begin{equation}
\mP_Nf=\sum_{k=-\frac{N}{2}}^{\frac{N}{2}} \hat{f}_k \e^{\im \frac{\pi}{L}k \cdot v}, \quad \hat{f}_k=\langle f, \e^{\im \frac{\pi}{L}k \cdot v}\rangle.
\end{equation}
We have the following basic fact regarding the projection operator (see for instance \cite{HGG07}).

\begin{lemma} \label{PP}
For any $m, r\in \mathbb{R}$ such that $0\leq m\leq r$, if a periodic function $f\in H^r(\mathcal{D}_L)$, then there hold
\begin{equation}
\|f-\mP_Nf\|_{L^2}\leq \frac{C}{N^r}\|f\|_{H^r}, \quad \|f-\mP_Nf\|_{H^m}\leq \frac{C}{N^{r-m}}\|f\|_{H^r}.
\end{equation}
\end{lemma}

We are ready to prove our main result.
\begin{theorem} \label{mainthm}
Assume that the collision kernel $ B\left( |v-v_{*}|,\cos\theta\right)=|v-v_*|^{\gamma} b(\cos \theta)$ with $\gamma \in \mathbb{R}$ and the angular part satisfying $ \sin^{d-2}\theta b\left(\cos\theta\right)\Big|_{\theta\rightarrow 0} \sim K \theta^{-1-\nu}$, $0 <\nu<2$, $d=2$ or $3$. Furthermore, assume $f$ is compactly supported in $\mathcal{B}_S$, and $R\geq 2S$, $L\geq \sqrt{2}S$. Then for any $r\in \mathbb{R}$ such that $r\geq \nu$, if $f\in H^r(\mathcal{D}_L)$, we have
\begin{equation}
\|\mQ^R(f,f)-\mP_N\mQ^R(\mP_Nf,\mP_Nf)\|_{L^2} \leq \frac{C}{N^{r-\nu}}\left(\|f\|_{L^2}\|f\|_{H^{r}}+\|f\|_{H^{\nu}}\|f\|_{H^{r-\nu}}\right).
\end{equation}
\end{theorem}

\begin{proof}
By the obvious triangle inequality, 
	\begin{equation}
	\begin{split}
	&\|\mQ^R(f,f)-\mP_N\mQ^R(\mP_Nf,\mP_Nf)\|_{L^2} \\
	\leq &\|\mQ^R(f,f)-\mP_N\mQ^R(f,f)\|_{L^2} + \|\mP_N\mQ^R(f,f)-\mP_N\mQ^R(\mP_Nf,\mP_Nf)\|_{L^2}.
	\end{split}
	\end{equation}
	
For the first term, we have for any $r\geq \nu$,
	\begin{equation} \label{term1}
	 \|\mQ^R(f,f)-\mP_N\mQ^R(f,f)\|_{L^2} \leq \frac{C}{N^{r-\nu}}\|\mQ^R(f,f)\|_{H^{r-\nu}}\leq \frac{C}{N^{r-\nu}}\|f\|_{L^2}\|f\|_{H^{r}}, 
	\end{equation}
where we used the Lemma \ref{PP} in the first inequality, and Lemma \ref{BB} in the second inequality.

For the second term, we have for any $r\geq \nu$,
	\begin{equation} \label{term2}
	\begin{split}
	& \|\mP_N\mQ^R(f,f)-\mP_N\mQ^R(\mP_Nf,\mP_Nf)\|_{L^2} \\
\leq  &\|\mQ^R(f,f)-\mQ^R(\mP_Nf,\mP_Nf)\|_{L^2}\\
 \leq  & \|\mQ^R(f-\mP_Nf,f)\|_{L^2}+\|\mQ^R(\mP_Nf,f-\mP_Nf)\|_{L^2}\\
 \leq & C\|f-\mP_Nf\|_{L^2}\|f\|_{H^{\nu}}+C\|\mP_N f\|_{L^2}\|f-\mP_Nf\|_{H^{\nu}}\\
 \leq  & \frac{C}{N^{r-\nu}}\|f\|_{H^{r-\nu}}\|f\|_{H^{\nu}}+\frac{C}{N^{r-\nu}}\|f\|_{L^2}\|f\|_{H^{r}},
	\end{split}
	\end{equation}
where we used the Parseval's inequality in the first inequality, Lemma \ref{BB} (with $m=0$) in the third inequality, and Lemma \ref{PP} in the last inequality.

Combining (\ref{term1}) and (\ref{term2}), we obtain the desired inequality.
\end{proof}

As a corollary, we have the spectral accuracy for the moments as well.
\begin{corollary}
Under the same condition as Theorem \ref{mainthm}, if a function $\phi\in L^2(\mathcal{D}_L)$, we have
\begin{equation}
\big|\langle \mQ^R(f,f),\phi\rangle -\langle \mP_N\mQ^R(\mP_Nf,\mP_Nf),\phi\rangle \big| \leq   \frac{C}{N^{r-\nu}}.
\end{equation}
\end{corollary}

\begin{proof}
Using the Cauchy-Schwartz inequality and Theorem \ref{mainthm},
\begin{equation}
\begin{split}
\big|\langle \mQ^R(f,f),\phi\rangle -\langle \mP_N\mQ^R(\mP_Nf,\mP_Nf),\phi\rangle \big| & \leq \|\mQ^R(f,f)-\mP_N\mQ^R(\mP_Nf,\mP_Nf)\|_{L^2} \|\phi\|_{L^2}\\
&\leq \frac{C}{N^{r-\nu}}\left(\|f\|_{L^2}\|f\|_{H^{r}}+\|f\|_{H^{\nu}}\|f\|_{H^{r-\nu}}\right)\|\phi\|_{L^2}.
\end{split}
\end{equation}
\end{proof}



\section{A fast algorithm and precomputation of the weight}
\label{sec:fast}

Now the validity of the Fourier spectral method for the non-cutoff Boltzmann equation has been justified. When it comes to implementation, the method requires the storage of the precomputed weight $G(l,m)$ as defined in (\ref{weight}) and a direct evaluation of the sum (\ref{sum}). Assume $N$ points (basis) are used in each velocity dimension, the total computational cost would be $O(N^{2d})$ and the same amount of memory is required to store the weight matrix. Therefore, the direct spectral method is both computationally expensive and memory consuming, especially for three dimensional problems.

Recently in  \cite{GHHH17, HM19}, a fast algorithm is introduced to accelerate the direct Fourier spectral method as well as to alleviate its memory requirement. The idea is to shift some offline precomputed items to online computation so that the sum (\ref{sum}), which is a weighted convolution, can be rendered into a few pure convolutions to be evaluated efficiently by the fast Fourier transform (FFT). Fortunately this idea can be generalized to the non-cutoff case without much change, which we briefly describe below.

Our goal is to find a low-rank decomposition of $G(l,m)$ in \eqref{weight} as follows
\begin{equation}\label{GG1}
G(l,m)\approx\sum_{p=1}^{N_{p}}\alpha_{p}(l+m)\beta_{p}(m),
\end{equation}
where $\alpha_{p}$ and $\beta_{p}$ are some functions to be determined and the number of terms $N_{p}$ in the expansion is small. With this approximation, (\ref{sum}) becomes
\begin{equation}
\mQ_{k}^R\approx \sum_{p=1}^{N_{p}}\alpha_{p}(k)\sum\limits_{\substack{l,m=-\frac{N}{2}\\l+m=k}}^{\frac{N}{2}}f_{l}\left(\beta_{p}(m)f_{m}\right),
\end{equation}
where the inner summation is a convolution of two functions $f_{l}$ and $\beta_{p}(m)f_{m}$. Hence the total cost to evaluate $\mQ^R_{k}$ (for all $k$) can be reduced from $O(N^{2d})$ to $O(N_{p}N^{d}\log N)$ with the help of a few FFTs. 

To find the decomposition as in (\ref{GG1}), one just needs to use the form (\ref{GG}) and approximates the integrals in $|q|$ and $\hat{q}$ using quadratures as
\begin{equation}\label{approxG}
\begin{split}
G(l,m) \approx \sum_{\aq,\q}w_{\aq}w_{\q} \Phi(|q|)\aq^{d-1}\e^{-\im \frac{\pi}{L}\aq m\cdot\q}F(l+m,\aq,\q),
\end{split}
\end{equation}
where $w_{\aq}$ and $w_{\q}$ are the corresponding quadrature weights. In practice, we use $N_{\aq}=O(N)$ Gauss-Legendre quadrature points to discretize $|q|$ and $N_{\q}\ll N$ Spherical Design \cite{Womersley} quadrature points to discretize $\q$. Now using (\ref{approxG}), (\ref{sum}) is approximated by
\begin{equation}\label{realQ}
\mQ^R_{k}\approx  \sum_{\aq,\q}w_{\aq}w_{\q}\Phi(|q|)\ \aq^{d-1}F(k,\aq,\q)\sum\limits_{\substack{l,m=-\frac{N}{2}\\l+m=k}}^{\frac{N}{2}}f_{l}\left(\e^{-\im \frac{\pi}{L}\aq m\cdot\q}f_{m}\right):=\tilde{\mQ}^R_k.
\end{equation}
Therefore, the total cost to evaluate ${\mQ}_{k}^R$ is $O(N_{\q}N^{d+1}\log N)$. What's more, the only term that needs to be precomputed and stored is the weight $F(k,\aq,\q)$ defined in (\ref{FF}), which requires $O(N_{\q} N^{d+1})$ memory at most.

\begin{remark}
The fast algorithm introduced above still preserves mass as in the direct spectral method. To see it, notice that 
\begin{equation}
\rho_N:=\int_{\mathcal{D}_L}f_N\,\rd{v}=(2L)^df_0(t),
\end{equation}
where $f_0$ is the zero-th mode of the numerical solution and is governed by
\begin{equation}
\frac{\rd}{\rd t} f_0 = \tilde{\mQ}^R_0.
\end{equation}
From (\ref{realQ}) and the definition of $F$ in (\ref{FF}), it is easy to see $\tilde{\mQ}^R_0\equiv 0$ since $F(0,|q|,\hat{q})\equiv 0$.
\end{remark}

\subsection{Strategy in precomputation of $F(k,\aq,\q)$}
\label{strategyF}

From the previous discussion, it is clear that the online part of the fast algorithm is no different from that in the cutoff case. The main difference lies in the offline stage, i.e., the precomputation of the weight $F(k,\aq,\q)$. Indeed if the kernel is integrable, computing $F$ is rather straightforward. However, in the non-cutoff case, as we proved in Section~\ref{subsec:intergability}, $F$ contains an integrable singularity as $\theta\rightarrow 0$. Due to the cancellation effects of terms $b(\sigma\cdot\q)$ and $(\e^{\im \frac{\pi}{2L}|q| k\cdot (\hat{q}-\sigma)}-1)$ in (\ref{FF}), extra care is needed to compute the integral accurately. This is especially true when the singularity in the kernel is strong.

To be precise, we take the following strategy:

\vspace{0.1in}
{\bf (i) 2D case}: We start with the formula (\ref{FF2D1}). Since the singularity of $b(\cos\theta)$ appears both when $\theta\rightarrow 0$ and $\theta\rightarrow 2\pi$, we split the integration domain $\theta\in[0,2\pi]$ into three parts $[0,\epsilon]$, $[\epsilon,2\pi-\epsilon]$, and $[2\pi-\epsilon,2\pi]$:
\begin{equation}
\begin{split}
  &F(k,\aq,\q) = \int_{0}^{2\pi}b(\cos\theta)\left(\e^{\F(\theta)}-1\right) \,\rd\theta  \\
 =&\int_{0}^{\epsilon}b(\cos\theta)\left(\e^{\F(\theta)}-1\right) \,\rd\theta
  +\int_{\epsilon}^{2\pi-\epsilon}b(\cos\theta)\left(\e^{\F(\theta)}-1\right)\, \rd\theta
  +\int_{2\pi-\epsilon}^{2\pi}b(\cos\theta)\left(\e^{\F(\theta)}-1\right) \,\rd\theta \\
 \approx& \underbrace{\int_{0}^{\epsilon}b(\cos\theta)\left(\F(\theta)+\frac{1}{2}\F^2(\theta)\right)\, \rd\theta}_{I} + \underbrace{\int_{\epsilon}^{2\pi-\epsilon}b(\cos\theta)\left(\e^{\F(\theta)}-1\right)\, \rd\theta}_{II}  \\
 & + \underbrace{\int_{2\pi-\epsilon}^{2\pi}b(\cos\theta)\left(\F(\theta)+\frac{1}{2}\F^2(\theta)\right) \,\rd\theta}_{III},
\end{split}
\end{equation}
where for parts $ I $ and $ III $  the Taylor expansion of $\e^{\F(\theta)}$ up to second order is used, hence some angular terms can be cancelled immediately. After this manipulation, standard quadrature can be applied to each part. In our implementation, we calculate part $I$ and part $III$ exactly (after Taylor expansion), and apply the MATLAB built-in function ``integral" to part $II$.

\vspace{0.1in}
{\bf (ii) 3D case}: We start with the formula (\ref{FF3D1}) and split the integration domain $\theta\in[0,\pi]$ into two parts $[0,\epsilon]$ and \textcolor{red}{$[\epsilon,\pi]$}:
\begin{equation}
\begin{split}
&F(k,\aq,\q)= \int_{0}^{2\pi}\int_{0}^{\pi}b(\cos\theta)\left(\e^{\F(\theta)}-1\right) \sin\theta \,\rd\theta\, \rd\phi   \\
= &\int_{0}^{2\pi}\int_{0}^{\epsilon}b(\cos\theta) \left(\e^{\F(\theta)}-1\right)\sin\theta \,  \rd\theta \,\rd\phi + \int_{0}^{2\pi}\int_{\epsilon}^{\pi}b(\cos\theta) \left(\e^{\F(\theta)}-1\right)\sin\theta \, \rd\theta \,\rd\phi   \\
\approx & \underbrace{\int_{0}^{2\pi}\int_{0}^{\epsilon}b(\cos\theta) \left(\F(\theta)+\frac{1}{2}\F^2(\theta)\right) \sin\theta \, \rd\theta\, \rd\phi}_{I} + \underbrace{ \int_{0}^{2\pi}\int_{\epsilon}^{\pi}b(\cos\theta)\left(\e^{\F(\theta)}-1\right)\sin\theta  \, \rd\theta \,\rd\phi}_{II},
\end{split}
\end{equation}
where for part $ I $  the Taylor expansion of $\e^{\F(\theta)}$ up to second order is again used to cancel some angular terms. After this manipulation, standard quadrature can be applied to each part. In our implementation, for the integral in $\theta$, we calculate part $I$ exactly (after Taylor expansion) and apply the MATLAB built-in function ``integral" to part $II$; for the integral in $\phi$, we use the mid-point rule for both part $I$ and part $II$.

\vspace{0.1in}
In practice, we choose $\epsilon=\pi/1000$ and the numerical results in the next section (in particular the BKW tests) imply that $F(k,\aq,\q)$ has been computed to the same accuracy as in the cutoff case. 

\begin{remark}
Similarly as in Section~\ref{subsec:intergability}, the assumption (\ref{sep}) is used to simplify the presentation but all the discussion in this section works for general kernels of the form $B(|q|,\sigma\cdot \hat{q})$.
\end{remark}

\subsection{Key differences between the cutoff case and non-cutoff case}

Although formally the fast Fourier spectral method presented above can be implemented the same in both cutoff and non-cutoff cases (provided the weight $F(k,\aq,\q)$ has been precomputed), we would like to point out a few key differences between the two cases.

First of all, in the cutoff case, for quite a few collision kernels commonly used for numerical purpose such as the variable hard sphere model (VHS) \cite{Bird}, where $B(|q|,\sigma\cdot \hat{q})=C|q|^{\gamma}$ only has the velocity dependence, there exists analytical formula for $F(k,\aq,\q)$ hence no precomputation is needed. Indeed, if $b(\cos\theta)\equiv C$, in 2D,
\begin{equation}
\begin{split} \label{cutoff2D}
F(k,\aq,\q) &=C \int_{S^{1}}\left(\e^{\im\frac{\pi}{2L}\aq k\cdot(\q-\sigma)}-1\right)\,\rd{\sigma}=C\left(\e^{\im\frac{\pi}{2L}\aq k\cdot \hat{q}}\int_{S^{1}} \e^{-\im\frac{\pi}{2L}\aq k\cdot \sigma}\,\rd{\sigma}-2\pi\right)\\
&=2\pi C\left[\e^{\im\frac{\pi}{2L}\aq k\cdot \hat{q}}J_0\left(\frac{\pi}{2L}|q||k|\right)-1\right];
\end{split}
\end{equation}
and in 3D,
\begin{equation}
\begin{split} \label{cutoff3D}
F(k,\aq,\q) &=C \int_{S^{2}}\left(\e^{\im\frac{\pi}{2L}\aq k\cdot(\q-\sigma)}-1\right)\,\rd{\sigma}=C\left(\e^{\im\frac{\pi}{2L}\aq k\cdot \hat{q}}\int_{S^{2}} \e^{-\im\frac{\pi}{2L}\aq k\cdot \sigma}\,\rd{\sigma}-4\pi\right)\\
&=4\pi C\left[\e^{\im\frac{\pi}{2L}\aq k\cdot \hat{q}}\text{Sinc}\left(\frac{\pi}{2L}|q||k|\right)-1\right].
\end{split}
\end{equation}
However, in the non-cutoff case, precomputation is always inevitable.

Secondly, in the cutoff case, one can separate the gain (positive) term and loss (negative) term in the collision operator. Since the loss term under the Fourier spectral approximation is readily a convolution, no extra low-rank approximation as in (\ref{approxG}) is needed. Numerical experiments suggest that this way would yield better accuracy in comparison to computing the gain and loss terms together using (\ref{realQ}), especially for anisotropic solutions, see \cite{HM19}. Unfortunately, this option is not available in the non-cutoff case as the gain and loss terms cannot be separated (they have to be viewed together since each of them is a divergent integral).


\section{Numerical results}
\label{sec:num}

In this section, we perform a series of numerical tests to demonstrate the accuracy and efficiency of the proposed method in 2D and 3D cases. We first carefully validate the accuracy of the method using an analytical solution, which can be constructed for both cutoff and non-cutoff collision kernels. We then use the method to simulate a few examples with measure valued initial data, where we observe very different solution behavior for different kernels.

\subsection{Some preliminaries on the BKW solution with non-cutoff kernels}
\label{sec:BKW}

The Bobylev-Krook-Wu (BKW) solution \cite{Bobylev75_1, KW77} is one of the few analytical solutions one can construct for the homogeneous Boltzmann equation with Maxwell molecules (i.e., $B(|q|,\sigma\cdot \hat{q})=b(\sigma\cdot \hat{q})$ in (\ref{CO})). Although the BKW solution (with cutoff Maxwell kernels) has been widely used to validate the deterministic numerical solvers for the Boltzmann equation, it is not well recognized that the solution is also valid for non-cutoff kernels, hence is an ideal candidate to test the accuracy of the proposed method. For this reason, we briefly describe the construction of the solution in this subsection.

The BKW solution is an isotropic function of the form:
\begin{equation} \label{BKW}
f(t,v)=\frac{1}{(2\pi \mathcal{K})^{d/2}}\exp \left ( -\frac{|v|^2}{2\mathcal{K}}\right)\left(\frac{(d+2)\mathcal{K}-d}{2\mathcal{K}}+\frac{1-\mathcal{K}}{2\mathcal{K}^2}|v|^2\right).
\end{equation}
In order for (\ref{BKW}) to be a solution of (\ref{eqn}), it can be verified by direct substitution that $\mathcal{K}=\mathcal{K}(t)$ must satisfy
\begin{equation} \label{odeK}
\mathcal{K}'=\lambda(1-\mathcal{K}),
\end{equation}
with 
\begin{equation}
\lambda=\frac{1}{4}\int_{S^{d-1}}\left(1-(\sigma\cdot \hat{q})^2\right)b(\sigma\cdot \hat{q})\,\rd{\sigma},
\end{equation}
which indicates
\begin{equation}
\mathcal{K}=1-C\exp(-\lambda t).
\end{equation}
Differentiating (\ref{BKW}) and using (\ref{odeK}), we obtain
\begin{equation} \label{QQ}
\mathcal{Q}(f,f)=\partial_tf=\frac{1}{(2\pi \mathcal{K})^{d/2}}\exp \left ( -\frac{|v|^2}{2\mathcal{K}}\right) \frac{(1-\mathcal{K})^2}{4\mathcal{K}^4}\lambda \left[ d(d+2)\mathcal{K}^2-2(d+2)\mathcal{K}|v|^2+|v|^4\right].
\end{equation}

In 2D, we can choose benchmark values 
\begin{equation}
C=\frac{1}{2}, \quad b(\sigma\cdot \hat{q})\equiv \frac{1}{2\pi},
\end{equation}
which leads to
\begin{equation} \label{KK2}
\lambda=\frac{1}{8}, \quad \mathcal{K}=1-\frac{1}{2}\exp\left(-\frac{t}{8}\right).
\end{equation}
Based on these values, we can construct several non-cutoff kernels $b(\sigma\cdot \hat{q})$ with different degree of singularity but all correspond to the same value of $\lambda$, hence the same shape of the solution.

In 3D, we can choose benchmark values
\begin{equation}
C=1, \quad b(\sigma\cdot \hat{q})\equiv \frac{1}{4\pi},
\end{equation}
which leads to
\begin{equation} \label{KK3}
\lambda=\frac{1}{6}, \quad \mathcal{K}=1-\exp\left(-\frac{t}{6}\right).
\end{equation}
Similarly to 2D, we can construct several non-cutoff kernels $b(\sigma\cdot \hat{q})$ that all correspond to the same $\lambda$ and same solution.


\subsection{2D BKW solution -- Maxwell molecule}
\label{sec:2DBKW}

Based on the discussion in Section~\ref{sec:BKW}, we construct the following collision kernels which all correspond to the same $\lambda$ and $\mathcal{K}$ as given in (\ref{KK2}).
\begin{itemize}
\item Cutoff kernel $b_1$:
\begin{equation}\label{b1}
b_{1}(\sigma\cdot \q)=b_1(\cos\theta)=\frac{1}{2\pi}, \quad \theta \in [0,2\pi] \quad \text{with} \quad \int_{S^{1}}b_{1}\,\rd\sigma=1.
\end{equation}
\item Non-cutoff kernel $b_2$:
\begin{equation}\label{b2}
b_{2}(\sigma\cdot \q)=b_{2}(\cos\theta)=\frac{3}{32\sin\frac{\theta}{2}}, \quad \theta \in [0,2\pi] \quad \text{with} \quad \int_{S^{1}}b_{2}\,\rd\sigma=+\infty.
\end{equation}
The order of singularity of $b_{2}$ is $ b_{2} \big|_{\theta\rightarrow 0}\sim\theta^{-1-\nu}$ with $\nu=0$. 
\item Non-cutoff kernel $b_3$:
\begin{equation}\label{b3}
b_{3}(\sigma\cdot \q)=b_{3}(\cos\theta) = \frac{1}{8\pi\sin^{2}\frac{\theta}{2}}, \quad \theta \in [0,2\pi]\quad \text{with} \quad \int_{S^{1}}b_{3}\,\rd\sigma=+\infty.
\end{equation}
The order of singularity of $b_{3}$ is $ b_{3}\big|_{\theta\rightarrow 0}\sim\theta^{-1-\nu}$ with $\nu=1$. 
\item  Non-cutoff kernel $b_4$:
\begin{equation}\label{b4}
b_{4}(\sigma\cdot \q)=b_{4}(\cos\theta)=\frac{5|\cos\frac{\theta}{2}|}{256\sin^{\frac{5}{2}}\frac{\theta}{2}}, \quad \theta \in [0,2\pi] \quad \text{with} \quad \int_{S^{1}}b_{4}\,\rd\sigma=+\infty.
\end{equation}
The order of singularity of $b_{4}$ is $ b_{4}\big|_{\theta\rightarrow 0} \sim\theta^{-1-\nu}$ with $\nu=\frac{3}{2}$. 
\end{itemize}

For all the above four kernels, they yield the same solution (\ref{BKW}). Without introducing any time discretization error, we verify the accuracy of our method by evaluating (\ref{QQ}) at certain time. The results are reported in Table~\ref{table1}, which demonstrate that the Fourier spectral method in the non-cutoff case works equally well as the cutoff case. Note that for the cutoff kernel $b_1$, one can just use the analytical formula (\ref{cutoff2D}) to get $F(k,|q|,\hat{q})$. This, on the other hand, indicates that our strategy of precomputing the weight $F(k,|q|,\hat{q})$ is reliable. 

\begin{table}[h!]
	\centering
	\begin{tabular}{c | c | c | c | c}\hline
		$N$ & $b_{1}$ & $b_{2}$ & $b_{3}$ & $b_{4}$\\
		\hline
		8 & 1.8411e-02 & 1.8612e-02 & 1.9054e-02 & 1.9569e-02\\
		16 & 1.0692e-03 & 1.0806e-03 & 2.1531e-03 & 3.9562e-03\\
		32 & 1.4704e-07 & 1.3363e-07 & 1.0620e-07 & 3.0431e-07\\
		64 & 2.8322e-09 & 2.9002e-09 & 3.1950e-09 & 4.4349e-09\\
		\hline
	\end{tabular}
	\caption{ Section~\ref{sec:2DBKW}: 2D BKW solution -- Maxwell molecule. $\|\mQ^{\text{ext}}(f,f)-\mQ^{\text{num}}(f,f)\|_{L^{\infty}}$ at $t=0$. $N$ is the number of points in each velocity dimension. $N_{\aq}=N$ is the number of points used in the radial direction (with Gauss-Legendre quadrature). $N_{\hat{q}}=32$ is the number of points used in the angular direction (with mid-point quadrature). $R=6$, $L=(3+\sqrt{2})R/4\approx6.62$.}
\label{table1}
\end{table}

To examine the error evolution in time, we next use our method to solve the homogeneous Boltzmann equation. The classical fourth-order Runge-Kutta method is employed for time discretization to ensure that the temporal error does not pollute the spectral accuracy in velocity. The result is shown in Figure~\ref{fig1}, where there is no significant difference among four kernels.
\begin{figure}[htp]
	\centering
	\includegraphics[width = \linewidth]{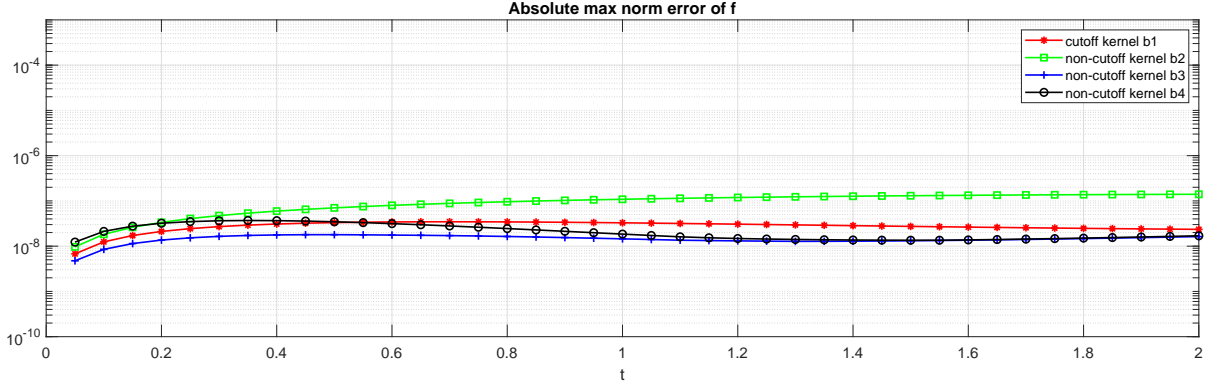}	
	\caption{ Section~\ref{sec:2DBKW}: 2D BKW solution -- Maxwell molecule. Time evolution of $\|f^{\text{ext}}-f^{\text{num}}\|_{L^\infty}$. Classical RK4 with $\Delta t=0.05$ for time discretization. $N=N_{|q|}=N_{\hat{q}}=32$. $R=6$, $L=(3+\sqrt{2})R/4\approx6.62$.}
	\label{fig1}
\end{figure}

\subsection{3D BKW solution -- Maxwell molecule}
\label{3D}

Based on the discussion in Section~\ref{sec:BKW}, we construct the following collision kernels which all correspond to the same $\lambda$ and $K$ as given in (\ref{KK3}).
\begin{itemize}
\item Cutoff kernel $b_5$:
\begin{equation}\label{b5}
b_5(\sigma\cdot \q) =b_{5}(\cos\theta)=\frac{1}{4\pi}, \quad \theta \in [0,\pi] \quad \text{with}\quad \int_{S^{2}}b_{5}\,\rd\sigma=1.
\end{equation}
\item Non-cutoff kernel $b_6$:
\begin{equation}\label{b6}
b_6(\sigma\cdot \q) = b_{6}(\cos\theta)=\frac{1}{8\pi\sin\theta\sin\frac{\theta}{2}}, \quad \theta \in [0,\pi]\quad \text{with} \quad \int_{S^{2}}b_{6}\,\rd\sigma=+\infty.
\end{equation}
The order of singularity of $b_{6}$ is $\sin\theta b_{6}\big|_{\theta\rightarrow 0}\sim\theta^{-1-\nu}$ with $\nu=0$. 
\item Non-cutoff kernel $b_7$:
\begin{equation}\label{b7}
b_7(\sigma\cdot \q) =b_{7}(\cos\theta)=\frac{1}{6\pi^2\sin\theta\sin^{2}\frac{\theta}{2}},\quad \theta \in [0,\pi] \quad \text{with} \quad \int_{S^{2}}b_{7}\,\rd\sigma=+\infty.
\end{equation}
The order of singularity of $b_{7}$ is $\sin\theta b_{7}\big|_{\theta\rightarrow 0}\sim\theta^{-1-\nu}$ with $\nu=1$. 
\item Non-cutoff kernel $b_8$:
\begin{equation}\label{b8}
b_8(\sigma\cdot \q) =b_{8}(\cos\theta)=\frac{5\cos\frac{\theta}{2}}{192\pi\sin\theta \sin^{\frac{5}{2}}\frac{\theta}{2}},\quad \theta\in [0,\pi]\quad \text{with} \quad \int_{S^{2}}b_{8}\,\rd\sigma=+\infty.
\end{equation}
The order of singularity of $b_{8}$ is $\sin\theta b_{8}\big|_{\theta\rightarrow 0}\sim\theta^{-1-\nu}$ with $\nu=\frac{3}{2}$. 
\end{itemize}

We now perform a similar test as in 2D with the above four kernels. The results are reported in Table~\ref{table2}. Since the integration on the sphere is harder than that over the circle, our focus here is to demonstrate the convergence with respect to the spherical quadrature. Note that for the cutoff kernel $b_5$, one can just use the analytical formula (\ref{cutoff3D}) to get $F(k,|q|,\hat{q})$. Again we can see that the method can achieve the same level of accuracy for both the cutoff and non-cutoff kernels. 

\begin{table}[h!]\centering
		\begin{tabular}{c | c | c | c | c }\hline
		$N_{\hat{q}}$ & $b_{5}$ & $b_{6}$ & $b_{7}$ & $b_{8}$\\
		\hline
		12 & 4.1224e-04 & 5.5098e-04 & 1.5792e-03   & 3.4256e-03 \\
		48 & 5.7277e-05 & 9.1134e-05 & 1.7077e-04   & 2.8334e-04\\
		70 & 1.1213e-05 & 1.9541e-05 & 4.5150e-05   & 8.8087e-05\\
		120 & 9.7623e-07 & 1.5891e-06 & 4.6049e-06  & 1.0831e-05\\
		192 & 5.6276e-07 & 4.2911e-07 & 3.4111e-07  & 5.5735e-07\\
		\hline
	\end{tabular}
	\caption{ Section~\ref{3D}: 3D BKW solution -- Maxwell molecule. $\|\mQ^{\text{ext}}(f)-\mQ^{\text{num}}(f)\|_{L^{\infty}}$ at $t=6.5$. $N=32$ is the number of points in each velocity dimension. $N_{\aq}=32$ is the number of points used in the radial direction (with Gauss-Legendre quadrature). $N_{\hat{q}}$ is the number of points used in the sphere (with Spherical Design quadrature). $R=6$, $L=(3+\sqrt{2})R/4\approx6.62$.} 
	\label{table2} 
\end{table}


\subsection{Measure valued solution in 2D -- Maxwell molecule}
\label{2measure}

We now perform a series of numerical tests for the (approximate) measure valued solutions with time evolution using different collision kernels. The existence of measure valued solutions has been established in \cite{TV1999,morimoto2016measure}. Furthermore, it is known that the solution to the non-cutoff equation enjoys the smoothing effect if the initial datum is not a single Dirac delta function. Though theoretical regularity is hard to justify under numerical discretization, one can expect quite different behavior for different kernels.

We consider an initial condition of the form:
\begin{equation}\label{initialdelta}
f^0(v) = \frac{1}{3} \left( \delta_{w}(v) + \delta_{w}(|v|-0.2)\right),
\end{equation}
where $\delta_{w}(v)$ is an approximated delta function given as follows:
\begin{equation}\label{delta}
\delta_{w}(v)=
\begin{cases}
\frac{1}{2w}\left(1+\cos|\frac{\pi v}{w}|\right), \quad & |v|\leq w,\\
0, & \text{otherwise},
\end{cases}
\end{equation} 
and $w$ is taken to be $0.5\sqrt{\Delta v}$ ($\Delta v$ is the mesh size in velocity).

We first take the 2D non-cutoff kernel $b_3$ (\ref{b3}) as an example to illustrate the time evolution of the solution, see Figure~\ref{fig2} where the trend to Gaussian equilibrium is clear.

\begin{figure}[htp]
	\centering
	\subfigure[t=1]{
		\includegraphics[width=7cm]{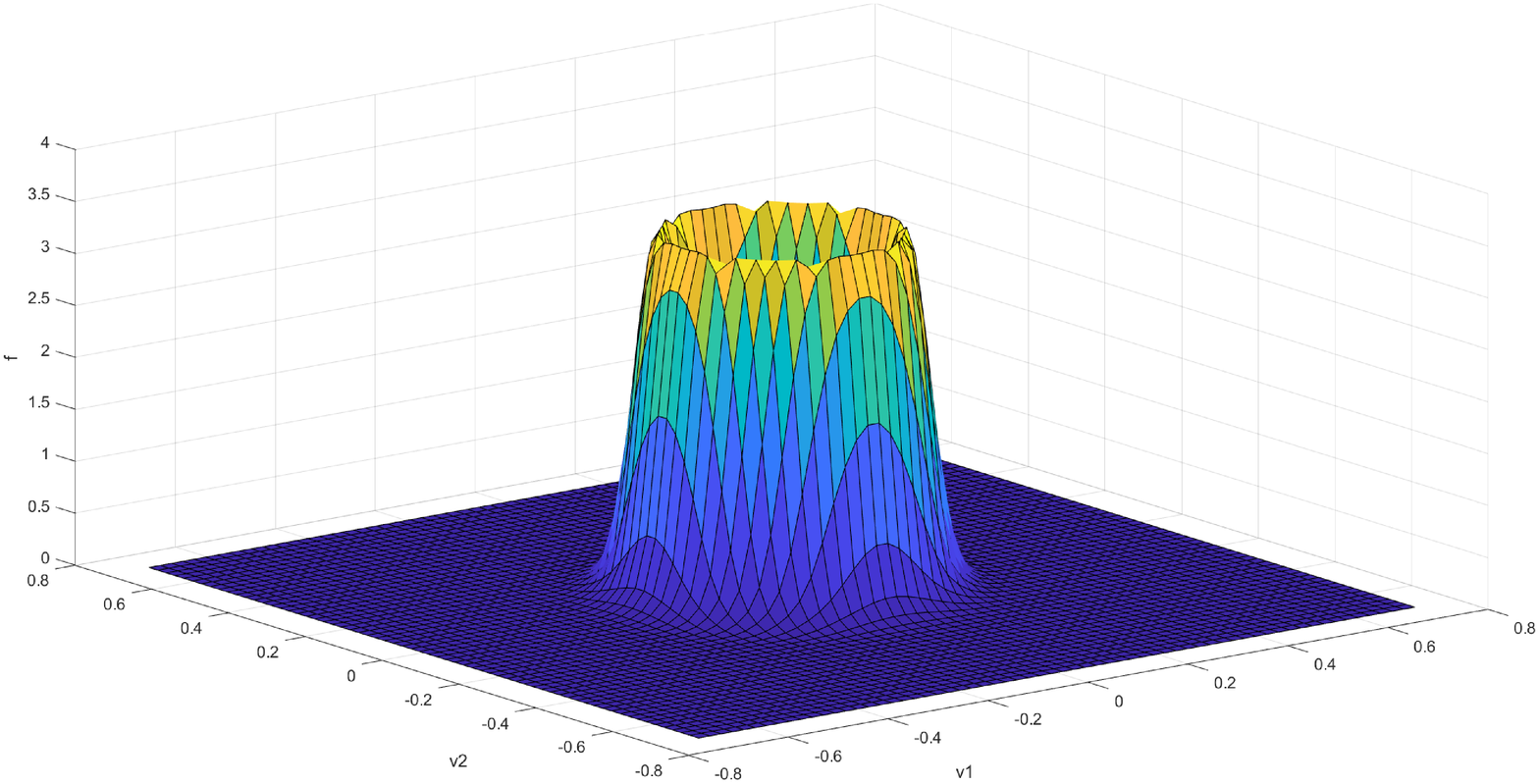}
	}
	\subfigure[t=2]{
		\includegraphics[width=7cm]{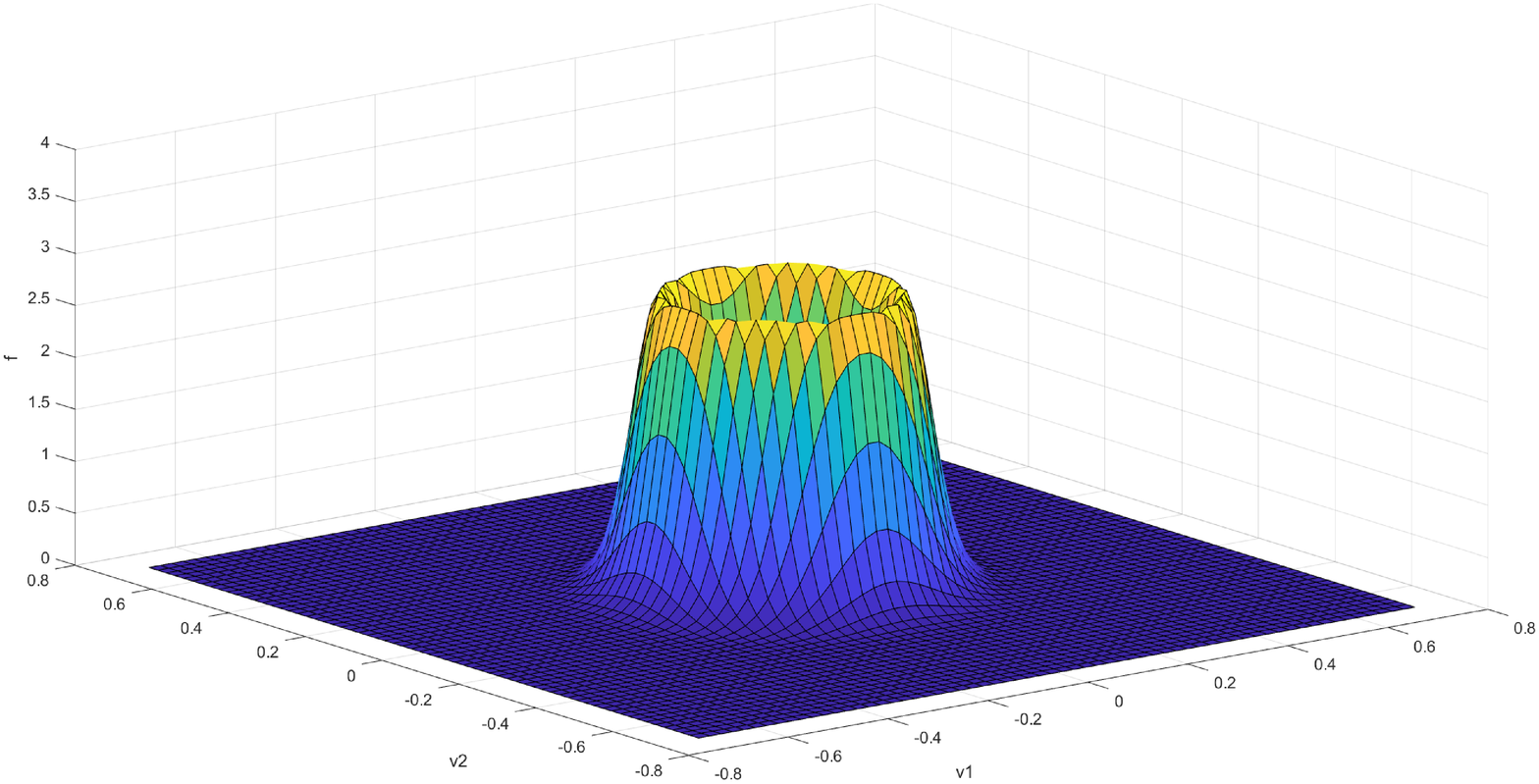}
	}\\
	\subfigure[t=4]{
		\includegraphics[width=7cm]{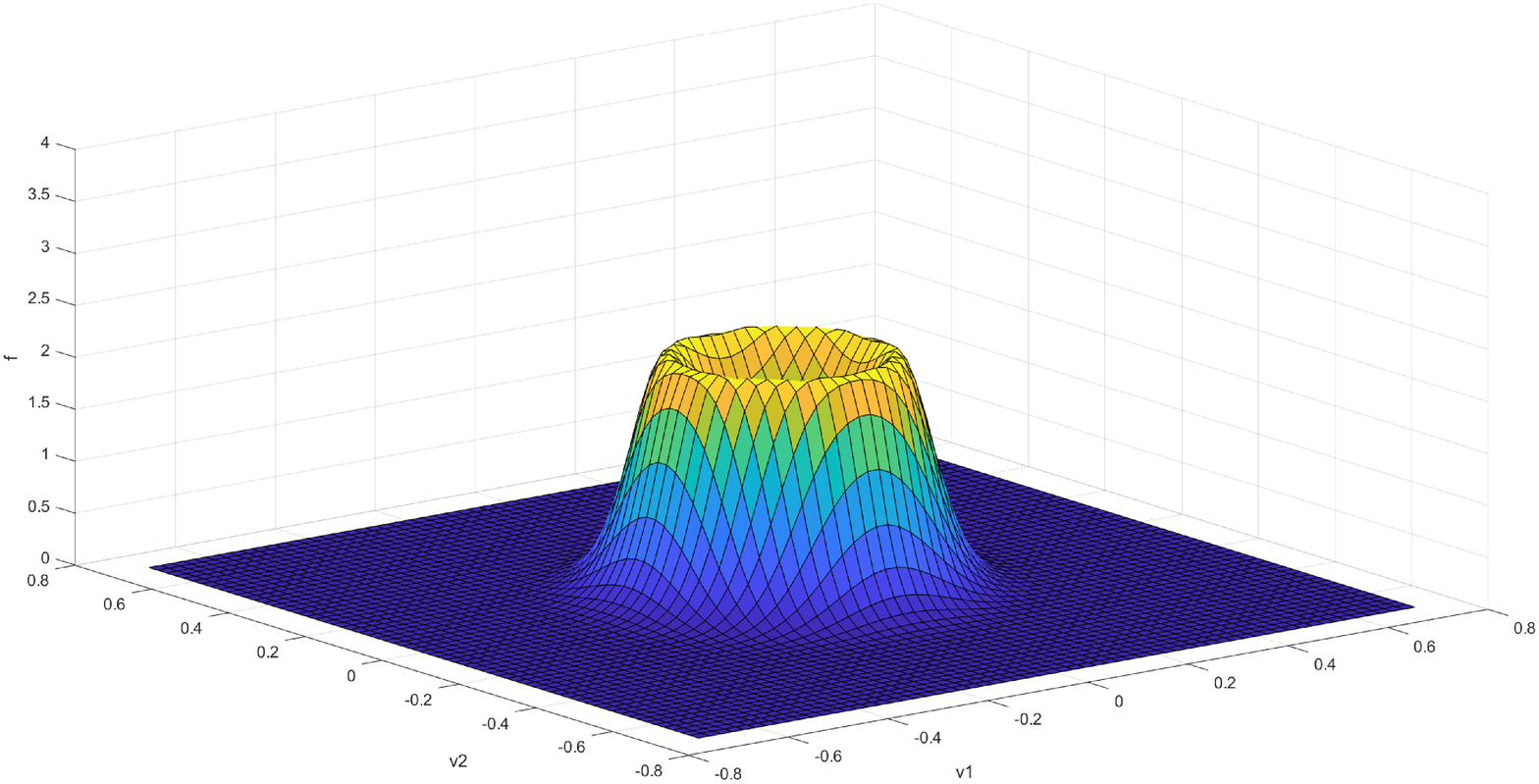}
	}
	\subfigure[t=6]{
		\includegraphics[width=7cm]{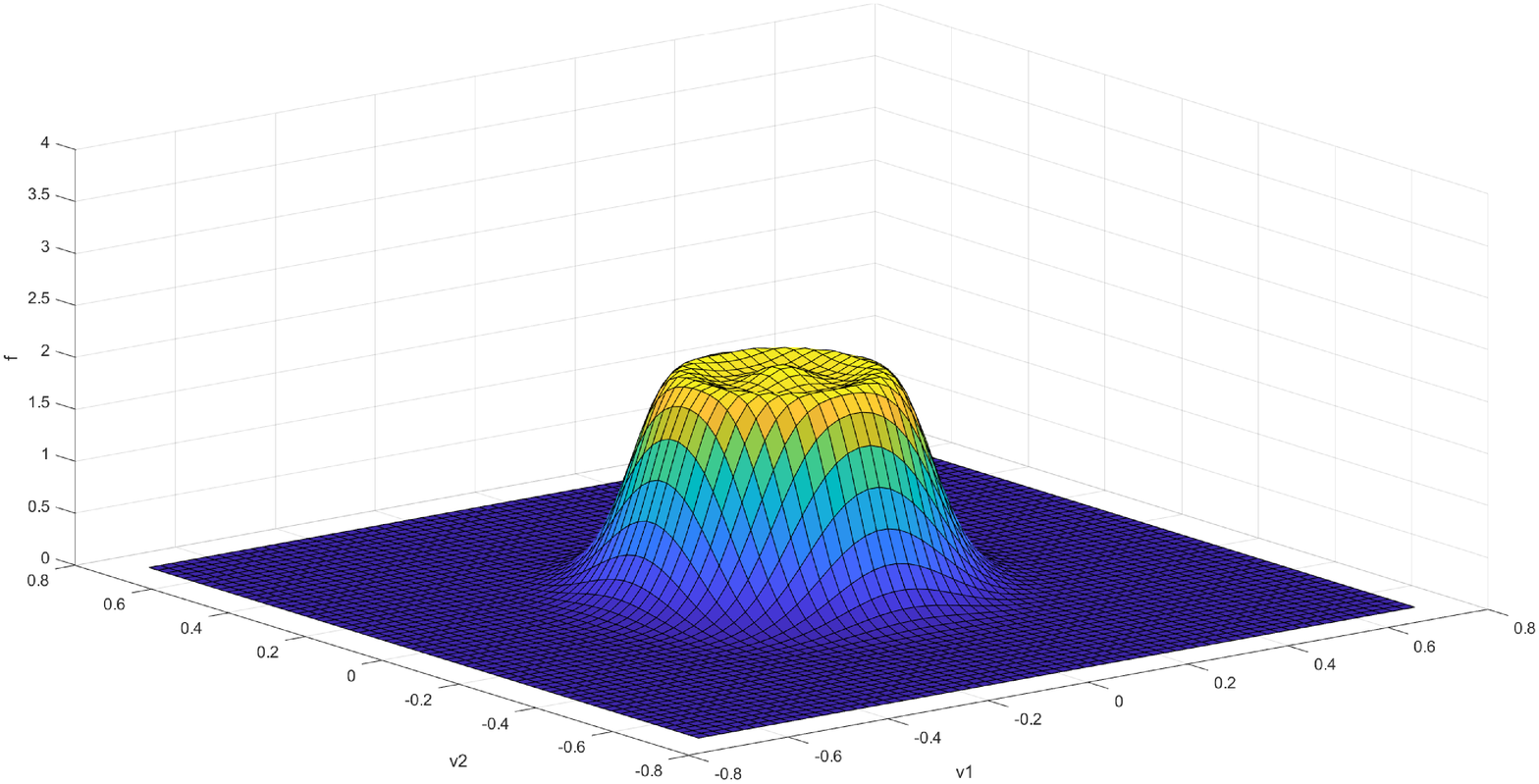}
	}\\
	\subfigure[t=8]{
		\includegraphics[width=7cm]{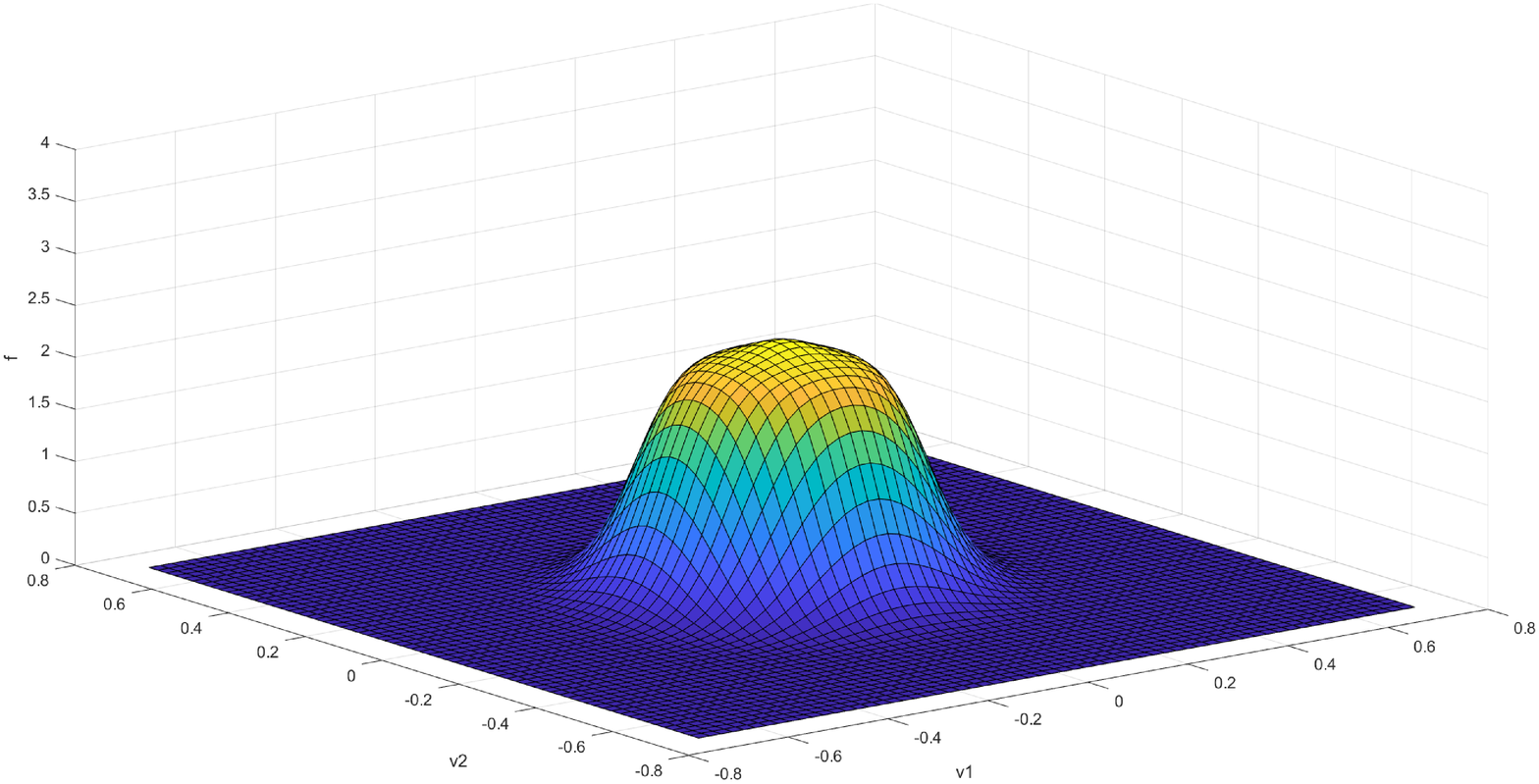}
	}
	\subfigure[t=10]{
		\includegraphics[width=7cm]{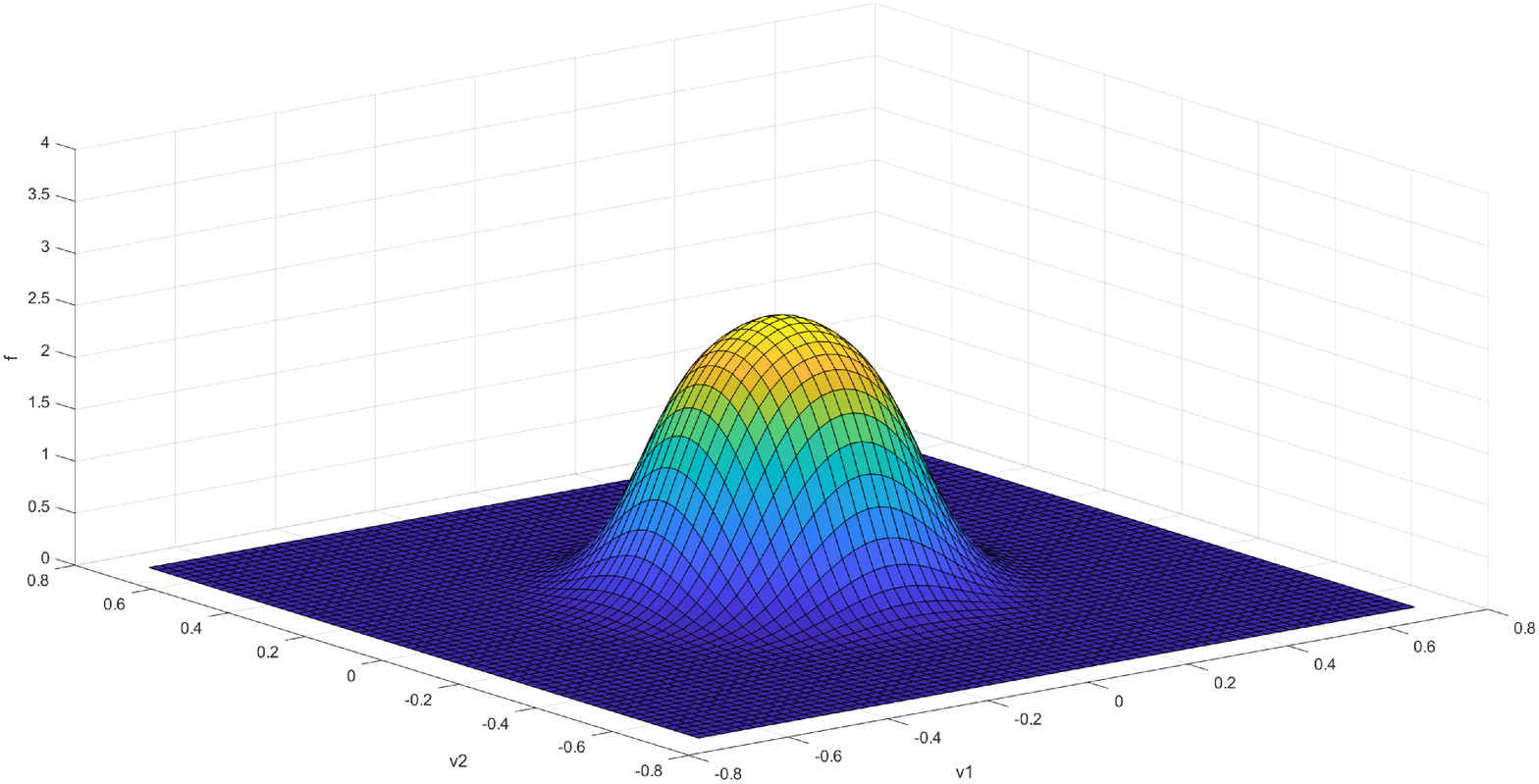}
	}
	\caption{ Section~\ref{2measure}: Measure valued solution in 2D -- Maxwell molecule. Time evolution of the distribution function $f$ with non-cutoff collision kernel $b_{3}$ and initial condition (\ref{initialdelta}). Classical RK4 with $\Delta t=0.05$ for time discretization. $N = N_{\aq} = 64$, $N_{\q} = 32$. $R=0.66$, $L=(3+\sqrt{2})R/4\approx0.73$.}
	\label{fig2}
\end{figure}

We then compare the solution profiles computed with four different kernels $b_1$ (\ref{b1}), $b_2$ (\ref{b2}), $b_3$ (\ref{b3}), and $b_4$ (\ref{b4}). The results are shown in Figure~\ref{fig3}. We can observe that although all solutions converge to the same equilibrium in the end, the non-cutoff solutions tend to be smoothed out faster compared to the cutoff one, and the higher the singularity is in the kernel, the smoother the solution behaves. This is quite striking and is the first time such differences between the cutoff and non-cutoff Boltzmann solutions are reported in the literature, as far as we know.

\begin{figure}[htp]
	\centering
	\subfigure[t=1]{
		\includegraphics[width=6.5cm]{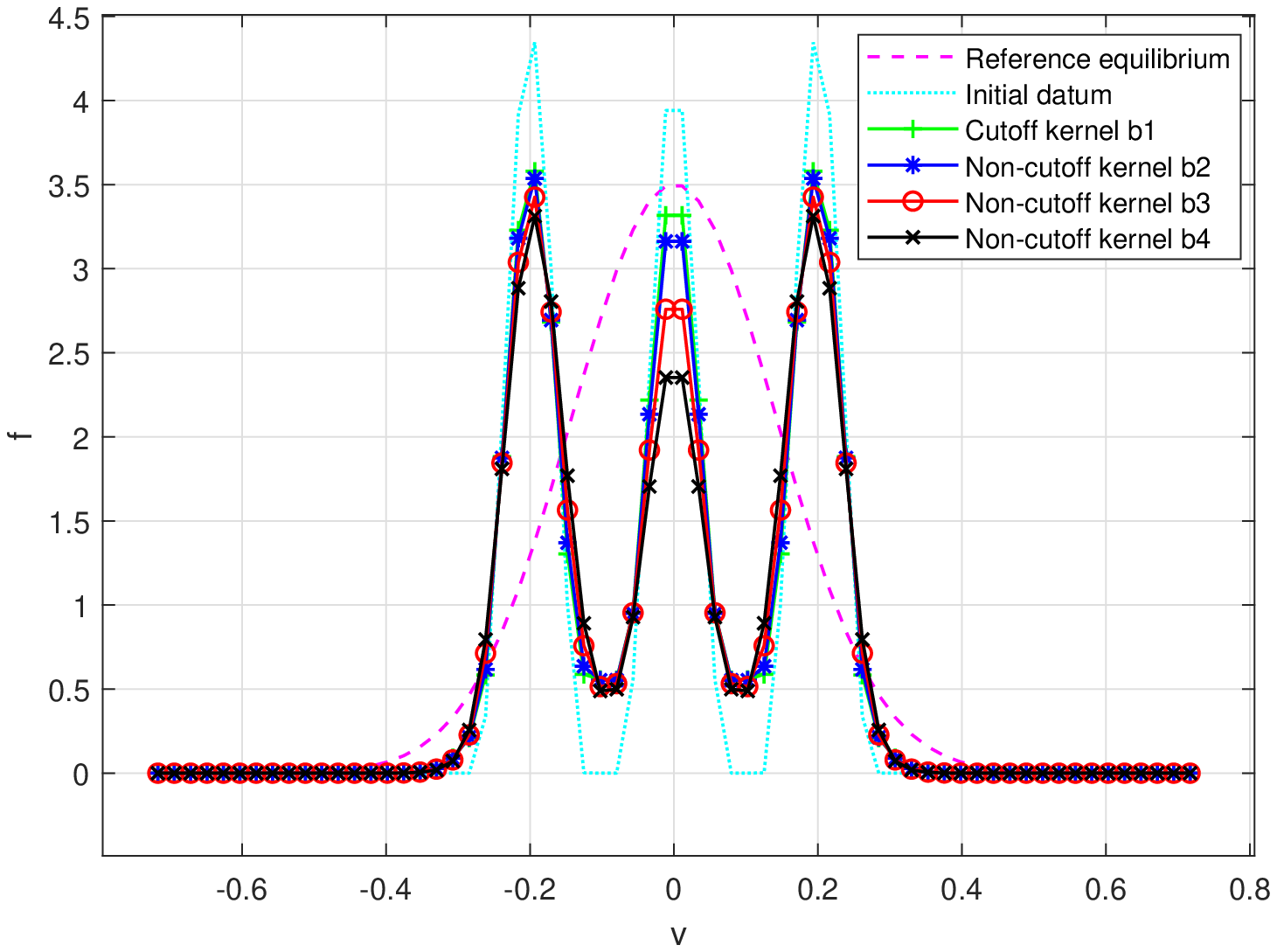}
	}
	\subfigure[t=3]{
		\includegraphics[width=6.5cm]{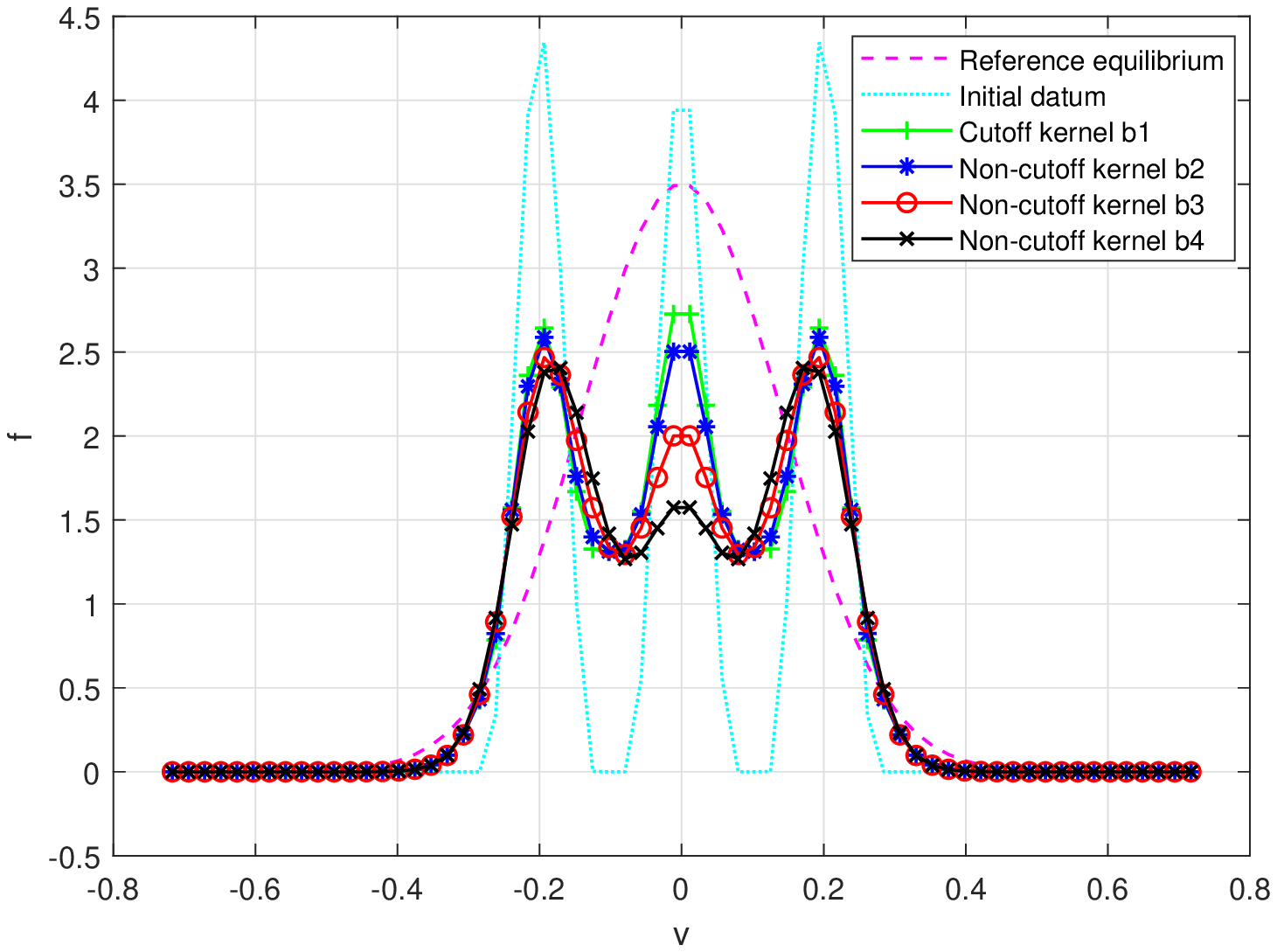}
	}\\
	\subfigure[t=6]{
		\includegraphics[width=6.5cm]{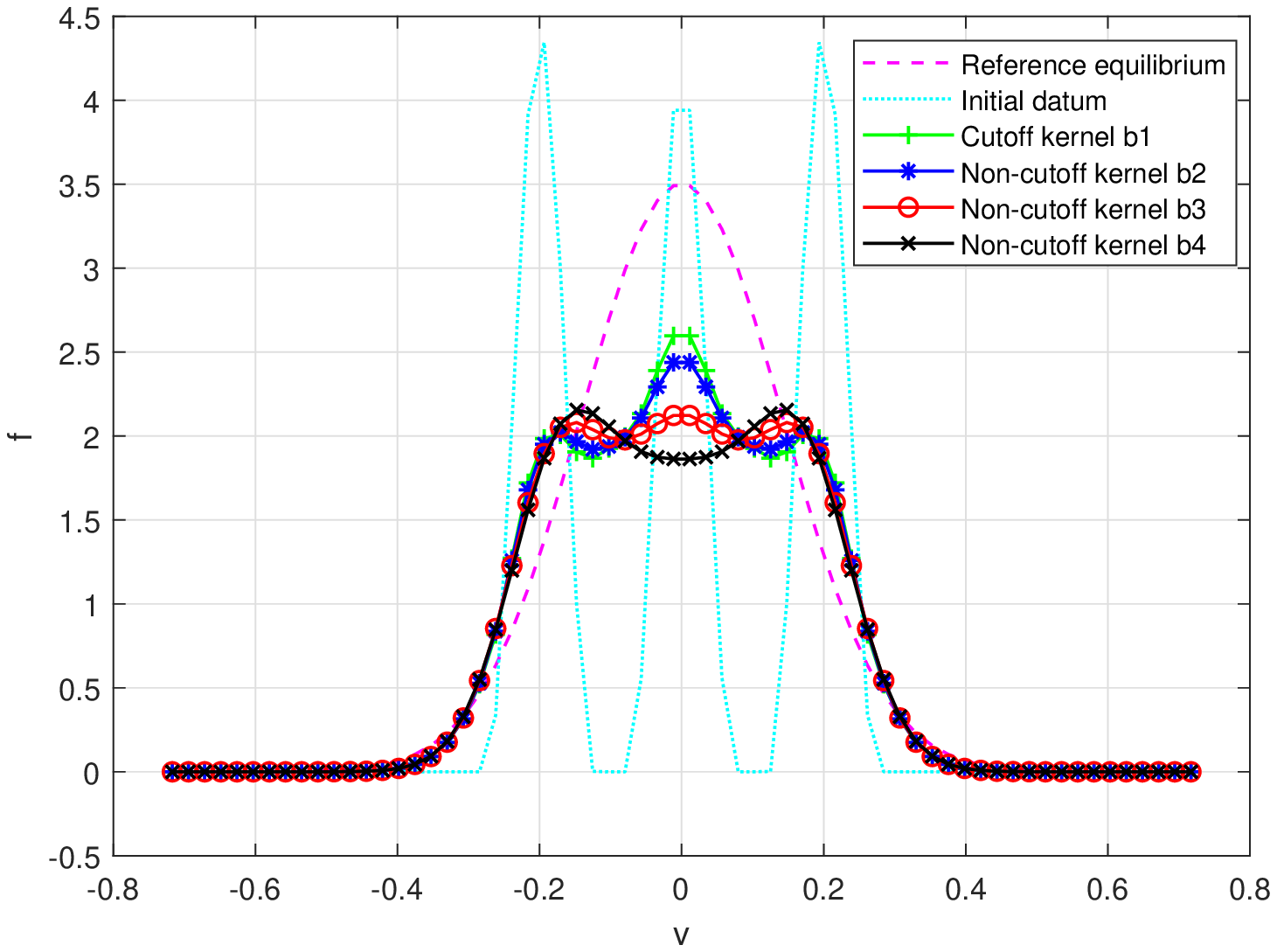}
	}
	\subfigure[t=9]{
		\includegraphics[width=6.5cm]{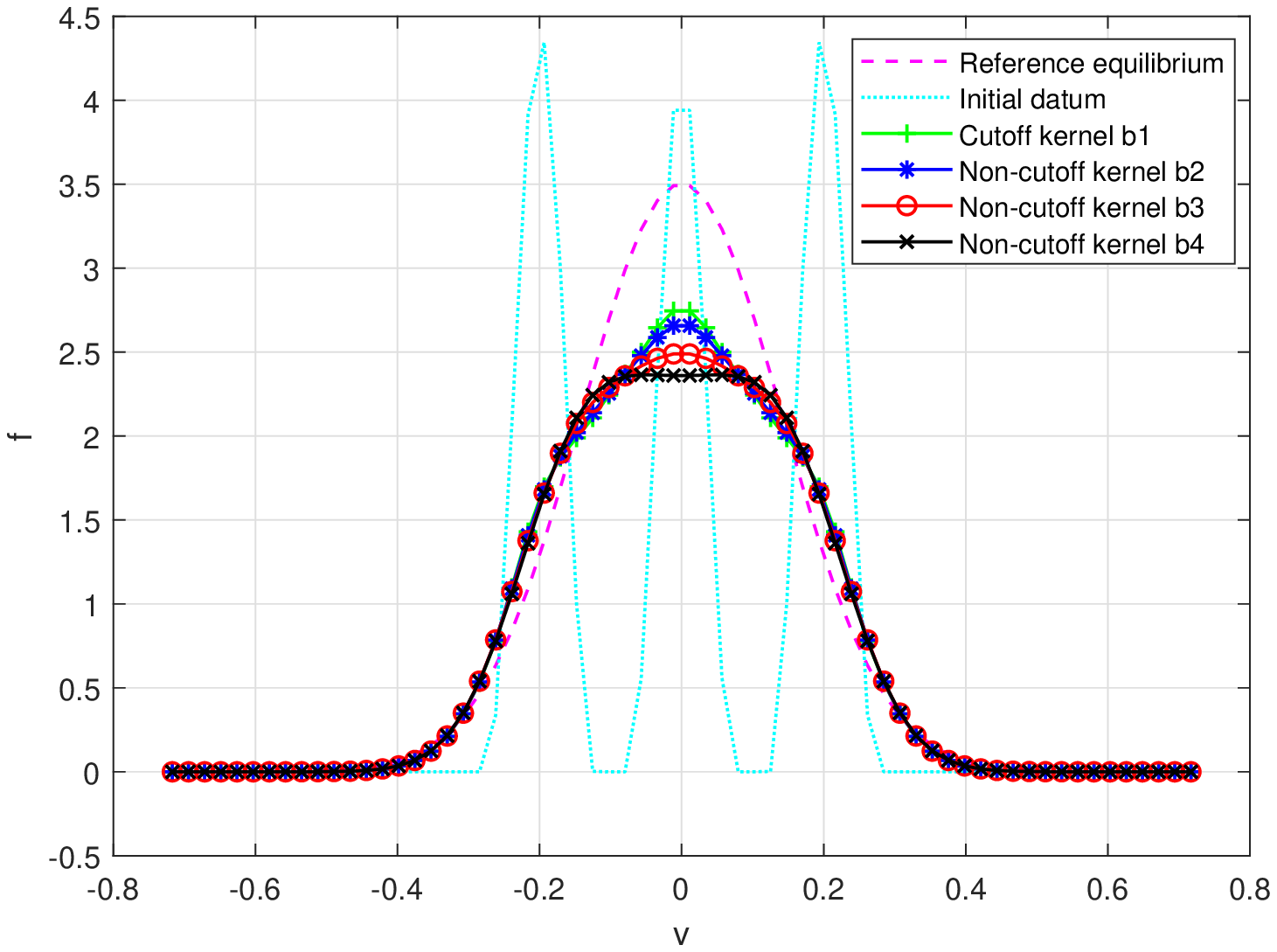}
	}\\
	\subfigure[t=12]{
		\includegraphics[width=6.5cm]{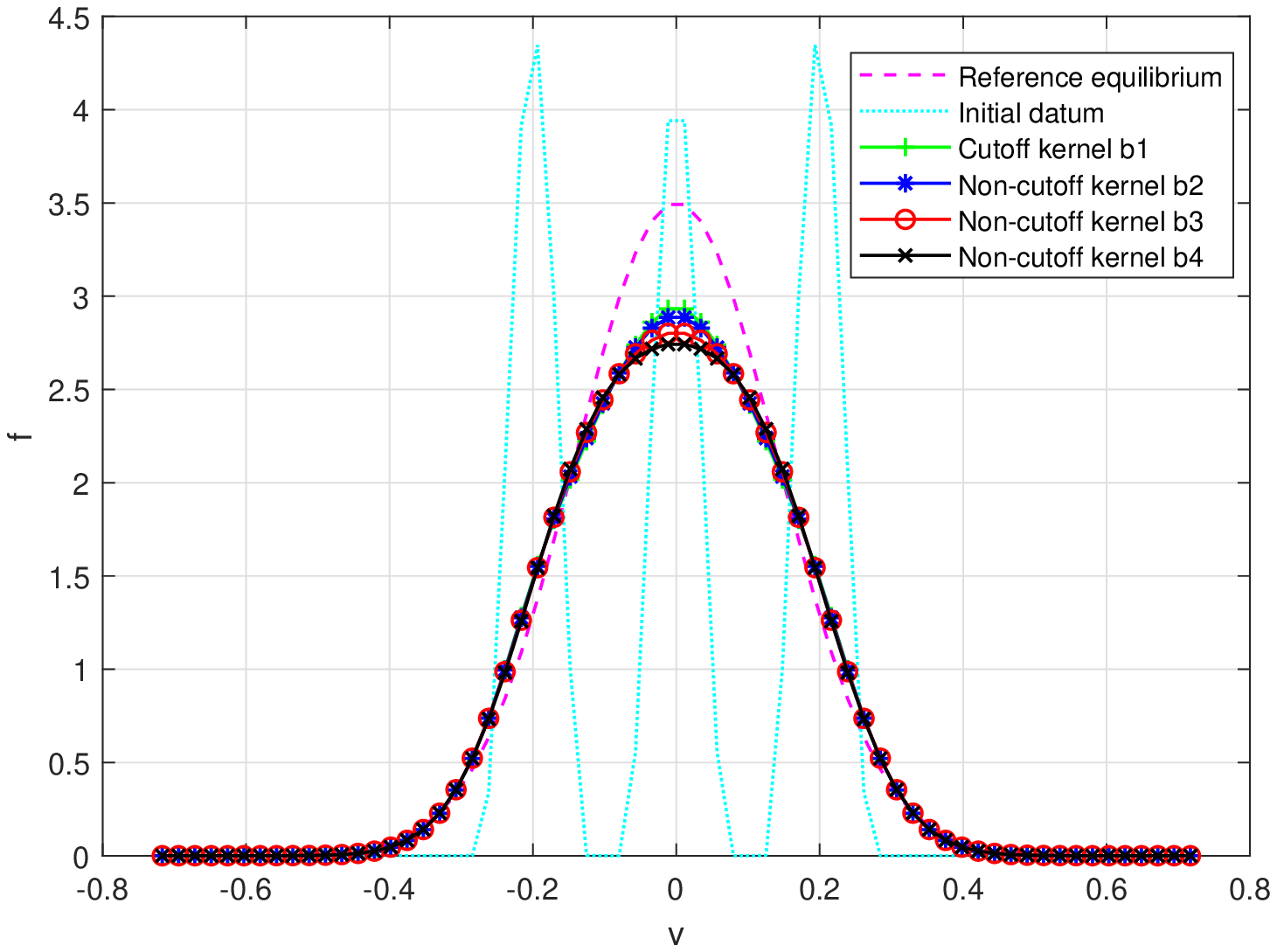}
	}
	\subfigure[t=15]{
		\includegraphics[width=6.5cm]{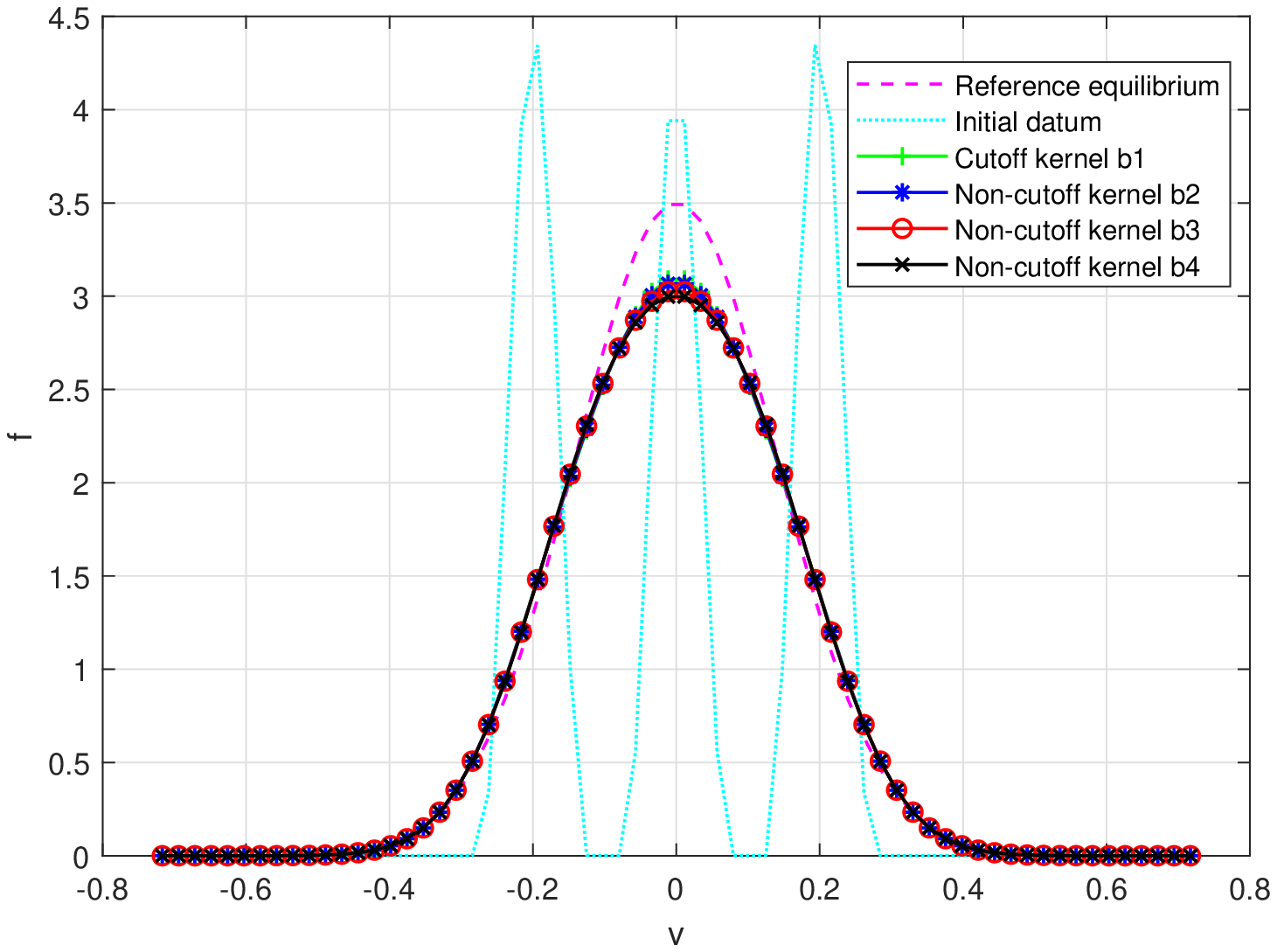}
	}\\
	\subfigure[t=20]{
		\includegraphics[width=6.5cm]{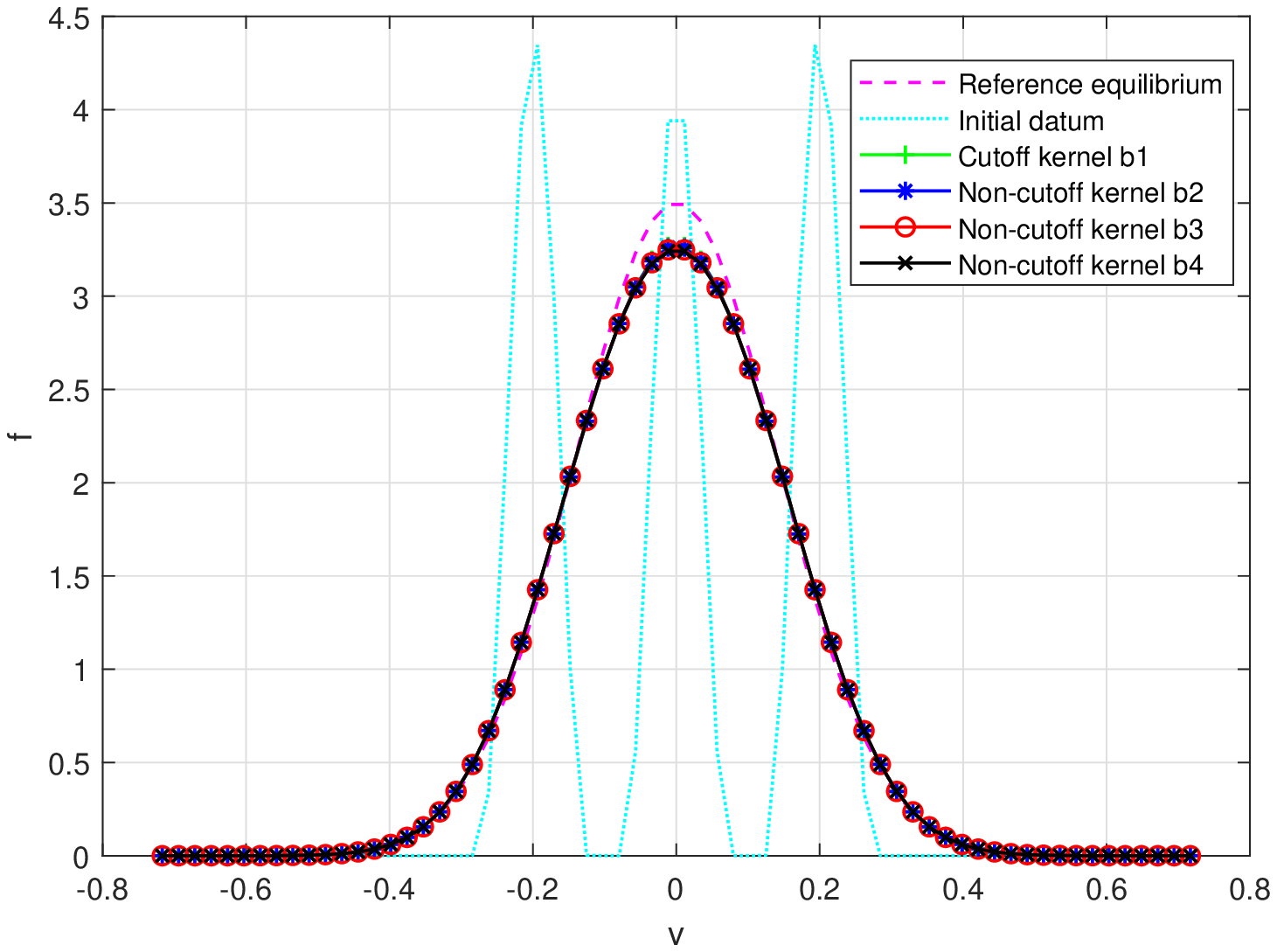}
	}
	\subfigure[t=30]{
		\includegraphics[width=6.5cm]{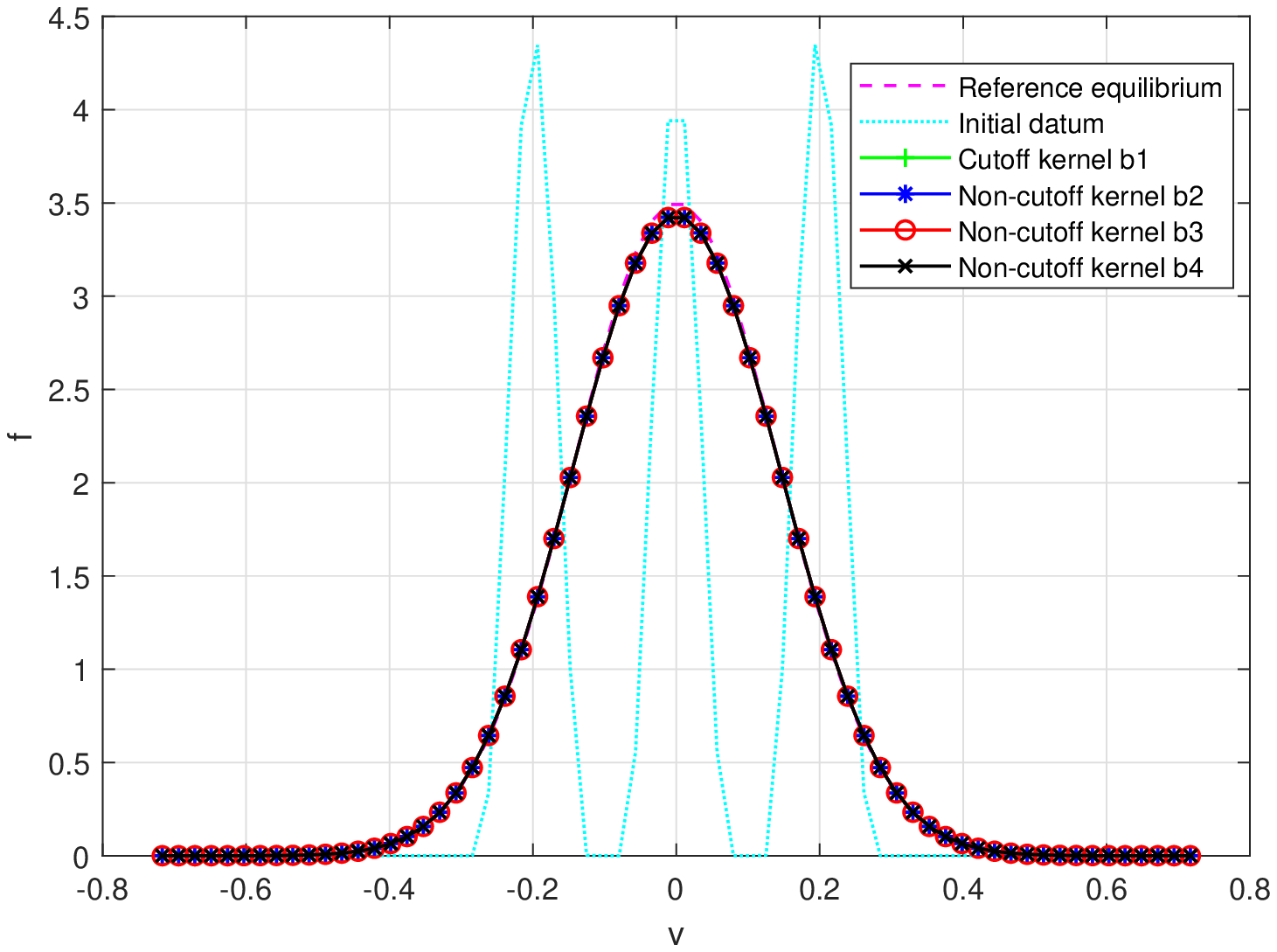}
	}
	\caption{ Section~\ref{2measure}: Measure valued solution in 2D -- Maxwell molecule. Time evolution of the distribution function $f$ (a slice of the solution along $v_1$ with $v_2=0$) computed with cutoff kernel $b_1$ and non-cutoff kernels $b_{2}$, $b_{3}$ and $b_{4}$. Initial condition given by (\ref{initialdelta}). Classical RK4 with $ \Delta t = 0.05 $ for time discretization. $N=N_{|q|}=64 $, $N_{\q}=32$. $R=0.66$, $L=(3+\sqrt{2})R/4\approx0.73$.}
	\label{fig3}
\end{figure}

\subsection{Measure valued solution in 3D -- Maxwell molecule}
\label{3measure}

We now perform a similar test as the last subsection using four 3D kernels $b_5$ (\ref{b5}), $b_6$ (\ref{b6}), $b_7$ (\ref{b7}), and $b_8$ (\ref{b8}). The results are gathered in Figure~\ref{fig4}, where similar behavior as in 2D is observed. 

\begin{figure}[htp]
	\centering
	\subfigure[t=1]{
		\includegraphics[width=6.5cm]{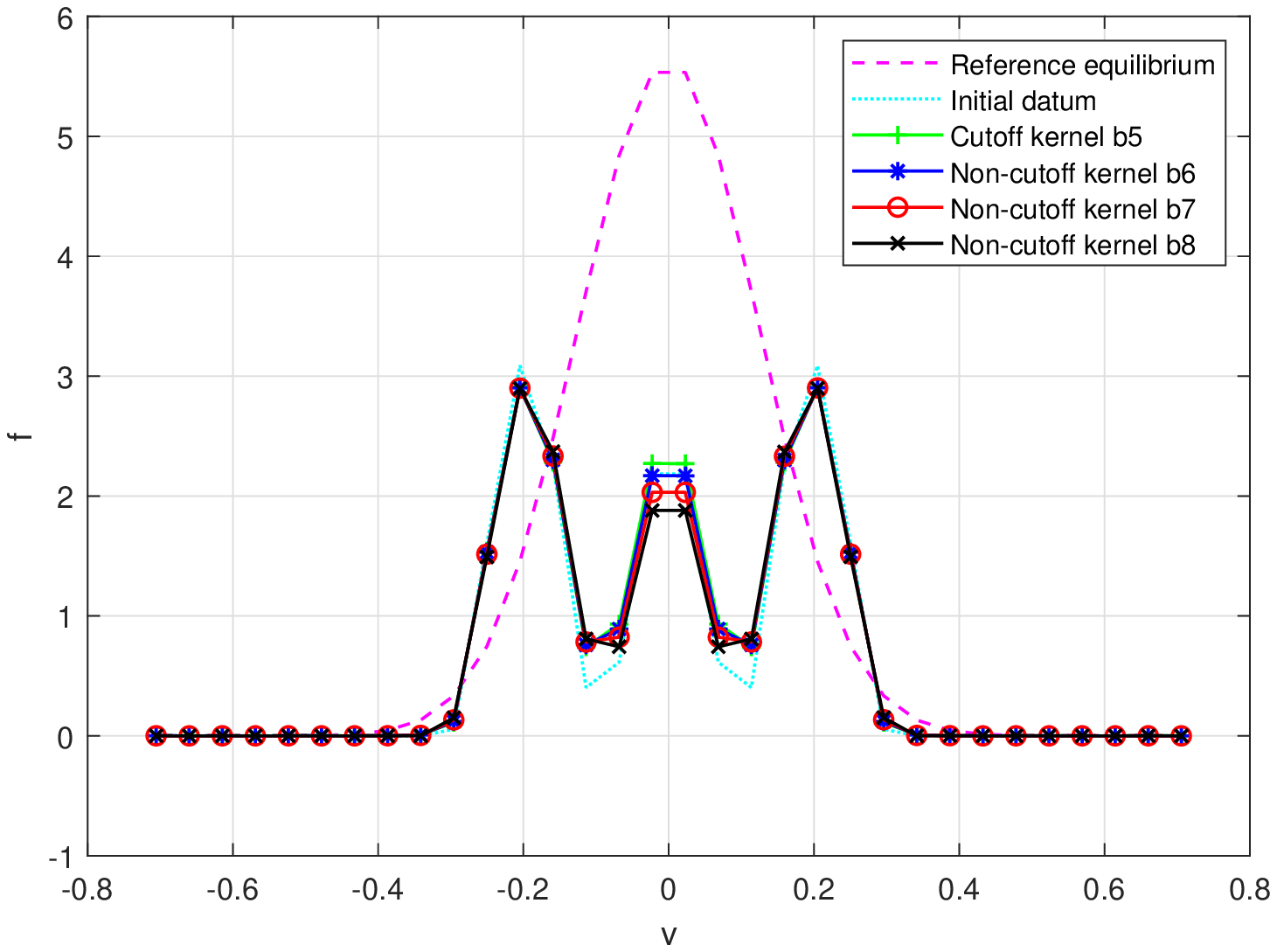}
	}
	\subfigure[t=3]{
		\includegraphics[width=6.5cm]{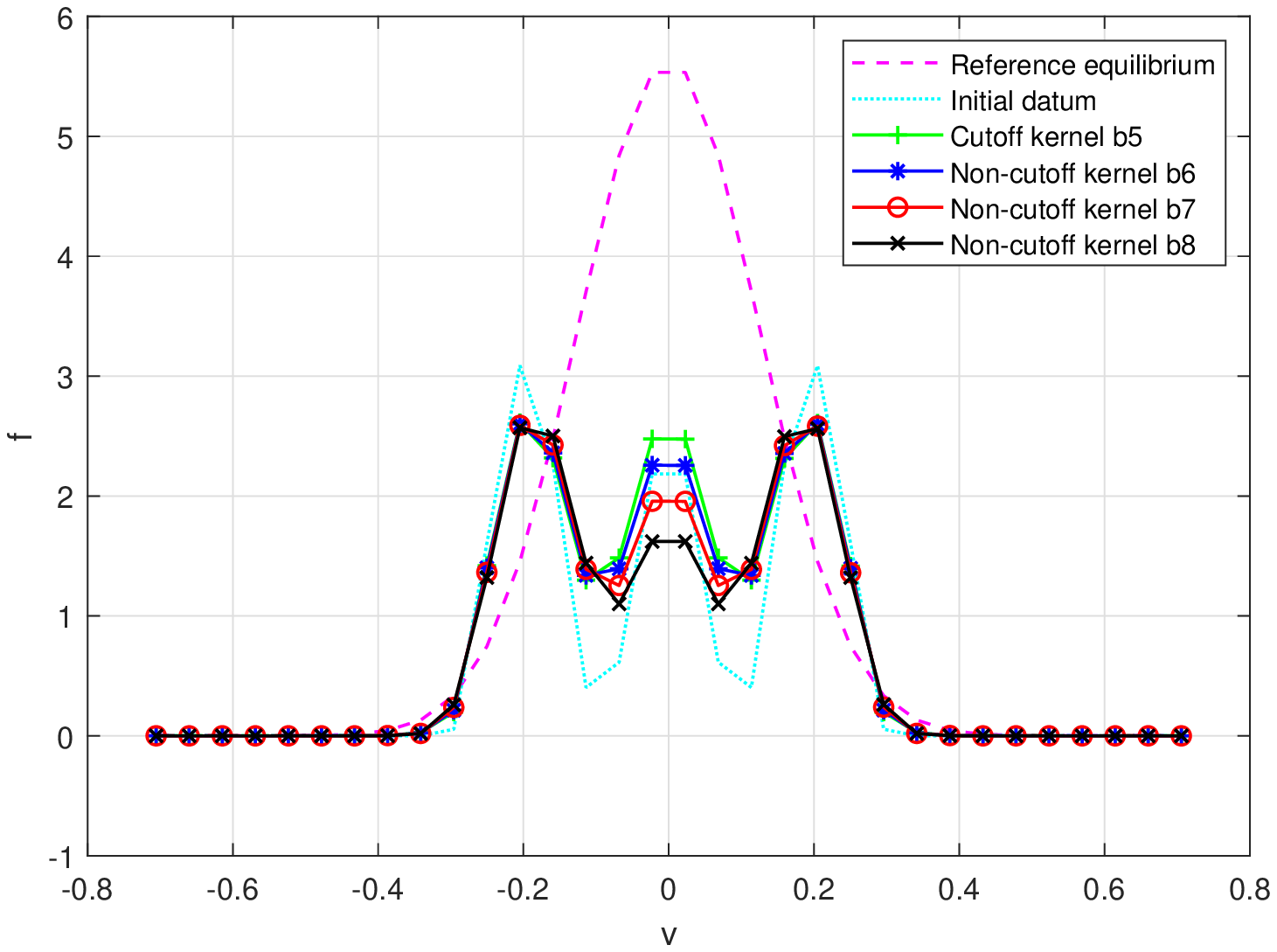}
	}\\
	\subfigure[t=6]{
		\includegraphics[width=6.5cm]{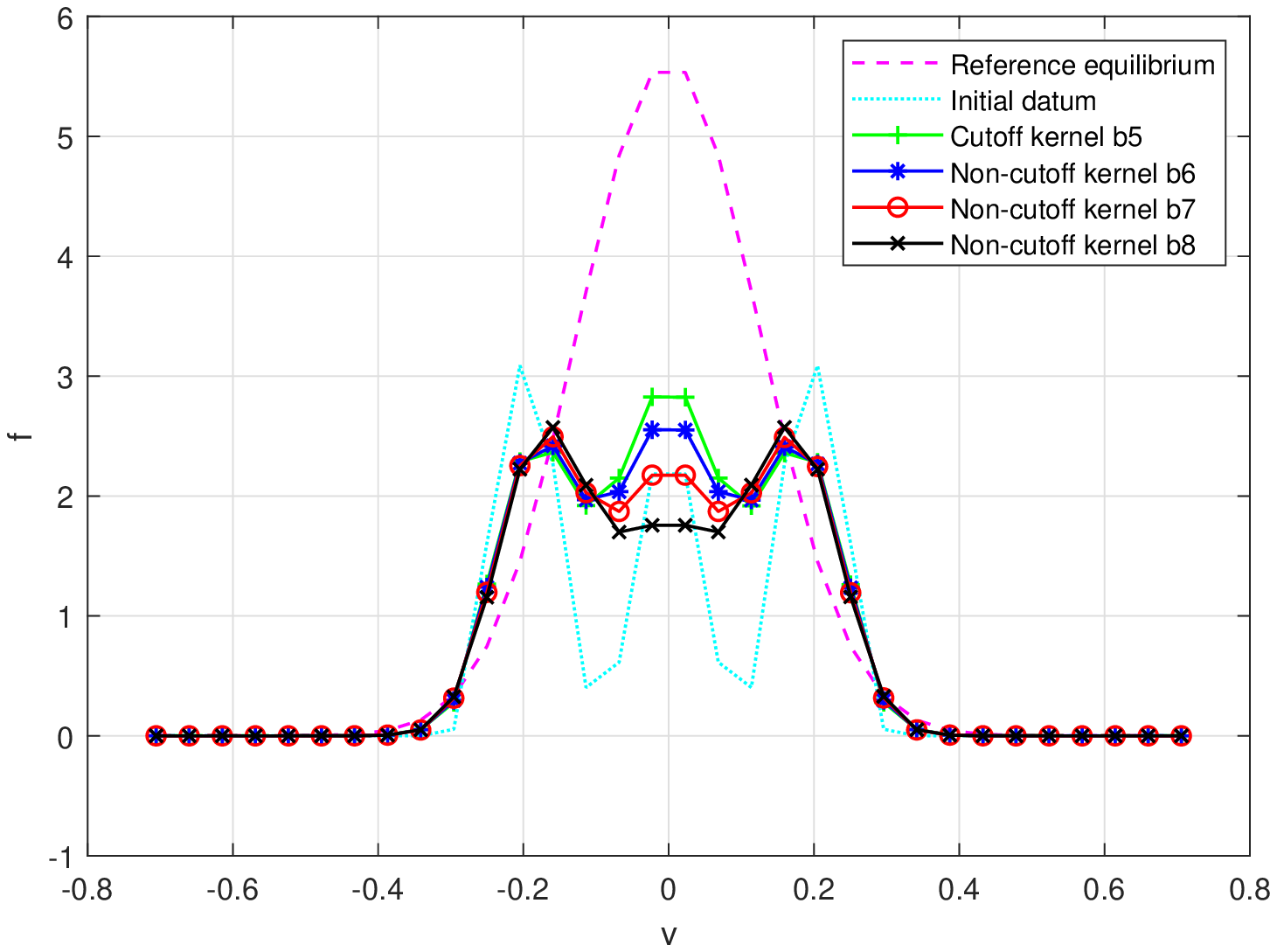}
	}
	\subfigure[t=9]{
		\includegraphics[width=6.5cm]{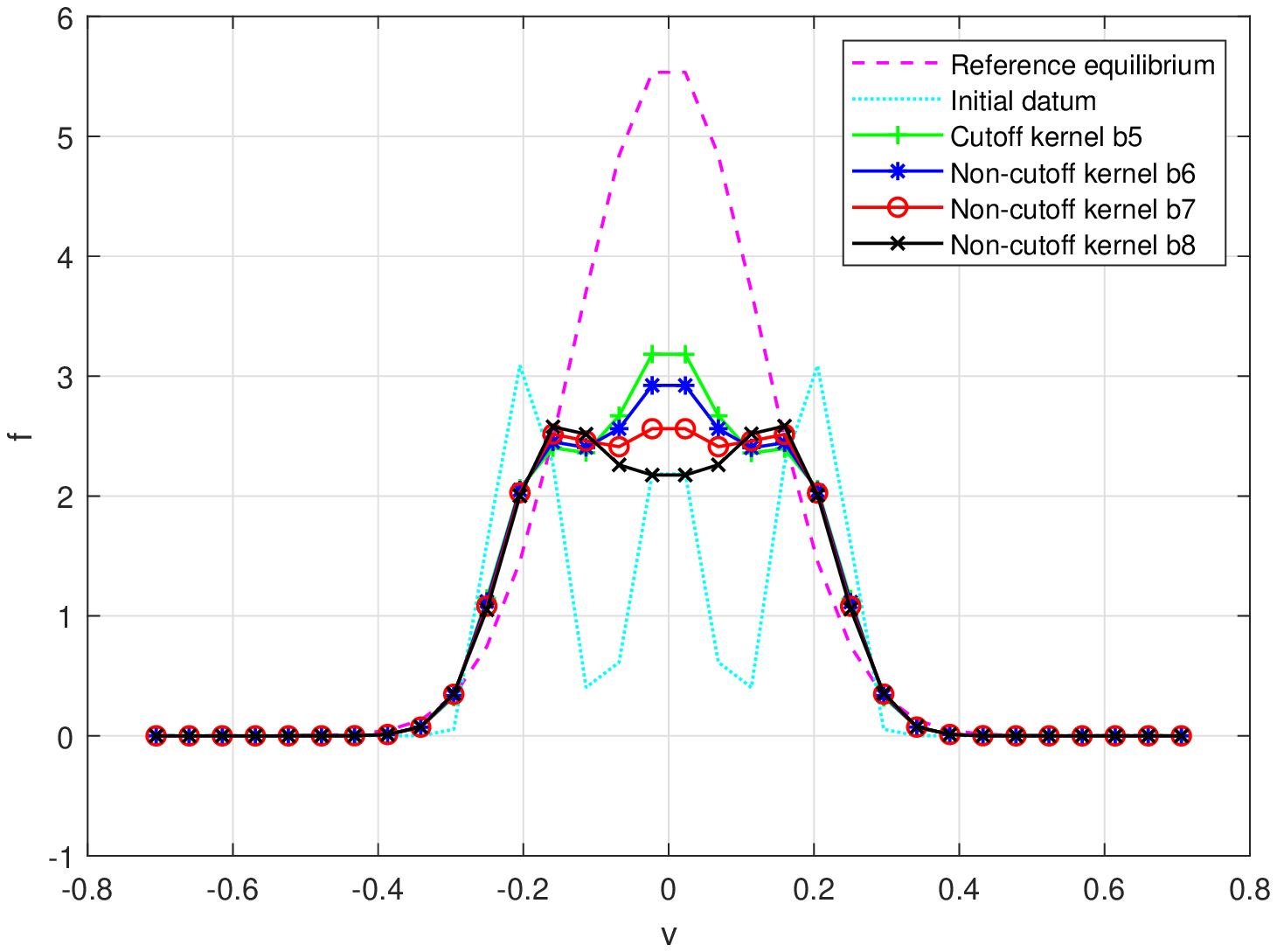}
	}\\
	\subfigure[t=12]{
		\includegraphics[width=6.5cm]{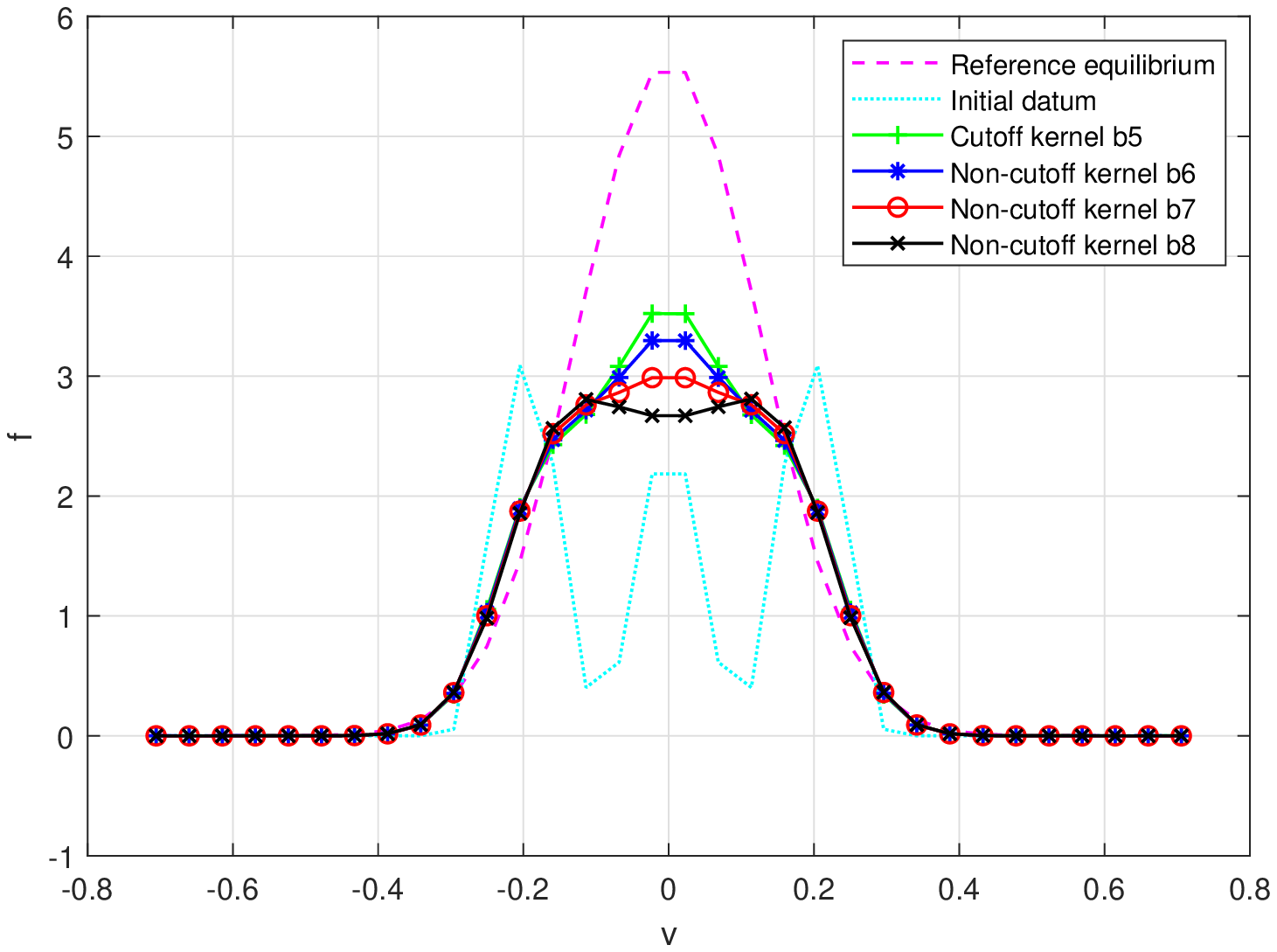}	
	}
	\subfigure[t=15]{
		\includegraphics[width=6.5cm]{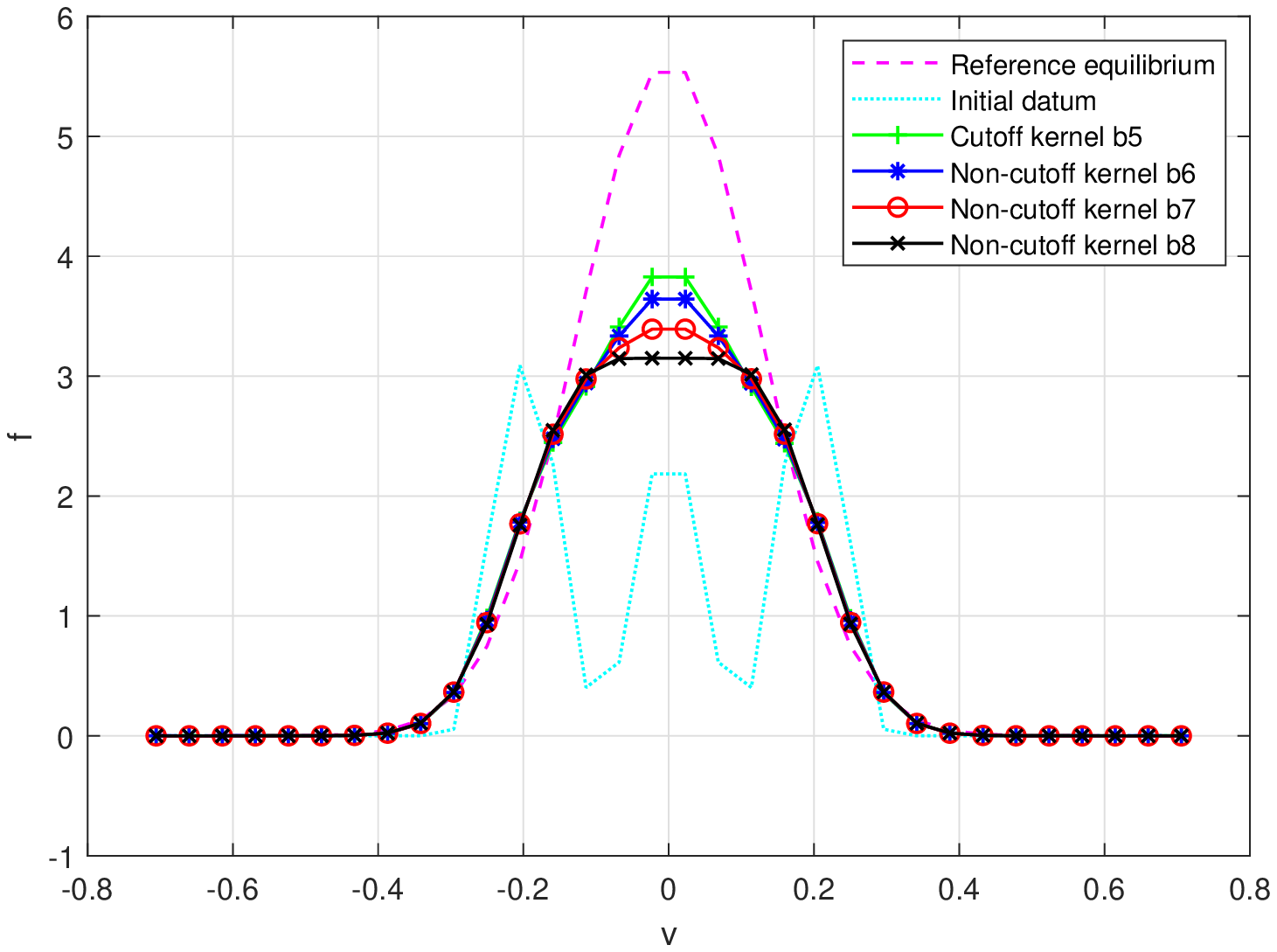}
	}\\
	\subfigure[t=30]{
	       \includegraphics[width=6.5cm]{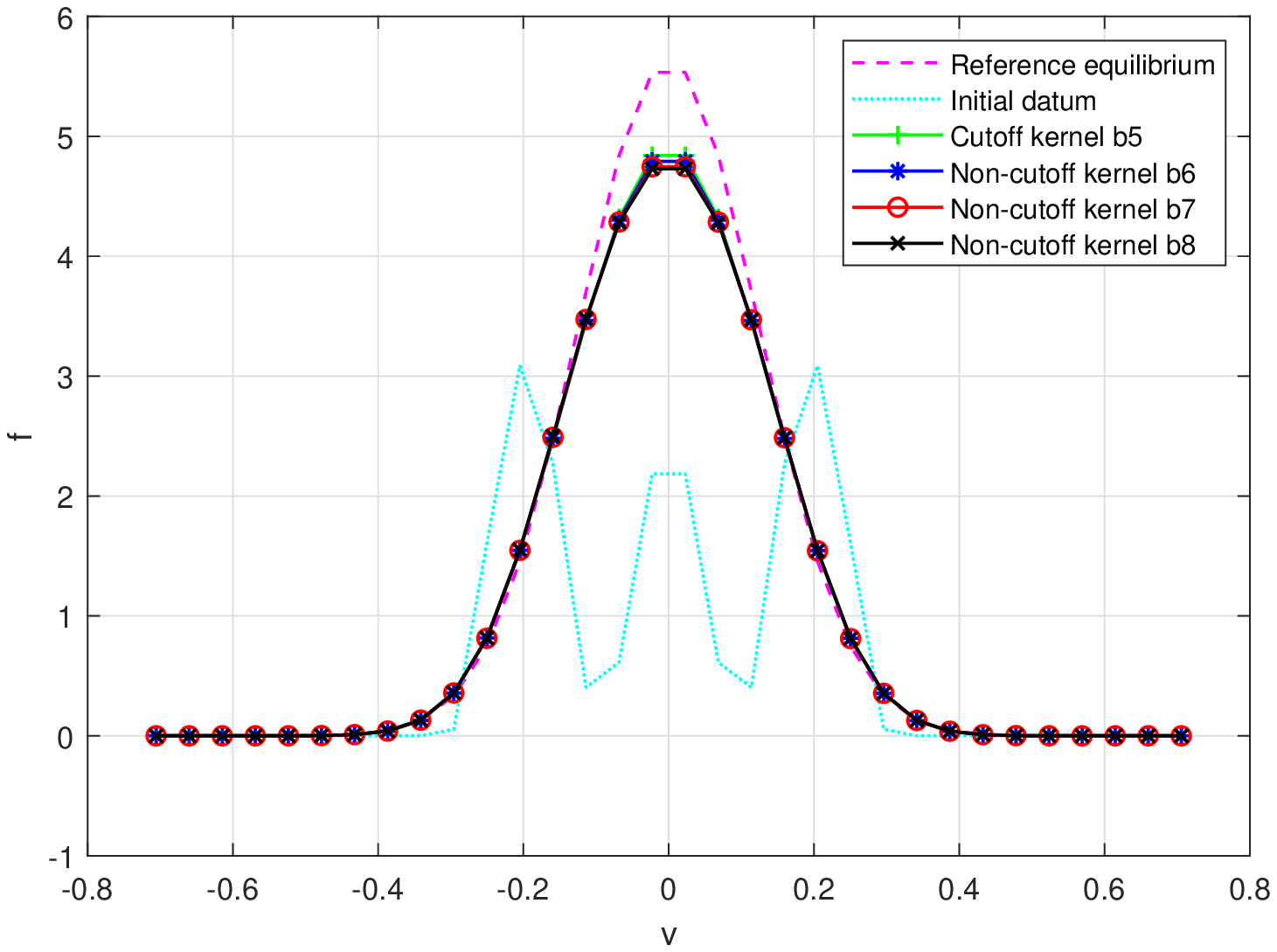}
       }
	\subfigure[t=50]{
	       \includegraphics[width=6.5cm]{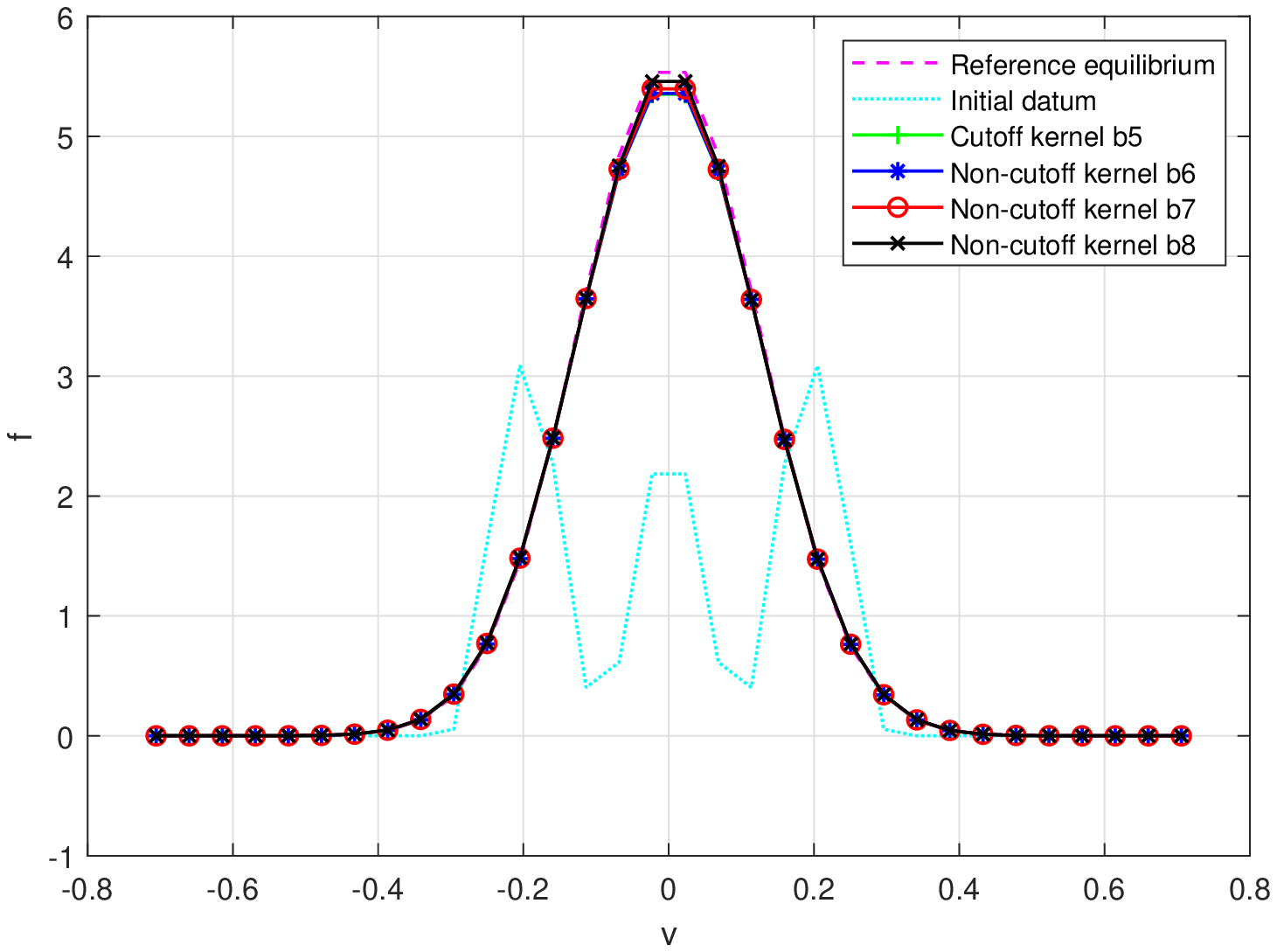}
       }
	\caption{ Section~\ref{3measure}: Measure valued solution in 3D -- Maxwell molecule. Time evolution of the distribution function $f$ (a slice of the solution along $v_1$ with $v_2=v_3=0$) computed with cutoff kernel $b_5$ and non-cutoff kernels $b_{6}$, $b_{7}$ and $b_{8}$. Initial condition given by (\ref{initialdelta}). Classical RK4 with $ \Delta t = 0.2 $ for time discretization. $N=N_{|q|}=N_{\q}=32$. $R=0.66$, $L=(3+\sqrt{2})R/4\approx0.73$.}
	\label{fig4}
\end{figure}

\subsection{Measure valued solution in 3D -- Debye-Yukawa kernel}
\label{4debye}

We then consider a more physically relevant collision kernel resulting from the Debye-Yukawa potential:
\begin{equation}\label{ApproximatedDebye}
\sin\theta B\left( |v-v_{*}|, \cos\theta \right) = \frac{1}{2\sin\frac{\theta}{2}} |v-v_{*}| \left|\log\left(2\sin\frac{\theta}{2}\right)^{-1}\right|, \quad \theta\in\left[0,\pi\right],
\end{equation}
which has a limiting singularity behavior as (\ref{debye}) when $\theta\rightarrow 0$. Note that this kernel contains a velocity dependence similar to hard spheres. As a comparison, we also consider a cutoff version of the kernel:
\begin{equation} \label{ApproximatedDebyecutoff}
\sin\theta B_{\text{cutoff}}\left( |v-v_{*}|, \cos\theta \right) =
\begin{cases}
0, & \theta\in\left[0,\frac{\pi}{10}\right],\\
\frac{1}{2\sin\frac{\theta}{2}} |v-v_{*}| \left|\log\left(2\sin\frac{\theta}{2}\right)^{-1}\right|, \quad & \theta\in\left[\frac{\pi}{10},\pi\right],
\end{cases}
\end{equation} 

Figure~\ref{fig5} shows the results obtained with the above two kernels subject to initial condition
\begin{equation} \label{initialdelta1}
f^0(v) = \frac{1}{2} \delta_{w}(|v|-0.2),
\end{equation}
where $ \delta_{w}(v) $ is given by \eqref{delta}. The difference of solutions in the cutoff case and non-cutoff case is obvious.

\begin{figure}[htp]
	\centering
	\subfigure[t=1]{
		\includegraphics[width=6.5cm]{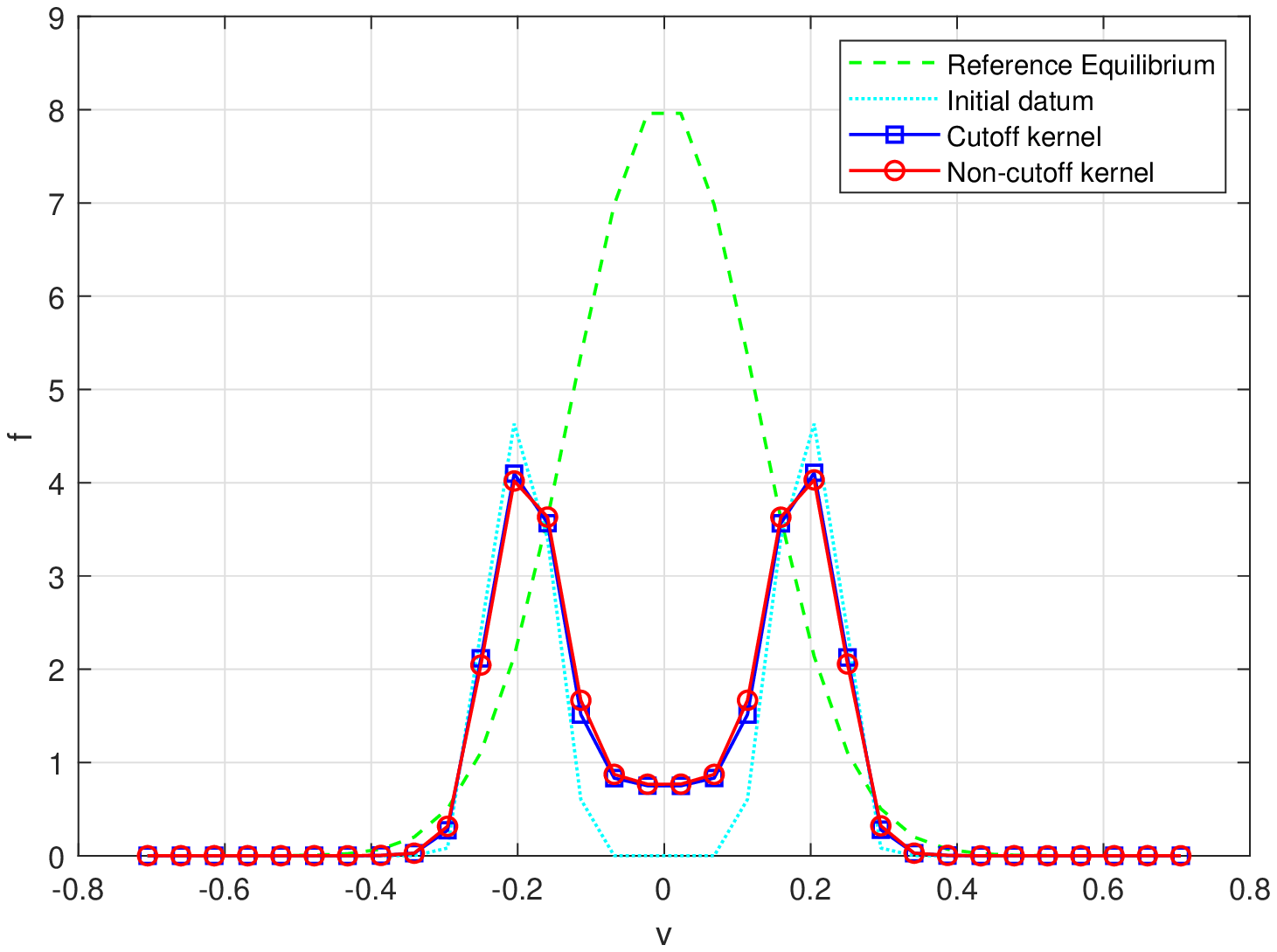}
	}
	\subfigure[t=3]{
		\includegraphics[width=6.5cm]{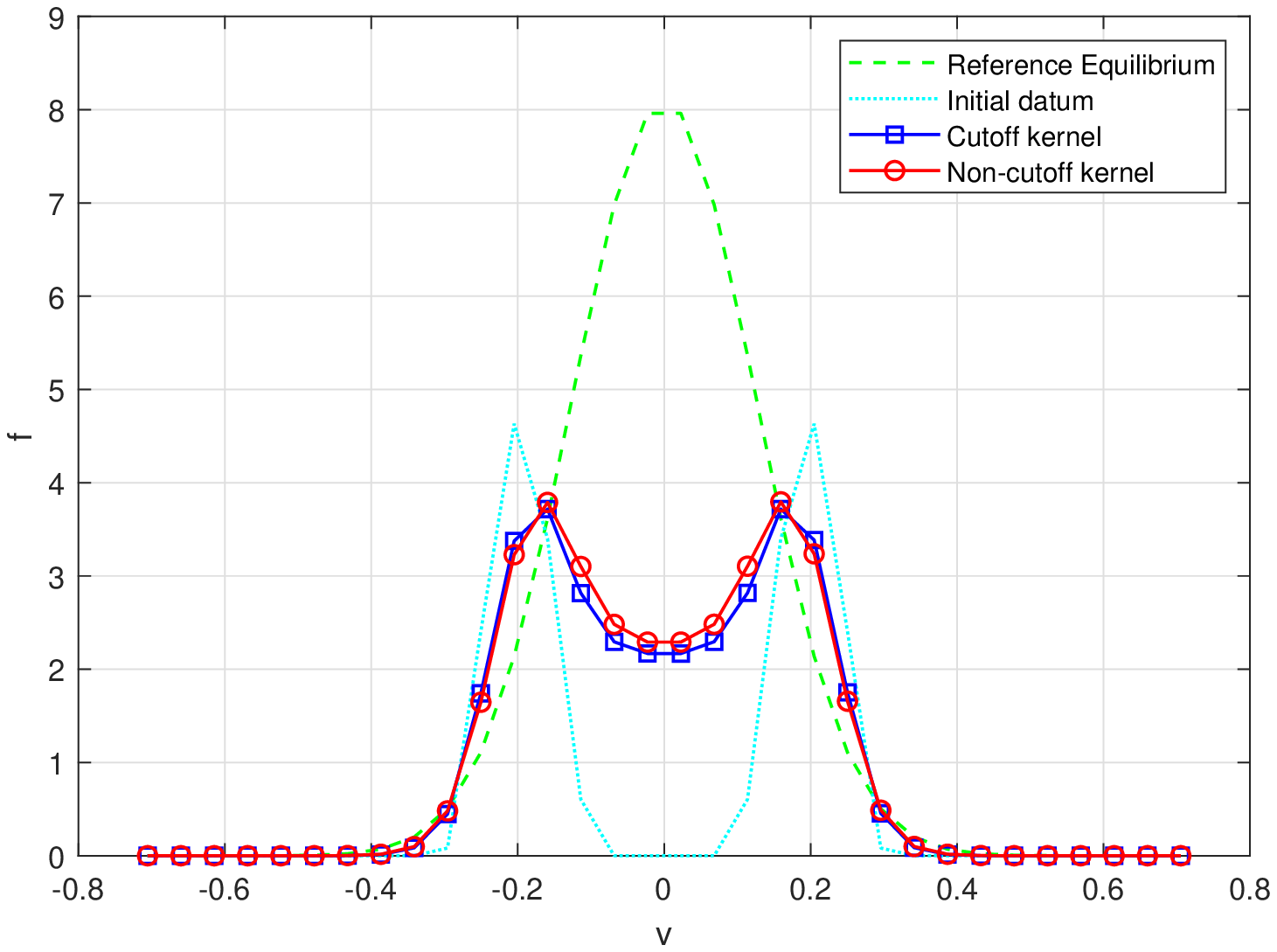}
	}\\
	\subfigure[t=5]{
		\includegraphics[width=6.5cm]{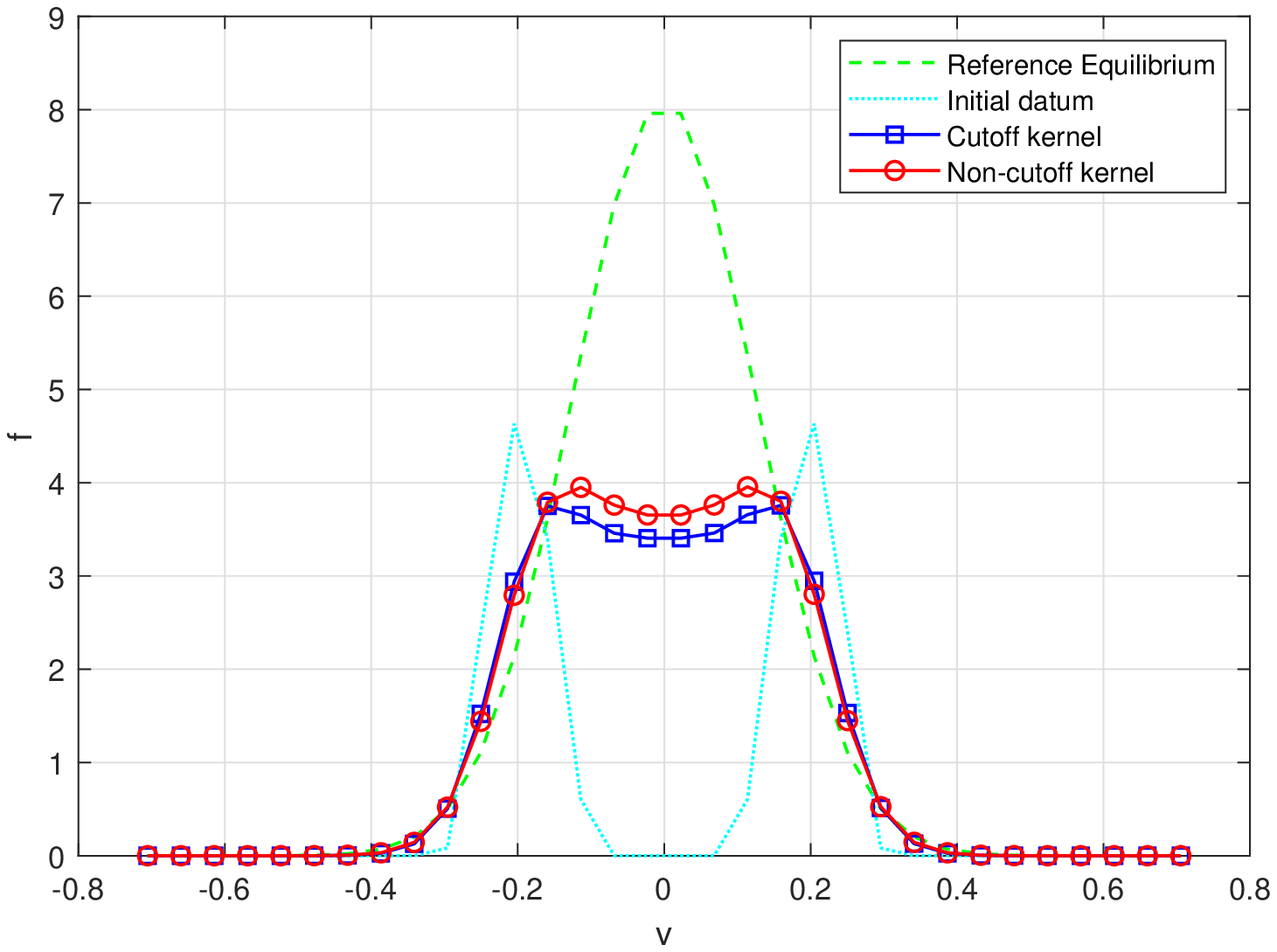}
	}
	\subfigure[t=8]{
		\includegraphics[width=6.5cm]{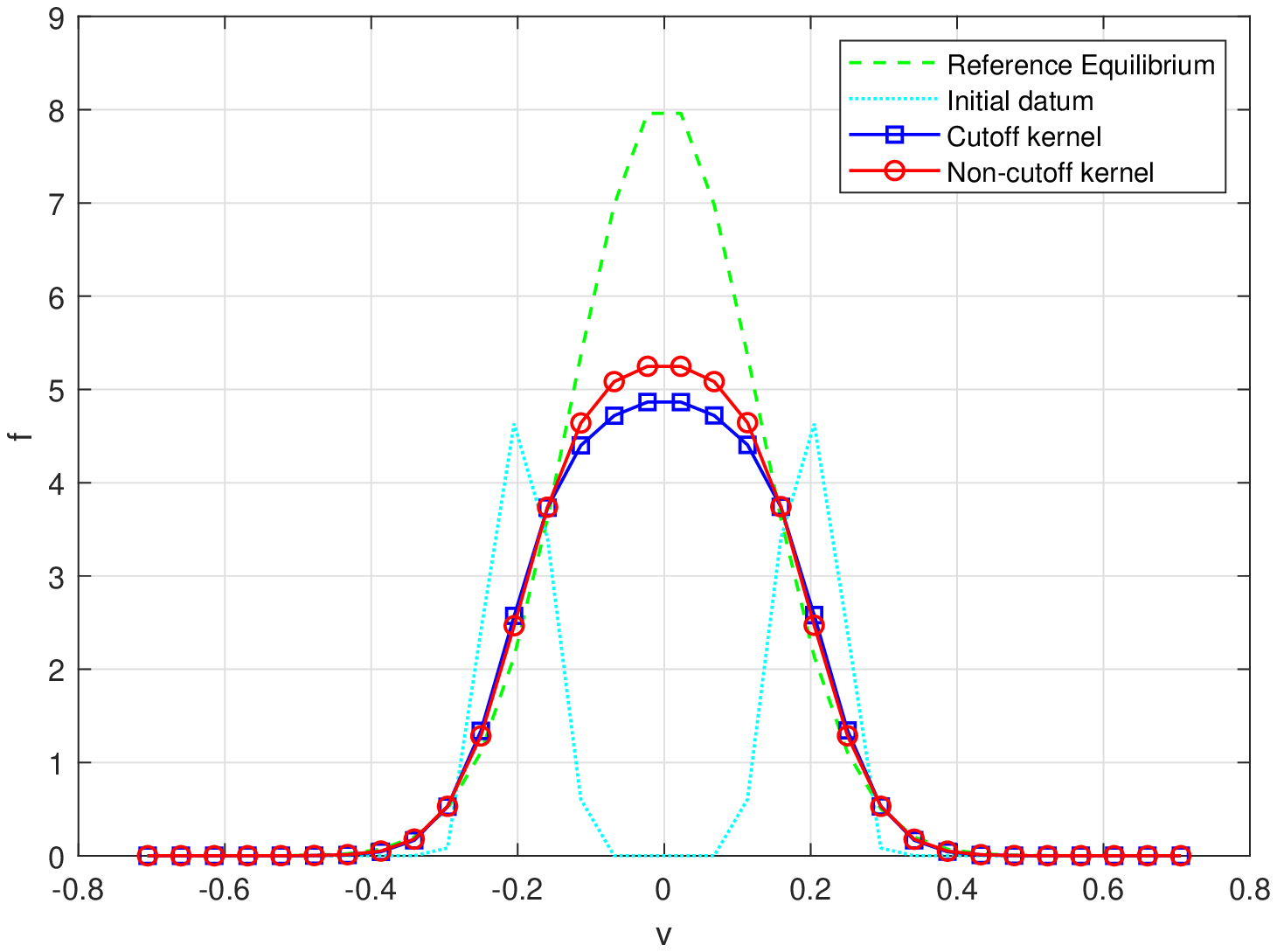}
	}\\
	\subfigure[t=10]{
		\includegraphics[width=6.5cm]{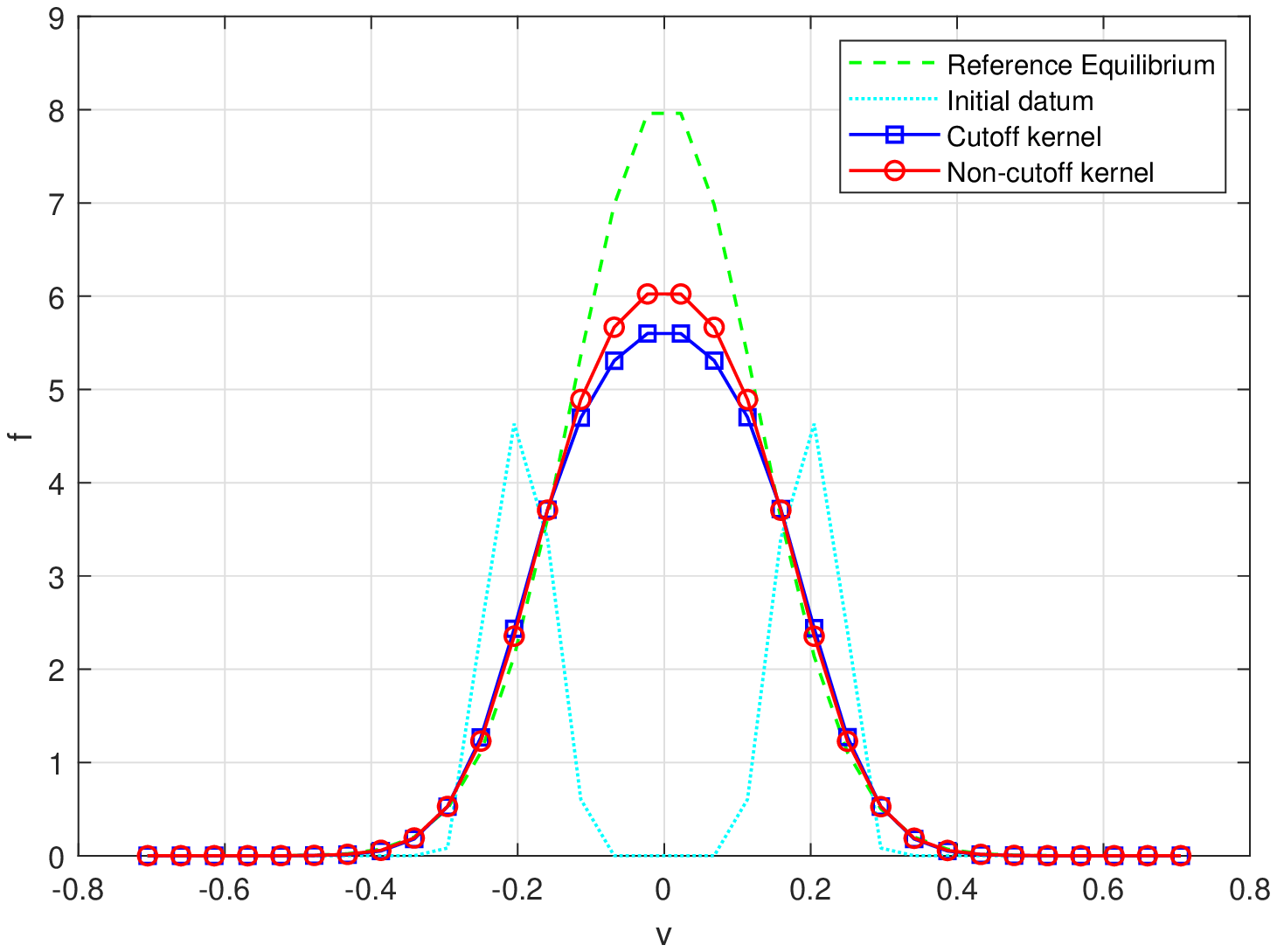}
	}
	\subfigure[t=15]{
		\includegraphics[width=6.5cm]{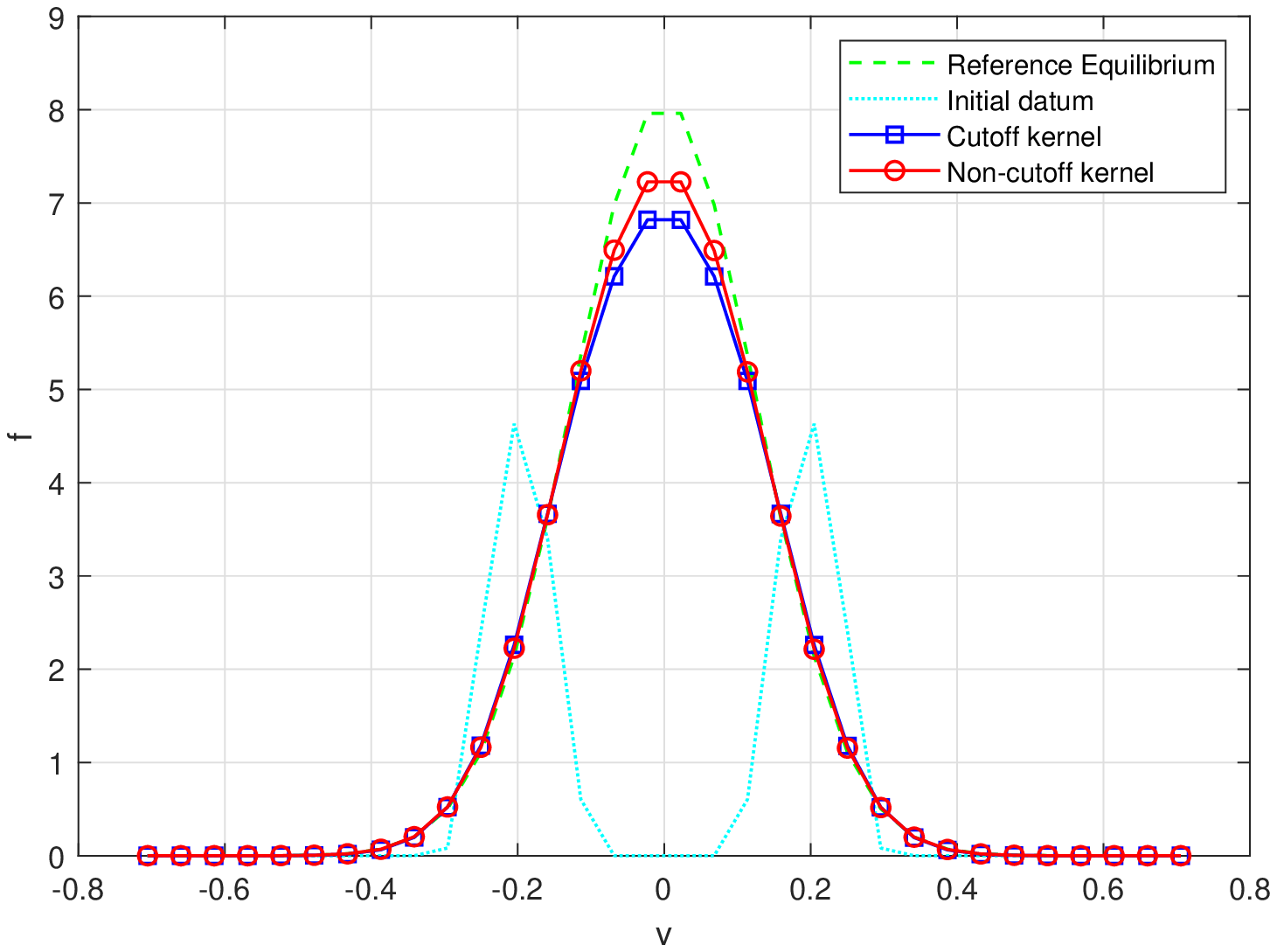}
	}\\
       \subfigure[t=20]{
	       \includegraphics[width=6.5cm]{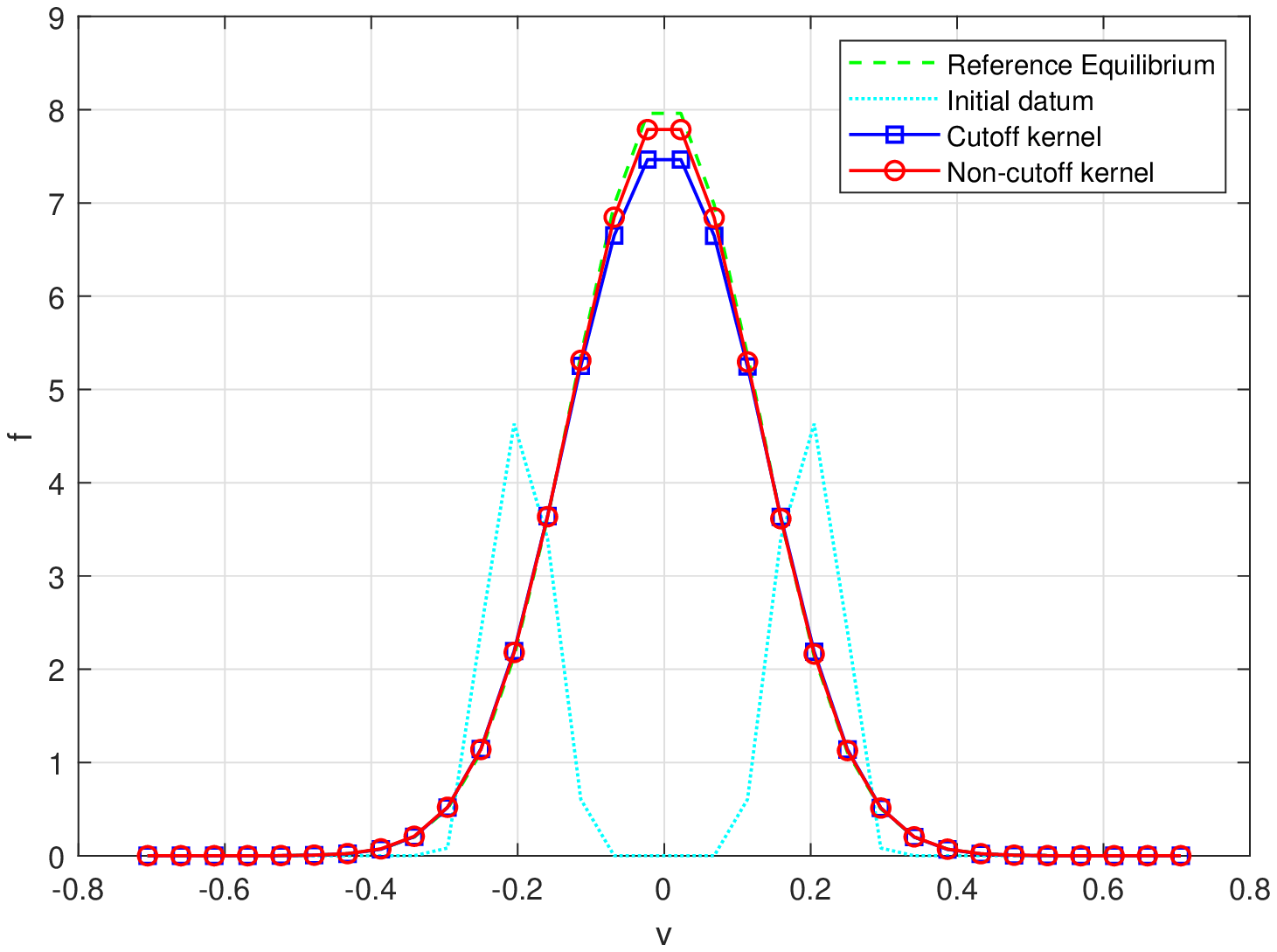}
       }
       \subfigure[t=40]{
              \includegraphics[width=6.5cm]{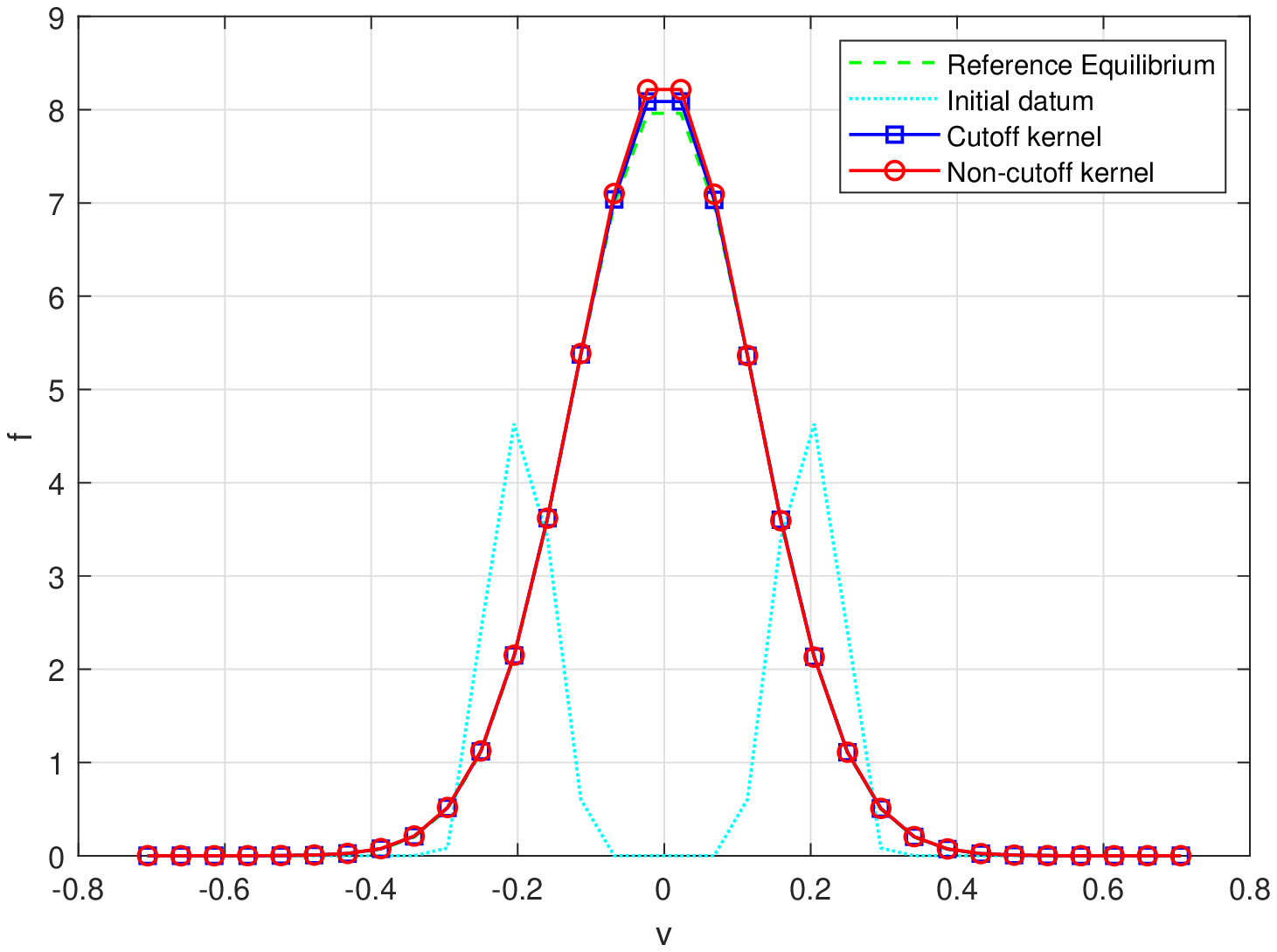}
       }
	\caption{ Section~\ref{4debye}: Measure valued solution in 3D -- Debye-Yukawa kernel. Time evolution of the distribution function $f$ (a slice of the solution along $v_1$ with $v_2=v_3=0$) computed with the Debye-Yukawa kernel (\ref{ApproximatedDebye}) and its cutoff version (\ref{ApproximatedDebyecutoff}). Initial condition given by (\ref{initialdelta1}). Classical RK4 with $ \Delta t = 0.05 $ for time discretization. $N=N_{|q|}=N_{\q}=32$. $R=0.66$, $L=(3+\sqrt{2})R/4\approx0.73$.}
	\label{fig5}
\end{figure}


\subsection{Discontinuous solution in 2D -- Maxwell molecule}
\label{2dis}
In this final test, we consider the following discontinuous initial data in 2D:
\begin{equation}\label{disinitial}
f^{0}(v)=
\begin{cases}
\frac{\rho_{1}}{2\pi T_1}\exp\left(-\frac{|v|^2}{2T_{1}}\right),\quad \text{for}\quad v_{1}>0,\\
\frac{\rho_{2}}{2\pi T_2}\exp\left(-\frac{|v|^2}{2T_{2}}\right),\quad \text{for}\quad v_{1}<0,
\end{cases}
\end{equation}
where we pick $\rho_{1}=\frac{6}{5}$, $\rho_{2}=\frac{4}{5}$, $T_{1}=\frac{2}{3}$, $T_{2}=\frac{3}{2}$ such that
\begin{equation}
\int_{\mathbb{R}^{2}} f^{0} \,\rd v = \frac{1}{2} \int_{\mathbb{R}^{2}} f^{0} |v|^2\,\rd v = 1,\quad \int_{\mathbb{R}^{2}} f^{0} v\,\rd v = 0,
\end{equation} 
which leads the normalized Gaussian distribution as the equilibrium:
\begin{equation}
f_{\text{ref}}(v) = \frac{1}{2\pi }\exp \left(-\frac{|v|^2}{2}\right).
\end{equation}

To clearly tell the difference in the smoothing effect between cutoff and non-cutoff kernels, we compare the non-cutoff kernel $ b_{3} $ with its corresponding cutoff version:
\begin{equation} \label{b3cutoff}
b_3^{\text{cutoff}}\left( \cos\theta \right) =
\begin{cases}
0, & \theta \in \left[0,\theta_{0}\right) \cup \left(2\pi - \theta_{0},2\pi\right] ,\\
\frac{1}{8\pi\sin^{2}\frac{\theta}{2}}, \quad & \theta \in \left[\theta_{0},2\pi - \theta_{0} \right],
\end{cases}
\end{equation}
where we choose $ \theta_{0} =\pi/4$ and $\pi/10$ respectively. Figure \ref{2DDis} shows the time evolution of the solutions, where we can see that the solution with the non-cutoff kernel is smoothed out more quickly than that with the cutoff kernels.
\begin{figure}[htp]
	\centering
	\subfigure[t=0.5]{
		\includegraphics[width=6.5cm]{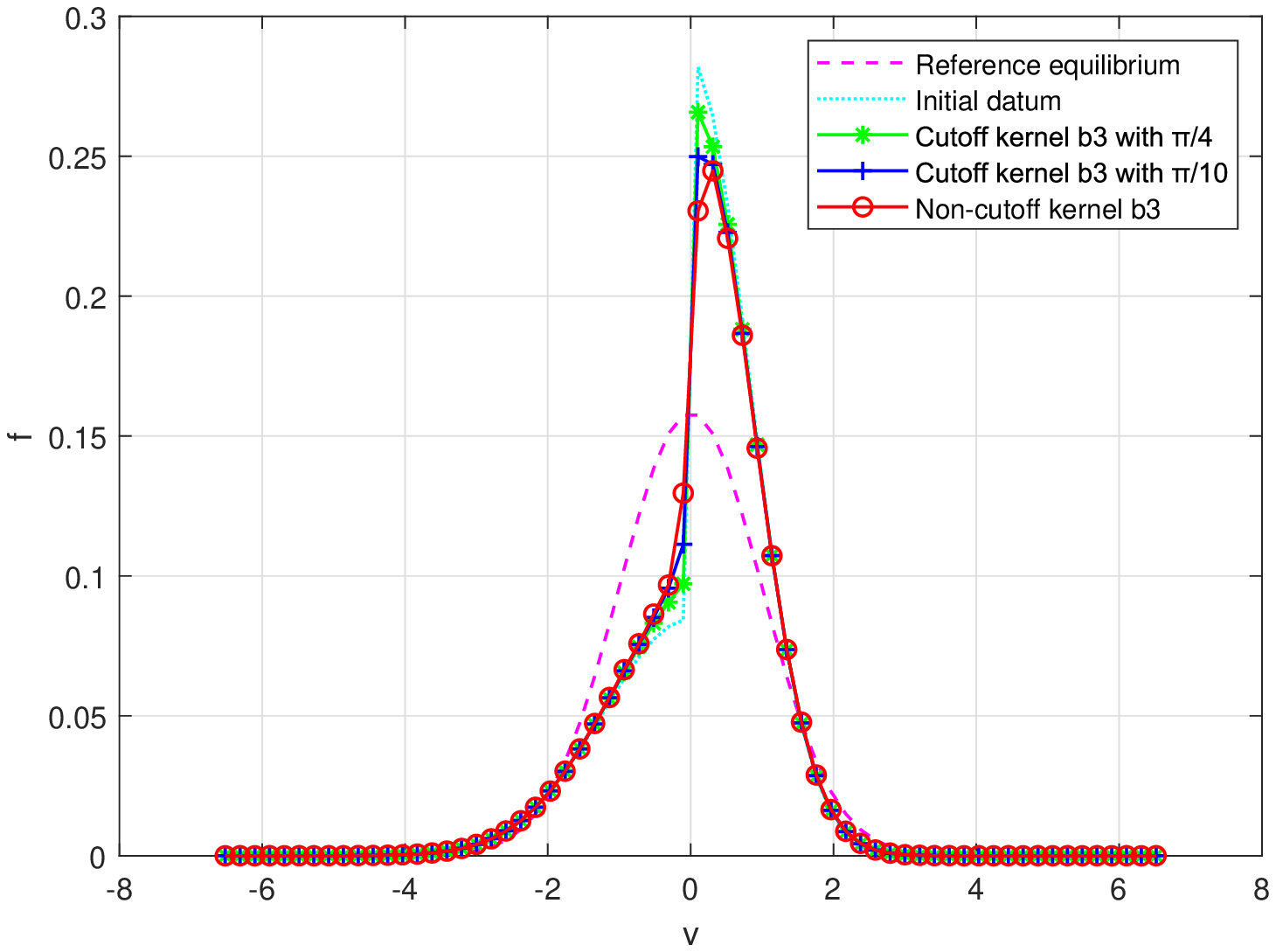}
	}
	\subfigure[t=1]{
		\includegraphics[width=6.5cm]{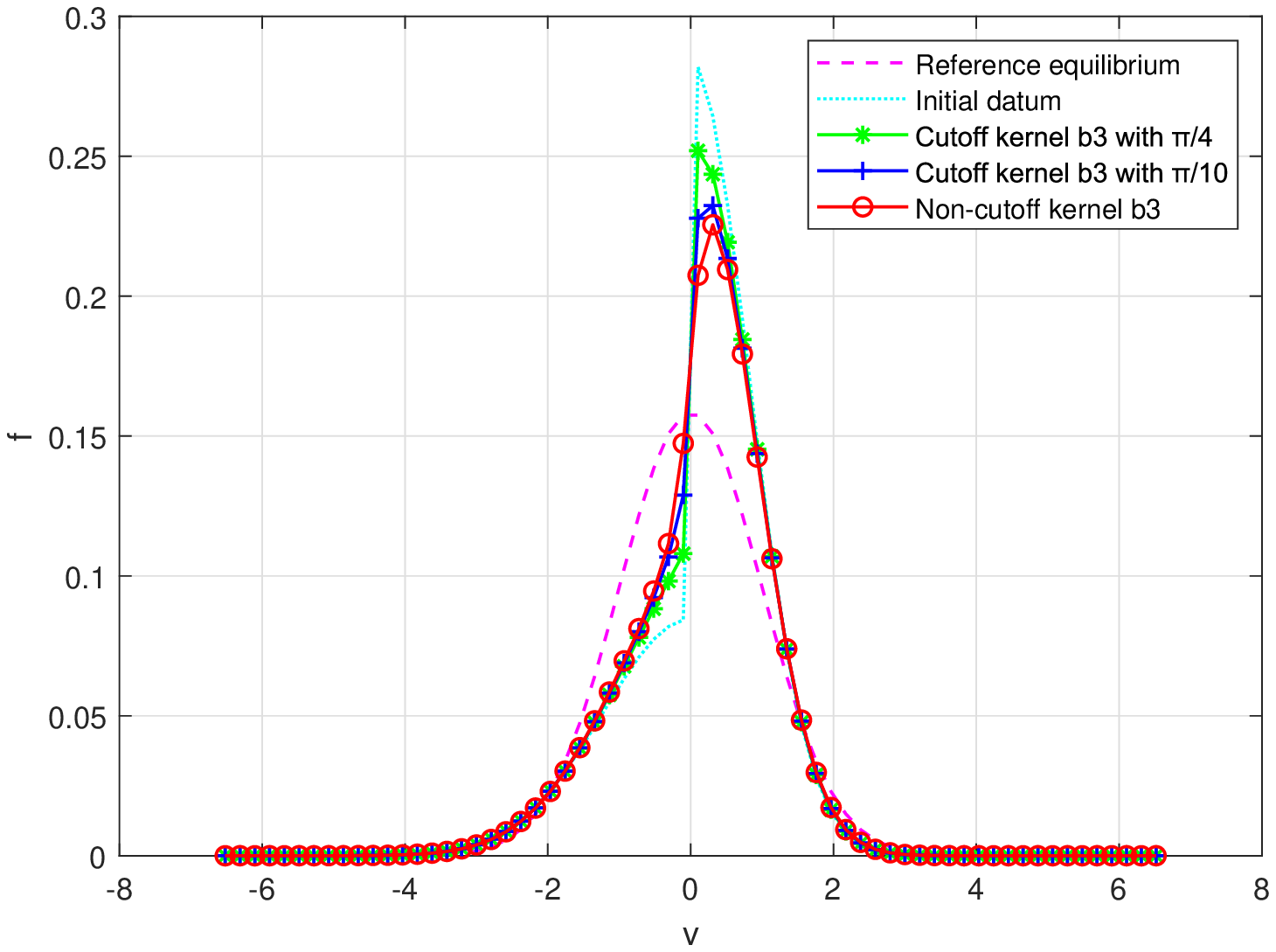}
	}\\
	\subfigure[t=1.5]{
		\includegraphics[width=6.5cm]{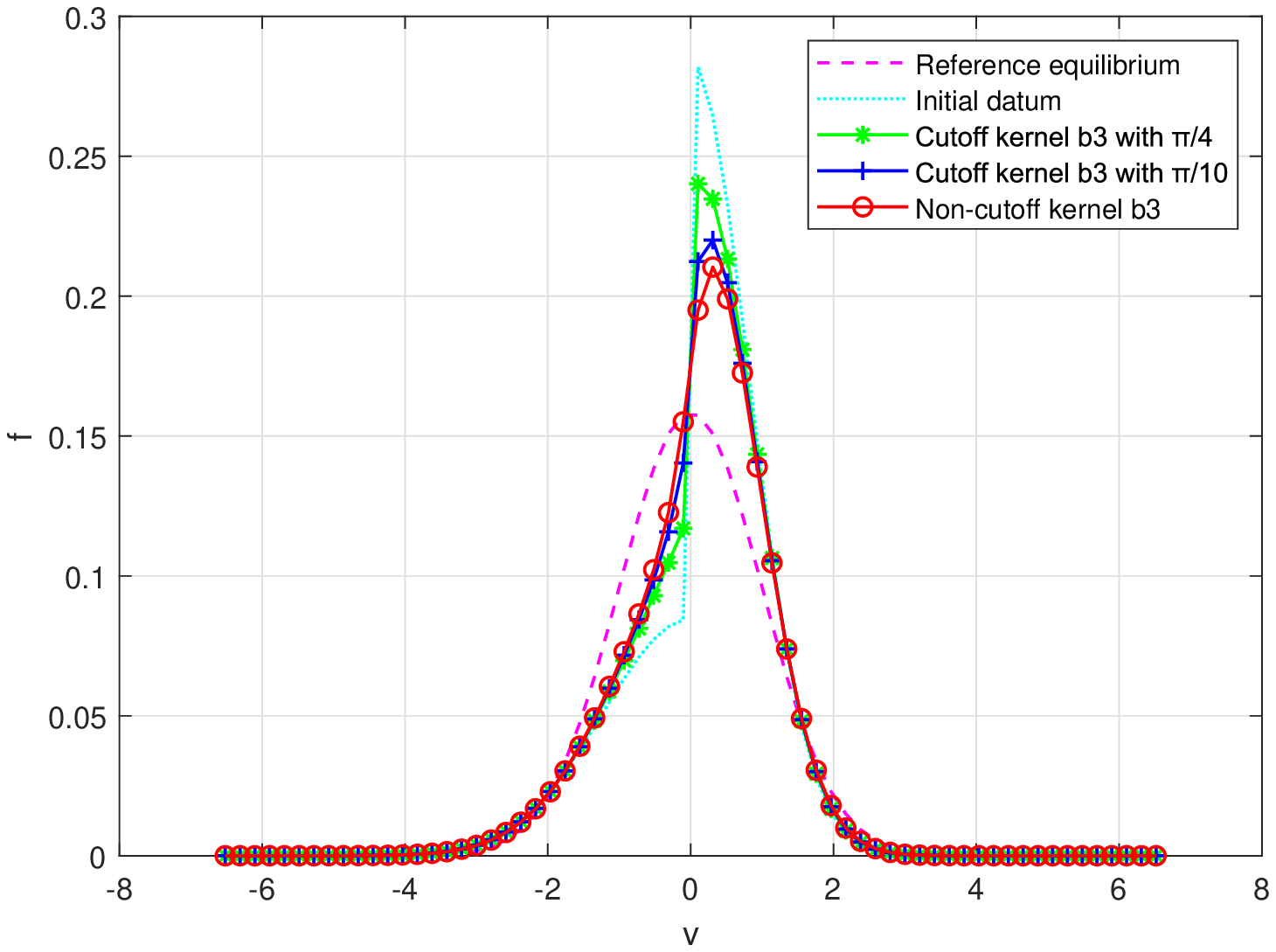}
	}
	\subfigure[t=3]{
		\includegraphics[width=6.5cm]{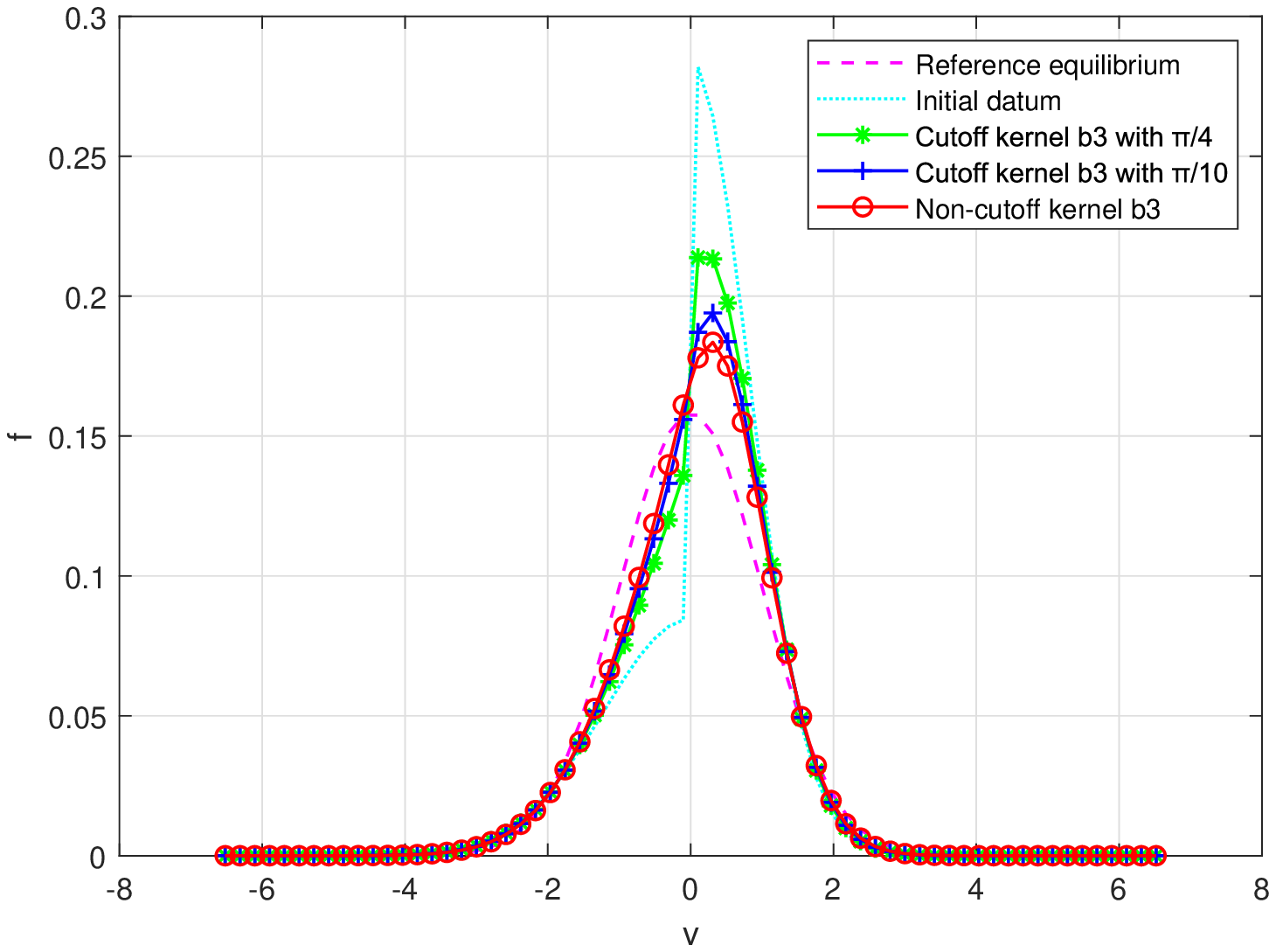}
	}\\
	\subfigure[t=5]{
		\includegraphics[width=6.5cm]{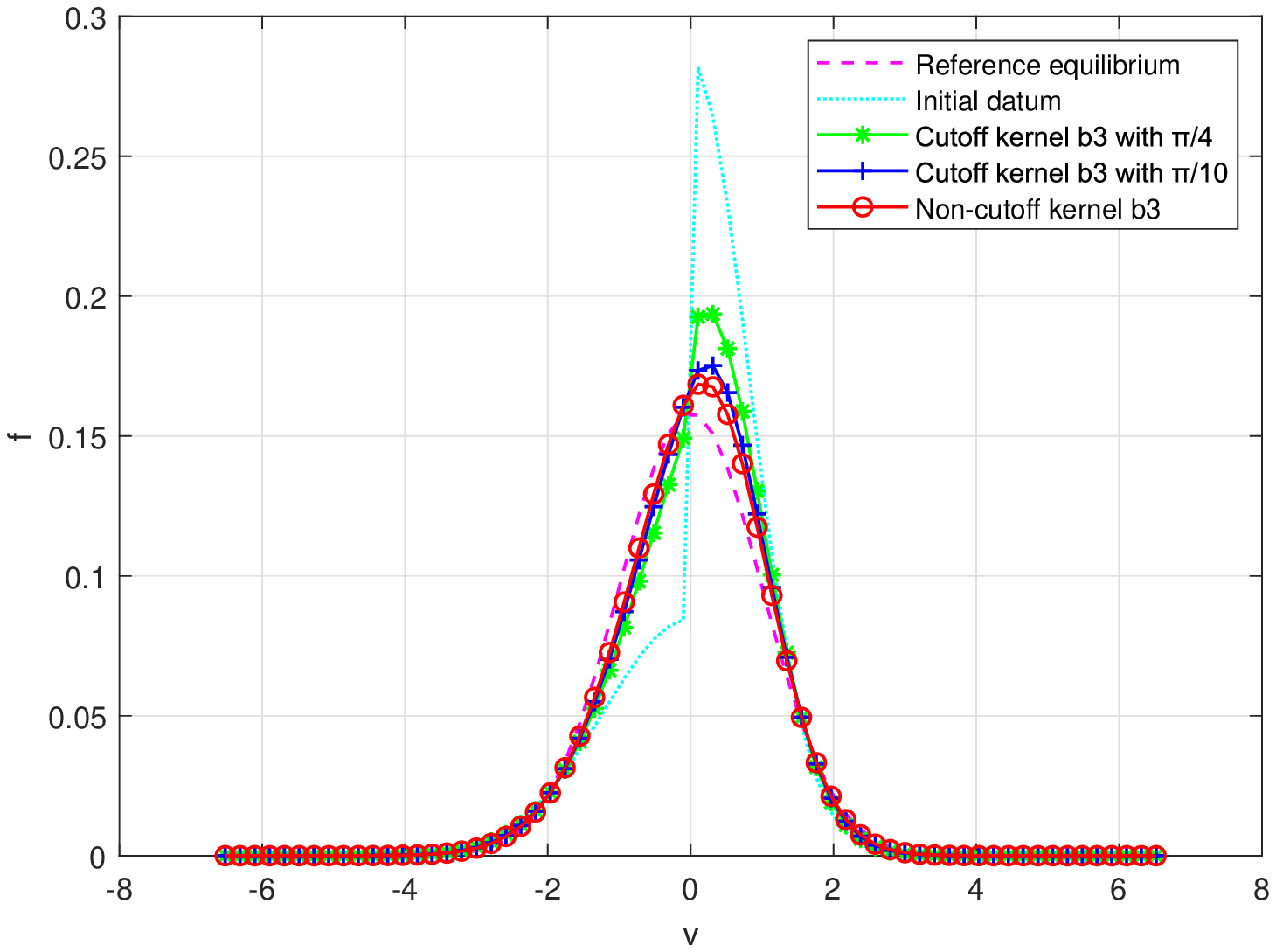}
	}
	\subfigure[t=9]{
		\includegraphics[width=6.5cm]{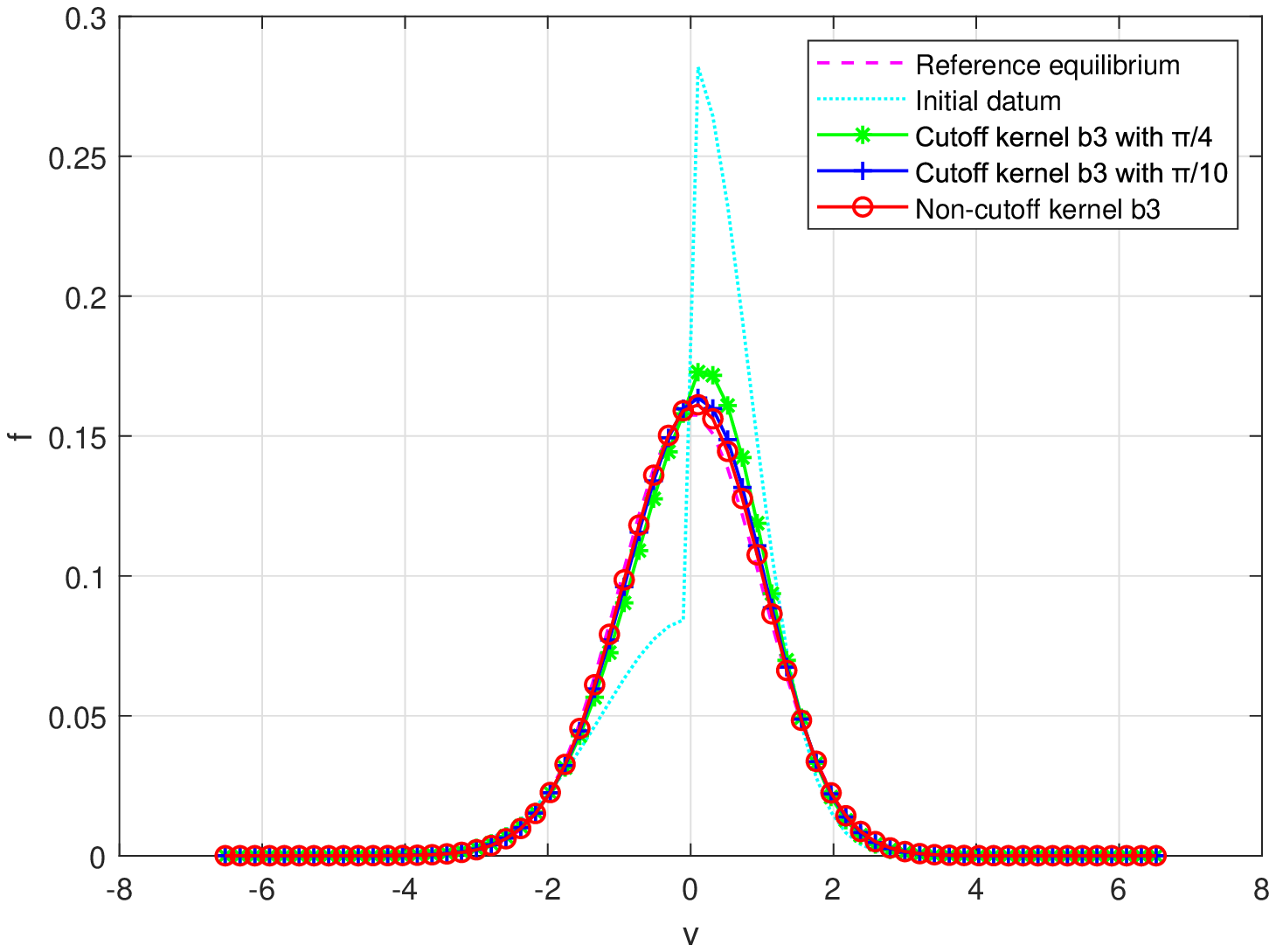}
	}\\
	\subfigure[t=15]{
		\includegraphics[width=6.5cm]{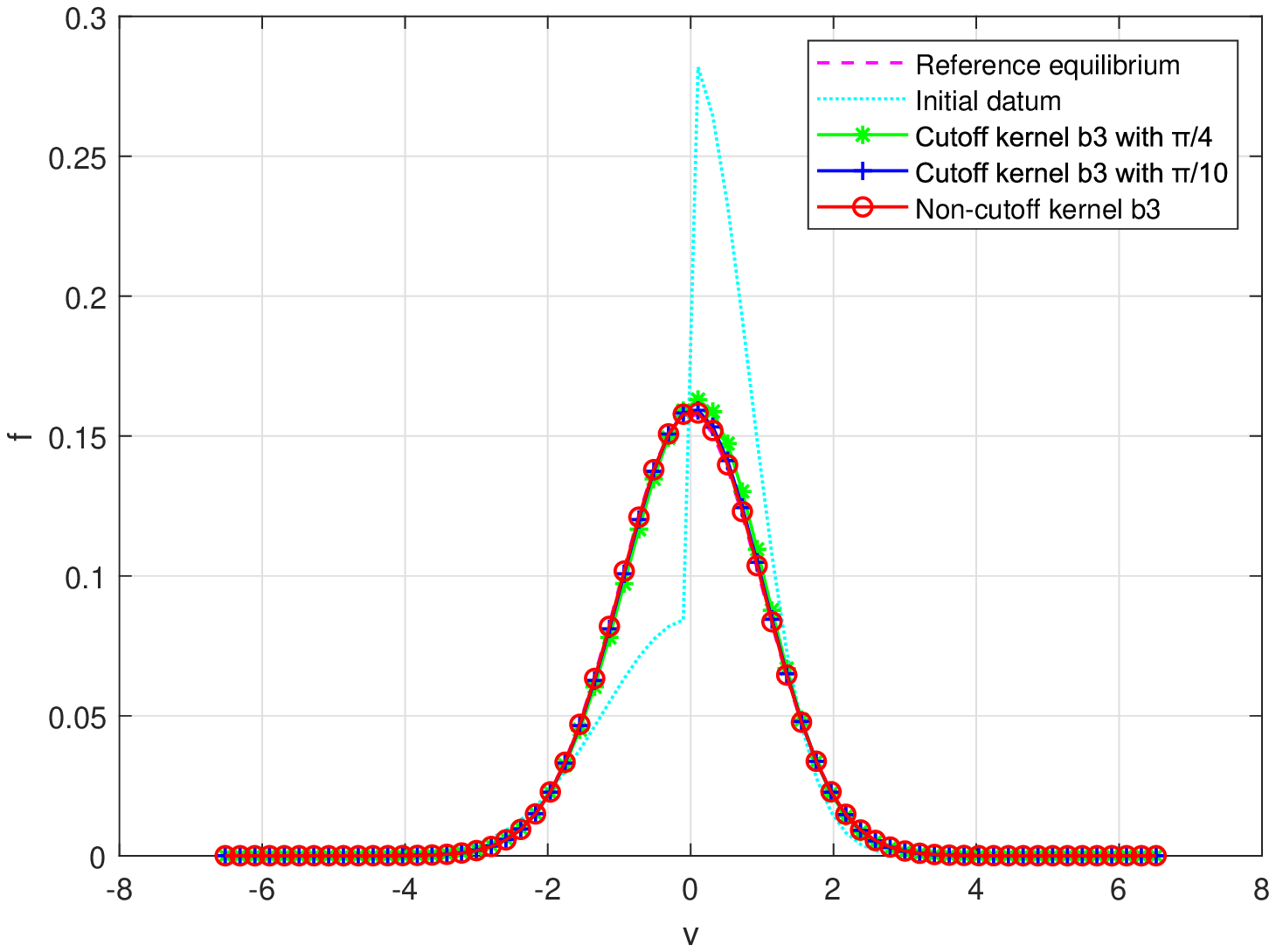}
	}
	\subfigure[t=22]{
		\includegraphics[width=6.5cm]{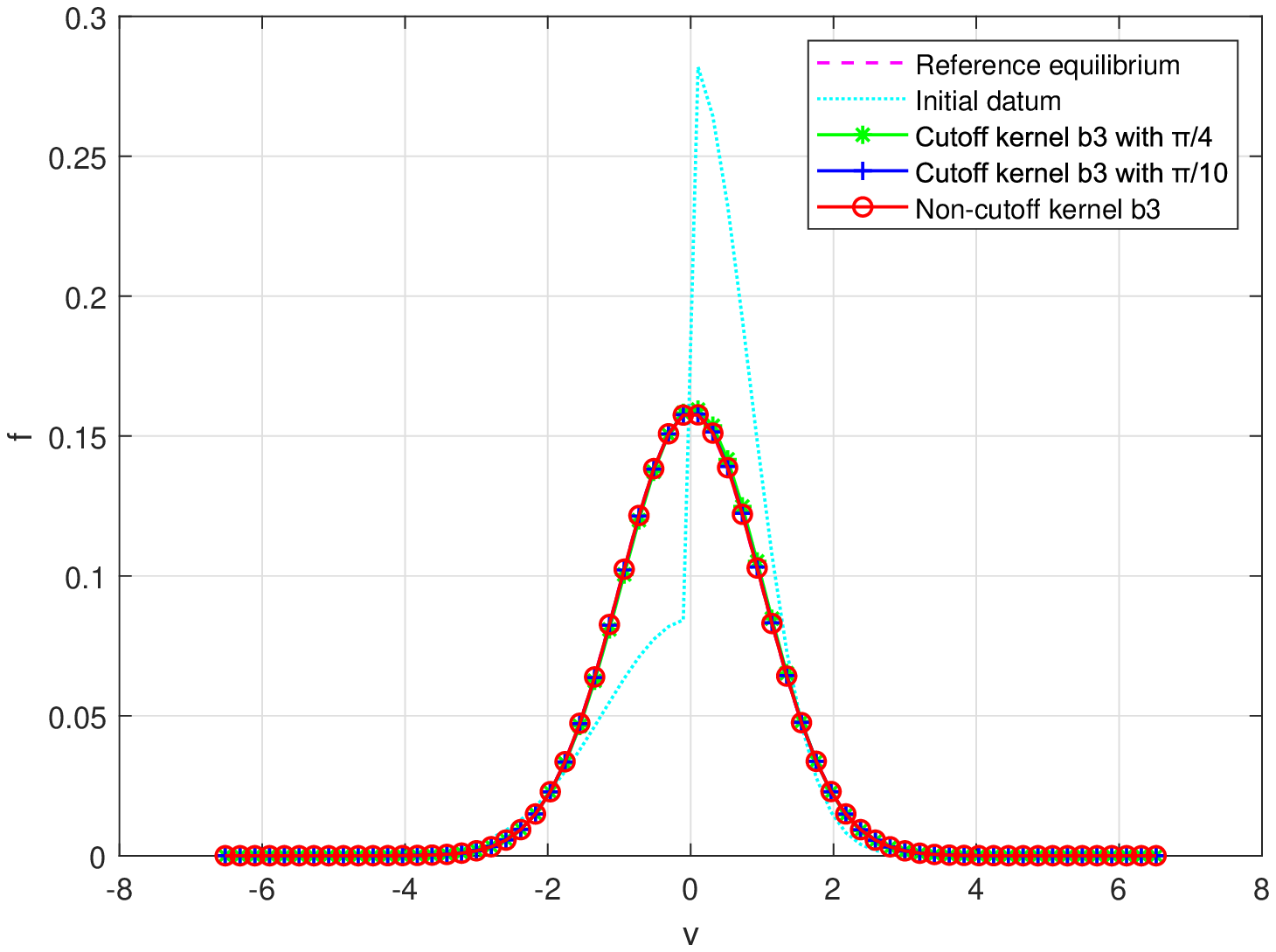}
	}
	
	\vspace{-0.2in}
	\caption{Section~\ref{2dis}: Discontinuous solution in 2D -- Maxwell molecule. Time evolution of the distribution function $f$  (a slice of the solution along $v_1$ with $v_2=0$) computed with cutoff kernel (\ref{b3cutoff}) (with $ \theta_{0} = \pi/4 $, $ \theta_{0}=\pi/10$ respectively), and its non-cutoff version (with $\theta_0=0$). Initial condition given by (\ref{disinitial}). Classical RK4 with $ \Delta t = 0.02 $ for time discretization. $N=N_{|q|}=64 $, $N_{\q}=32$. $R=6$, $L=(3+\sqrt{2})R/4\approx6.62$.}
	\label{2DDis}
\end{figure}


\section{Conclusion}
\label{sec:conc}

We have introduced a fast Fourier spectral method for the spatially homogeneous Boltzmann equation with non-cutoff collision kernels. These kernels arise in a large range of interaction potentials but are often cut off in numerical simulations for simplicity. This, as a result, changes the qualitative behavior of the solutions: the non-cutoff Boltzmann collision operator behaves like a fractional Laplacian, hence regularizes the solution immediately, whereas the solution in the cutoff case does not enjoy any smoothing property. We demonstrated that the Fourier spectral method is a well-defined framework to solve the non-cutoff Boltzmann equation and established the consistency and spectral accuracy of the method. Furthermore, the fast algorithms proposed previously for the cutoff Boltzmann equation \cite{GHHH17, HM19} can be readily generalized to the non-cutoff case, resulting in a method of the same numerical complexity. Through a series of examples, we have validated the accuracy and efficiency of the method, as well as verified the regularizing effect of the equation. The proposed method can be used as a black box solver to simulate the spatially nonhomogeneous Boltzmann equation, where many interesting problems remain open.

\section*{Acknowledgements}
We thank Prof.~Tong Yang for helpful discussion on the theory of the non-cutoff Boltzmann equation.

\bibliographystyle{plain}
\bibliography{hu_bibtex}
\end{document}